\documentclass[aps,pra,groupedaddress,tightenlines,notitlepage,nofootinbib,onecolumn,longbibliography, superscriptaddress,11pt]{revtex4-2}
\usepackage{amsmath,amssymb,graphicx,amsmath,amssymb,amsthm,physics,bm,bbm,color,xcolor,mathtools,anyfontsize,booktabs,multirow}%
\usepackage[sans]{dsfont}

\usepackage{tikz}
\usetikzlibrary{arrows.meta}
\usepackage{amsthm}
\usepackage[colorlinks=true, citecolor=magenta, linkcolor=blue,            urlcolor=blue, linktocpage=true,hyperfootnotes=false]{hyperref}

\newcommand\numeq[1]%
  {\stackrel{\scriptscriptstyle(\mkern-1.5mu#1\mkern-1.5mu)}{=}}

\newcommand\numleq[1]%
  {\stackrel{\scriptscriptstyle(\mkern-1.5mu#1\mkern-1.5mu)}{\le}}

\usepackage{newpxtext,newpxmath,array}
\def\tr{\operatorname{tr}}
\def\id{\operatorname{id}}

\newcommand{\dket}[1]{\vert #1\rangle\!\rangle}
\newcommand{\dbra}[1]{\langle\!\langle#1\vert}

\newcommand{\n}{\mathcal{N}}
\newcommand{\m}{\mathcal{M}}
\newcommand{\U}{\mathcal{U}}
\newcommand{\Q}{\mathcal{Q}}
\newcommand{\T}{\mathcal{T}}
\newcommand{\R}{\mathcal{R}^{\mathbbm{1}}}
\newcommand{\rp}{\mathcal{R}^{\pi}}
\newcommand{\NS}{{\rm S}_{A'\not\to B}}
\newcommand{\V}{\mathcal{V}}
\newcommand{\bS}{\mathbb{S}}
\newcommand{\bD}{\mathbb{D}}
\newcommand{\St}{\mathrm{St}}
\newcommand{\Ch}{\mathrm{Ch}}
\newcommand{\SCh}{\mathrm{SCh}}

\usepackage{cleveref}

\usepackage{accents}

\def\tr{\operatorname{tr}}
\def\supp{\operatorname{supp}}
\def\St{\operatorname{St}}
\def\Ch{\operatorname{Ch}}
\def\SCh{\operatorname{SCh}}

\def\id{\operatorname{id}}

\newtheorem{theorem}{Theorem}
\newtheorem*{theorem*}{Theorem}
\newtheorem{proposition}{Proposition}%
\newtheorem*{proposition*}{Proposition}
\newtheorem{corollary}{Corollary}
\newtheorem{lemma}{Lemma}
\newtheorem*{lemma*}{Lemma}
\newtheorem*{example*}{Example}%
\newtheorem{definition}{Definition}%

\begin{document}
\title{Conditional channel entropy sets fundamental limits on thermodynamic quantum information processing}

\author{Himanshu Badhani}\email{himanshubadhani@gmail.com}
\affiliation{q4i, Centre for Quantum Science and Technology (CQST), Center for Security, Theory and Algorithmic Research (CSTAR), International Institute of Information Technology Hyderabad, Gachibowli 500032, Telangana, India}

\author{Siddhartha Das}
\email{das.seed@iiit.ac.in}
\affiliation{q4i, Centre for Quantum Science and Technology (CQST), Center for Security, Theory and Algorithmic Research (CSTAR), International Institute of Information Technology Hyderabad, Gachibowli 500032, Telangana, India}

\begin{abstract}
The thermodynamic resourcefulness of quantum channels primarily depends on their underlying causal structure and their ability to generate quantum correlations. We quantify this interplay within the resource theory of athermality for bipartite quantum channels in the presence of a side channel acting as memory, referred to as the resource theory of conditional athermality. For channels with trivial output Hamiltonians, we characterize the optimal one-shot rates for distilling the identity gate from a given channel, as well as the cost of simulating the channel using the identity gate, under conditional Gibbs-preserving superchannels. We show that these rates have a direct trade-off relation with the conditional channel entropies, attributing operational significance to signaling in quantum processes. Furthermore, we establish an asymptotic equipartition property for the conditional channel min-entropy for classes of channels that are either tele-covariant or no-signaling from the non-conditioning input to the conditioning output. As a consequence, we demonstrate asymptotic reversibility of the resource theory for these channels. The asymptotic conditional athermality capacity of a tele-covariant channel is half the superdense coding capacity of its Choi state. Our work establishes the conditional channel entropy as a primitive information-theoretic concept for quantum processes, elucidating its potential for wider applications in quantum information science.
\end{abstract}

\maketitle

\section{Introduction}
\subsection{Motivation and Background}
Quantum processes, also called quantum channels, are able to generate correlations such as entanglement in quantum systems that have no classical counterpart~\cite{EPR35}. They encompass quantum states, measurements, and their transformations, thereby providing a general description of physical systems, their dynamics, and measurement processes in quantum theory. All quantum information processing and computation can be understood as transformations of states and gates, where these transformations correspond to compositions of quantum channels, including state preparation and measurement. Motivated by these basic and practical considerations, it is pertinent to understand how thermodynamic constraints limit the synthesis of bipartite quantum channels and the extraction of unitary gates locally from such channels. 

To address this, we pose the following two related questions for an arbitrary bipartite quantum channel $\n_{A'B'\to AB}$:
\begin{itemize}
\item \textit{To synthesize a single copy of $\n$ from a decoupled bipartite channel $\id \otimes \Q$ via conditional Gibbs-preserving superchannels, what is the minimum dimension of the identity (or any unitary) gate $\id$ required, for some side channel $\Q$?}
\item \textit{To extract a decoupled channel $\id \otimes \Q$ from a single use of $\n$ via conditional Gibbs-preserving superchannels, what is the maximum possible dimension of $\id$, for some side channel $\Q$?}
\end{itemize}
The restriction of allowed physical transformations for synthesis or extraction to conditional Gibbs-preserving superchannels implicitly encodes thermodynamic constraints, as will become clear in later discussions.

The answers to these questions provide fundamental limitations on the formation and distillation of bipartite quantum channels in the presence of a quantum side channel. We show that the one-shot and asymptotic optimal rates for an arbitrary channel under the action of conditional Gibbs-preserving superchannels are characterized by channel divergences between the given channel and a decoupled channel of the form $\T^\beta\otimes\Q$, where $\T^\beta$ is the absolutely thermal channel that always outputs the thermal state with inverse temperature $\beta$. The conditional Gibbs-preserving superchannels $\Omega$ preserve conditional thermality, $\Omega(\mathcal{T}^\beta\otimes\Q)=\mathcal{T}^\beta\otimes\Q'$ for some channel $\Q'$. The absolutely thermal channel has maximum channel entropy and minimum channel free energy~\cite{BGD25,BGD25b}. When the channels have trivial output Hamiltonians, the one-shot and asymptotic optimal rates are related to the conditional channel entropies introduced in~\cite{DGP24} (see also~\cite{CE16}). This provides a clear operational interpretation of conditional channel entropies. 

The transformation of a bipartite channel $\n_{A'B'\to AB}$ to $\id\otimes\Q$ incurs erasure of the intricate causal structure of $\n$. $\id\otimes\Q$ is decoupled and has no correlation generating ability, while $\n$ can be rich in correlation generating abilities. If $\n$ has any signaling or causal influence from $A'\to B$, $B'\to A$, or both then that needs to be erased. Analogously, the transformation of $\id\otimes\Q$ to $\n$ requires synthesis or preparation of intricate causal structure from a decoupled one. Our results highlight the interplay between the causal structure of a bipartite quantum channel~\cite{PHHH06,Per21,Sim26}, its ability to manipulate bipartite correlations~\cite{DFP11,Das19,BDWW19,KHD22,RMSZ25}, the cost of its synthesis, and its usefulness for extracting local unitary operations. We make several contributions toward determining fundamental limits on the conversion of bipartite channels under conditional Gibbs-preserving superchannels. In particular, the properties of conditional channel entropies capture how causal structure influences both the distillation yield and the formation cost of bipartite quantum channels. We present our main results in Section \ref{sec:result} and discuss their relations to prior works in Section \ref{sec:prior}; see preliminaries in Section \ref{section:prelim}. 

\subsection{Main Results}\label{sec:result}
 \begin{table}[ht!]
\centering
\small 
\renewcommand{\arraystretch}{2.0}
\begin{tabular}{|p{3cm}|p{12cm}|}
\hline
\multicolumn{2}{|c|}{\textbf{Conditional min-entropy}: {quantum \textcolor{red}{states} $\rho_{AB}$ vs.~\textcolor{blue}{channels} $\mathcal{N}_{A'B'\to AB}$}} \\ \hline

Definition & 
\begin{tabular}{@{}p{12cm}@{}} 
\textcolor{red}{$S_{\infty}(A|B)_{\rho}:=-\inf_{\sigma\in\mathrm{St}(B)}D_{\infty}(\rho_{AB}\Vert\mathbbm{1}_A\otimes\sigma_B)$}~\cite{Ren05,KRS09}\\ \hline 
\textcolor{violet}{$S_{\infty}[A|B]_{\mathcal{N}}:= -\inf_{\mathcal{Q}\in\mathrm{Ch}(B',B)}D_{\infty}[\mathcal{N}\Vert\mathcal{R}^\mathbbm{1}_{A'\to A}\otimes\mathcal{Q}_{B'\to B}]$~\cite{DGP24}} 
\end{tabular} \\ \hline

Dimensional bounds & 
\begin{tabular}{@{}p{12cm}@{}} 
\textcolor{red}{$-\log \min\{|A|,|B|\}\leq S_{\infty}(A|B)_{\rho}\leq \log|A|$}~\cite{Ren05} \\ \hline 
\textcolor{blue}{$-\log \min\{ |A'||B'||B|,|A'|^2|A|\} \leq S_{\infty}[A|B]_{\mathcal{N}}\leq \log|A|$} 
\end{tabular} \\ \hline

Minimal value & 
\begin{tabular}{@{}p{12cm}@{}} 
\textcolor{red}{$S_{\infty}(A|B)_{\rho}=-\log \min\{|A|,|B|\}$ iff $\rho_{AB}$ is a maximally entangled state.}~\cite{KRS09} \\ \hline 
\textcolor{blue}{$S_{\infty}[A|B]_{\mathcal{N}}=-\log \min\{ |A'||B'||B|,|A'|^2|A|\}$ iff $\mathcal{N}_{A'B'\to AB}$ is a maximally entangling operation.} 
\end{tabular} \\ \hline

Negative values & 
\begin{tabular}{@{}p{12cm}@{}} 
\textcolor{red}{$S_{\infty}(A|B) < 0$ implies $\rho_{AB}$ is (NPT) entangled.}~\cite{T21} \\ \hline 
\textcolor{blue}{$S_{\infty}[A|B]_{\mathcal{N}} < -\log |A|$ implies $\mathcal{N}_{A'B'\to AB}$ is signaling from $A'\to B$. 
\newline
$S_{\infty}[A|B]_{\mathcal{N}} < -\log |A'|$ implies $\mathcal{N}_{A'B'\to AB}$ is (NPT) entangling.} 
\end{tabular} \\ \hline

Nonnegative values & 
\begin{tabular}{@{}p{12cm}@{}} 
\textcolor{red}{For all PPT states $\rho_{AB}$, $S_{\infty}(A|B) \geq 0$.}~\cite{T21} \\ \hline 
\textcolor{blue}{For all completely PPT-preserving channels $\mathcal{N}_{A'B'\to AB}$, $S_{\infty}[A|B]_{\mathcal{N}}\geq -\log |A'|$.} 
\end{tabular} \\ \hline

Maximum value & 
\begin{tabular}{@{}p{12cm}@{}} 
\textcolor{red}{$S_{\infty}(A|B)=\log|A|$ iff $\rho_{AB}=\pi_A\otimes\omega_B$ for any $\omega\in\mathrm{St}(B)$.}~\cite{T21} \\ \hline 
\textcolor{blue}{$S_{\infty}[A|B]_{\mathcal{N}}=\log |A|$ iff $\mathcal{N}_{A'B'\to AB}=\mathcal{R}^{\pi}_{A'\to A}\otimes\mathcal{Q}_{B'\to B}$ for any $\mathcal{Q}\in\mathrm{Ch}(B',B)$.} 
\end{tabular} \\ \hline

Asymptotic equipartition property & 
\begin{tabular}{@{}p{12cm}@{}} 
\textcolor{red}{$\lim_{\varepsilon\to 0^+}\lim_{n\to \infty}\frac{1}{n}S^\varepsilon_{\infty} (A^n|B^n)_{\rho^{\otimes n}}=S(A|B)_{\rho}$.}~\cite{TCR09} \\ \hline 
\textcolor{blue}{In general, $\lim_{\varepsilon\to 0^+}\lim_{n\to \infty}\frac{1}{n}S^\varepsilon_{\infty} [A^n|B^n]_{\mathcal{N}^{\otimes \n}}\leq \inf_{\psi\in\St(RA'B')}S(A|RB)_{\n(\psi)}$ and $ \lim_{\varepsilon\to 0^+}\lim_{n\to \infty}\frac{1}{n}S^\varepsilon_{\infty} [A^n|B^n]_{\mathcal{N}^{\otimes \n}}\leq S(R_AA|R_BB)_{\Phi^\n}-\log|A'|$. If $\n$ is no-signaling from $A'\to B$ or tele-covariant, then $\lim_{\varepsilon\to 0^+}\lim_{n\to \infty}\frac{1}{n}S^\varepsilon_{\infty} [A^n|B^n]_{\mathcal{N}^{\otimes n}}=S[A|B]_\n$.} 
\end{tabular} \\ \hline
\end{tabular}
\vspace{10pt} 
\caption{We compare features between the conditional min-entropies of an arbitrary bipartite quantum state $\rho_{AB}$ and an arbitrary bipartite quantum channel $\mathcal{N}_{A'B'\to AB}$. $S[A|B]_{\n}$ is the von Neumann conditional entropy of $\n$~\cite{DGP24} and $S(A|B)_{\rho}$ is the von Neumann conditional entropy of $\rho$~\cite{CA97,HOW06}.  $\R(\cdot):=\tr(\cdot)\mathbbm{1}$ and $\Phi^{\n}_{R_AAR_BB}:=\n(\Phi_{R_AA'}\otimes\Phi_{R_BB'})$ is the Choi state of $\n$, where $\Phi_{R_AA'}$ denotes a maximally entangled state. For a quantum channel $\n_{A'B'\to AB}$, $S[A|B]_{\n}=S(R_AA|R_BB)_{\Phi^\n}-\log|A'|$ if $\n$ is tele-covariant and $S[A|B]_{\n}=\inf_{\psi\in\St(RA'B')}S(A|RB)_{\n(\psi)}$ if $\n$ is no-signaling from $A'\to B$~\cite{DGP24}.}\label{tab:characteristics}
\end{table}

We analyze quantum information processing under thermodynamic constraints. We consider a canonical and physically motivated scenario where the interaction between two systems and transformations are in general described by bipartite quantum channels. A natural, freely occurring process is depicted by a bipartite quantum channel called conditional absolutely thermal channel that is tensor-product of quantum channels $\T^\beta\otimes\Q$: the first is a absolutely thermalizing process $\mathcal{T}^{\beta}$ that drives all inputs to equilibrium (thermal or Gibbs state $\gamma^\beta$~\cite{Len78}) with a bath at inverse temperature $\beta$ and the side channel $\Q$, by contrast, represents non-equilibrating dynamics. We focus on two related aspects of quantum processes: a) resource theory of the conditional athermality of quantum processes, b) the properties of conditional channel min-entropy and its interconnectedness with the underlying causal structure of the given bipartite channel (\autoref{tab:characteristics}). The two aspects are related as the optimal one-shot conditional athermality distillation yield and formation cost of a bipartite quantum channel depend on its conditional channel entropies.

 We formalize the framework of the resource theory of conditional athermality for bipartite quantum channels. We consider the conditional absolutely thermal channels $\mathcal{T}^\beta\otimes\Q$ as free objects, where $\Q_{B'\to B}$ is any quantum channel and $\mathcal{T}^\beta_{A'\to A}(\cdot)=\tr(\cdot)\gamma^\beta$ for a given thermal state $\gamma^\beta_A$. To show the conditional athermal resourcefulness of a bipartite quantum channel $\n_{A'B'\to AB}$ relative to free objects, we write $(\n,\mathcal{T}^\beta_{A'\to A}\otimes\Q_{B'\to B})$ for some side channel $\Q$. The free operations are conditional Gibbs-preserving superchannels (CGPSs) $\Theta^\beta$ that preserve conditionally absolute thermalization, meaning that $\Theta^\beta(\T^\beta_{A'\to A}\otimes\Q_{B'\to B})=\T^\beta_{C'\to C}\otimes\Q'_{D'\to D}$ for some side channels $\Q,\Q'$. The standard resource units are taken to be the conditional identity channels $(\id\otimes\Q,\mathcal{R}^{\pi}\otimes\Q)$ and they are resource equivalent to conditional unitary channels $(\mathcal{U}\otimes\Q,\mathcal{R}^{\pi}\otimes\Q)$, where $\gamma^\beta_A=\pi_A:=\frac{1}{|A|}\mathbbm{1}_A$ for trivial Hamiltonian $\widehat{H}_A\propto\mathbbm{1}_A$ and $\mathcal{R}^{\pi}(\cdot):=\tr(\cdot)\pi$. 
 
 We consider resource interconversion under CGPSs $\Theta^{\beta}$ up to an error $\varepsilon\in[0,1]$ and determine the exact one-shot rates for both the resource distillation yield ${\rm Dist^{\varepsilon}(\n,\T^\beta)}$ and the formation cost ${\rm Cost^{\varepsilon}(\n,\T^\beta)}$ of an arbitrary channel $\n_{A'B'\to AB}$, $A'\simeq A$, (\autoref{th:dist-for-cost}):
\begin{align}
\mathrm{Dist}^{(1,\varepsilon)}(\mathcal{N},\T^\beta)
&=\frac{1}{2}\inf_{\Q\in\Ch({B',B})}D_{\mathrm{H}}^{\varepsilon^2}[\n_{A'B'\to AB}\Vert\T^\beta_{A'\to A}\otimes \Q_{B'\to B}],\\
\mathrm{Cost}^{(1,\varepsilon)}(\mathcal{N},\T^\beta)
&=\frac{1}{2}\inf_{\mathcal{Q}\in\Ch({B',B})}D_{\infty}^{\varepsilon}[\n_{A'B'\to AB}\Vert\T^\beta_{A'\to A}\otimes\mathcal{\mathcal{Q}}_{B'\to B}],
\end{align}
where $D^{\varepsilon}_H[\cdot\|\cdot]$ and $D^\varepsilon_{\infty}[\cdot\|\cdot]$ are hypothesis-testing and max channel relative entropies respectively~\cite{CMW16}. $\mathrm{Dist}^{(1,\varepsilon)}(\mathcal{N},\T^\beta)$ is the maximum number of conditional identity $\id_2$ channels that can be distilled from a single-use of a bipartite resource channel $(\mathcal{N},\T^\beta\otimes\Q)$ under CGPSs for the worst-possible side channel $\Q$. $\mathrm{Cost}^{(1,\varepsilon)}(\mathcal{N},\T^\beta)$ is the minimum number of conditional qubit identity $\id_2$ channels that needs to be diluted to prepare a single copy of a bipartite resource channel $(\mathcal{N},\T^\beta\otimes\Q)$ under CGPSs for the best-possible side channel $\Q$. 

The conditional athermality distillation and formation capacities of a quantum channel $\n$ are the optimal asymptotic rates of the distillation yield and formation cost, respectively, as the error vanishes over asymptotically many (parallel) uses of the channel $\n$:
\begin{align}
    {\rm Dist}^{(\infty,0)}(\n,\T^\beta) &:=\lim_{\varepsilon\to 0^+}\lim_{n\to\infty}\frac{1}{n}{\rm Dist}^{(1,\varepsilon)}(\n^{\otimes n},{\T^{\beta}}^{\otimes n}),\\
      {\rm Cost}^{(\infty,0)}(\n,\T^\beta) &:=\lim_{\varepsilon\to 0^+}\lim_{n\to\infty}\frac{1}{n}{\rm Cost}^{(1,\varepsilon)}(\n^{\otimes n},{\T^{\beta}}^{\otimes n}).
\end{align}
A direct consequence of the resource-theoretic approach is the fact that the resource distillation capacity of a channel is never greater than its formation capacity,
\begin{equation}
      {\rm Dist}^{(\infty,0)}(\n,\T^\beta)\leq   {\rm Cost}^{(\infty,0)}(\n,\T^\beta),
\end{equation}
and it follows from \autoref{cor:asym-dist-lb} that the resource distillation capacity is upper bounded as
\begin{equation}
    {\rm Dist}^{(\infty,0)}(\n,\T^\beta)\leq \frac{1}{2}\inf_{\Q\in\Ch(B',B)}D^{\rm reg}[\n_{A'B'\to AB}\|\T^\beta_{A'\to A}\otimes\Q_{B'\to B}],
\end{equation}
where $D^{\rm reg}[\n\|\m]:=\lim_{n\to\infty}\frac{1}{n}D[\n^{\otimes n}\|\m^{\otimes n}]$ is the regularized relative entropy between channels $\n,\m$~\cite{CMW16,FFRS20}.

We observe that for the trivial Hamiltonian, the resource theory of conditional athermality reduces to the resource theory of conditional purity as $\gamma^\beta_A=\pi_A$ and $\T^\beta_{A'\to A}=\rp_{A'\to A}$ when $\widehat{H}_A\propto \mathbbm{1}_A$. For a quantum channel $\n_{A'B'\to AB}$ with $\widehat{H}_A\propto \mathbbm{1}_A$, its one-shot resource distillation yield and formation cost depend on its hypothesis-testing conditional entropy $S^\varepsilon_H[A|B]_{\n}$ and $\varepsilon$-smoothed conditional min-entropy $S^{\varepsilon}_{\infty}[A|B]_{\n}$ respectively (\autoref{cor:one-shot-con-purity-rates}),
    \begin{align}
\mathrm{Dist}^{(1,\varepsilon)}(\mathcal{N},\mathcal{R}^\pi)
&=\frac{1}{2}\left(\log |A|-S^{\varepsilon^2}_H[A|B]_{\n}\right),\\
\mathrm{Cost}^{(1,\varepsilon)}(\mathcal{N},\mathcal{R}^\pi)
&=\frac{1}{2}\left(\log|A|-S^\varepsilon_{\infty}[A|B]_{\n}\right).
\end{align}
The resource distillation and formation capacities of a quantum channel $\n_{A'B'\to AB}$, $A'\simeq A$, are lower bounded as (\autoref{prop:asy-dist-lb} and \autoref{prop:asy-form-lb})
\begin{align}
  \frac{1}{2}  \min\left\{\log|A|-S^{\not\to}[A|B]_{\n},\log|A'||A|-S(R_AA|R_BB)_{\Phi^{\n}}\right\}&\leq {\rm Dist}^{(\infty,0)}(\n,\rp)\\
    &\leq {\rm Cost}^{(\infty,0)}(\n,\rp),
\end{align}
where $S^{\not\to}[A|B]_{\n}:=\inf_{\psi\in\St(RA'B')}S(A|RB)_{\n(\psi)}$ is the no-signaling min-entropy~\cite{DGP24} and $\Phi^{\n}_{R_AAR_BB}:=\n(\Phi_{R_AA}\otimes\Phi_{R_BB})$ is the Choi state of $\n_{A'B'\to AB}$.

We prove that the asymptotic equipartition property holds for the conditional min-entropy of a quantum channel $\n_{A'B'\to AB}$ if it is either tele-covariant~\cite{GC99,HHH99,DBW20,DBWH21} (\autoref{th:equipartition}) or no-signaling from $A'\to B$~\cite{BGNP01,PHHH06} (\autoref{thm:ns-asym-con-ent}),
\begin{align}
    \lim_{\varepsilon\to 0^+}\lim_{n\to \infty}\frac{1}{n}S^\varepsilon_{\infty}[A^n|B^n]_{\n^{\otimes n}}=S[A|B]_{\n},
\end{align}
where $S[A|B]_{\n}:=-\inf_{\Q\in\Ch(B',B)}D[\n\|\R\otimes\Q]$ is the von Neumann conditional channel entropy~\cite{DGP24}. The asymptotic equipartition property along with \autoref{th:dist-for-cost} allows us to exactly determine the conditional purity distillation and formation capacities of a quantum channel $\n_{A'B'\to AB}$ obeying certain symmetries. If a channel $\n_{A'B'\to AB}$ is tele-covariant (\autoref{thm:cap-tel}) or no-signaling channel from $A'\to B$ (\autoref{thm:asym-ns}), then
\begin{equation}\label{eq:cap-ns-tel}
    {\rm Dist}^{(\infty,0)}(\n,\rp)= {\rm Cost}^{(\infty,0)}(\n,\rp)=\frac{1}{2}(\log|A|-S[A|B]_{\n}),
\end{equation}
and equivalently, the following asymptotic equipartition property holds,
\begin{equation}
     \lim_{\varepsilon\to 0^+}\lim_{n\to \infty}\frac{1}{n}S^{\varepsilon^2}_H[A^n|B^n]_{\n^{\otimes n}}=S[A|B]_{\n}=\lim_{\varepsilon\to 0^+}\lim_{n\to \infty}\frac{1}{n}S^\varepsilon_{\infty}[A^n|B^n]_{\n^{\otimes n}}.
\end{equation}
Thus, the resource theory of conditional purity is asymptotically reversible for quantum channels $\n_{A'B'\to AB}$ that are tele-covariant or no-signaling from $A'\to B$. Furthermore, we find that the superdense coding capacity~\cite{ZB03,BDL+04} of the Choi state $\Phi^{\n}$ of a tele-covariant channel $\n_{A'B'\to AB}$, where $A'\simeq A$, is double the conditional purity capacity of the channel, see Eq.~\eqref{eq:sdc-dist-cost},
\begin{equation}
    \frac{1}{2}{\rm Sdc}(R_AA;R_BB)_{\Phi^{\n}}= {\rm Dist}^{(\infty,0)}(\n,\rp)= {\rm Cost}^{(\infty,0)}(\n,\rp).
\end{equation}
We note that for a bipartite quantum channel $\mathcal{R}_{A'B'\to AB}$ that is conditional uniformly mixing channel, i.e., $ \mathcal{R}=\mathcal{R}^{\pi}_{A'\to A}\otimes\Q_{B'\to B}$ for any side channel $\Q$, we have, ${S}^{\varepsilon=0}_{H}[A|B]_{\mathcal{R}}=S[A|B]_{\mathcal{R}}=S_{\infty}[A|B]_{\mathcal{R}}=\log|A|$; if $\Q$ is taken to be $\rp_{B'\to B}$, then $\mathcal{R}=\rp_{A'\to A}\otimes\rp_{B'\to B}$ is both tele-covariant and no-signaling from $A'\to B$, and a free object in the resource theory of athermality for channels with trivial output Hamiltonian. We can then interpret Eq.~\eqref{eq:cap-ns-tel} as
\begin{equation}
   {\rm Dist}^{(\infty,0)}(\n,\rp)= {\rm Cost}^{(\infty,0)}(\n,\rp)  =\frac{1}{2}(S[A|B]_{\rp\otimes\Q}-S[A|B]_{\n}),
\end{equation}
for any side channel $\Q$. That is, the resource distillation and formation capacities of a channel $\n_{A'B'\to AB}$ that is tele-covariant or no-signaling from $A'\to B$, is equal to half of the difference between the von Neumann conditional entropies of a conditional uniformly mixing channel and the channel $\n$ itself.

Consider two quantum channels $\n_{A'B'\to AB}$ and $\m_{C'D'\to CD}$, where $\widehat{H}_{A}\propto \mathbbm{1}_A$ and $\widehat{H}_{C}\propto\mathbbm{1}_C$, that are tele-covariant or no-signaling from $A'\to B$. The asymptotic optimal rate of the resource (conditional athermality or purity) gain $\Delta[\n\to\m]$ during transformation of $\n$ to $\m$ under the action of CGPSs is equal to,
\begin{equation}
    \Delta(\n\to\m,\rp_{A'\to A})={\rm Dist}^{(\infty,0)}(\n,\rp)-{\rm Cost}^{(\infty,0)}(\m,\rp)=\frac{1}{2}\left[\log\frac{|A|}{|C|}+S[C|D]_{\m}-S[A|B]_{\n}\right].
\end{equation}
For a quantum channel $\n_{A'B'C'\to ABC}$ that is tripartite tele-covariant~\cite{DBWH21} or no-signaling from $A'\to BC$,
\begin{align}
    \Delta(\n\to\tr_C\circ\n,\rp_{A'\to A})&={\rm Dist}^{(\infty,0)}(\n,\rp_{A'\to A})-{\rm Cost}^{(\infty,0)}(\tr_C\circ\n,\rp_{A'\to A})\\
    &=\frac{1}{2}\left(S[A|B]_{\n}-S[A|BC]_{\n}\right)\geq 0,\label{eq:s-subadd}
\end{align}
where $S[A|B]_{\n}:=\inf_{\Q\in\Ch(B'C',B)}D[\tr_C\circ\n\|\T^\beta\otimes\Q]$ for a tripartite channel $\n_{A'B'C'\to ABC}$. Inequality~\eqref{eq:s-subadd} follows because the channel entropy satisfies the strong subaddivity property~\cite{DGP24}: the conditional entropy $S[A|B]_{\n}$ of a quantum channel $\n_{A'B'C'\to ABC}$ is nonincreasing upon conditioning, $S[A|B]_{\n}\geq S[A|BC]_{\n}$.

\begin{figure}[h]
\centering
\begin{tikzpicture}[
    xscale=1.7, 
    axis/.style={->, thick},
    tick/.style={thick},
    guide/.style={dotted, thick},
    label/.style={font=\small},
    annot/.style={font=\small, red},
    >=Stealth
]

\draw[axis]  (-4,0) -- (4,0);

\foreach \x/\lab in {
    -4/{-3\log|A|},
    -2/{-2\log|A|},
     0/{-\log|A|},
     2/{0},
     4/{\log|A|}
}{
    \draw[tick] (\x,0.15) -- (\x,-0.15);
    \draw[guide] (\x,0) -- (\x,3.0);
    \node[label, below] at (\x,-0.2) {$\lab$};
}

\draw[<-, thick] (-4,0) -- (-3.4,0);

\node[
    annot,
    draw,
    fill=green!10,
    rectangle,
    inner sep=4pt
] at (-4.0,0.6) {\scriptsize $\mathrm{SWAP}$} ;
\node[annot] at (-3.3,0.6) {\scriptsize [Thm. \ref{thm:SwapBound}]};

\node[
    annot,
    draw,
    fill=green!10,
    rectangle,
    inner sep=4pt
] at (-2,0.6) {\scriptsize $\widehat{\mathcal{C}}_U$};
\node[annot] at (-1.4,0.6) {\scriptsize [Prop. \ref{thm:contr-unit}]};

\node[
    annot,
    draw,
    fill=green!10,
    rectangle,
    inner sep=4pt
] at (0,0.6) {\scriptsize $\mathcal{U}_A\otimes\mathcal{U}_B$};
\node[annot] at (0.85,0.6) {\scriptsize [Eq. \eqref{eq:nsu-product}]};

\node[
    annot,
    draw,
    fill=green!10,
    rectangle,
    inner sep=4pt
] at (4,0.6) {\scriptsize $\mathcal{R}^\pi$};

\draw[annot,->] (-2.0,1.2) -- (-0.6,1.2);
\node[annot, above] at (-1.3,1.2) {\scriptsize $\mathcal{C}_U$ [Prop.~\ref{thm:contr-unit}]};

\draw[annot,->] (0.0,1.4) -- (1.5,1.4);
\node[annot, above] at (0.9,2.2) {\scriptsize $\mathrm{C{\text -}SEP{\text -}P}$ [Cor.\ \ref{prop:seperable}],};
\node[annot, above] at (0.9,1.8) {\scriptsize $\mathrm{C{\text -}PPT{\text -}P}$ [Prop.\ \ref{prop:ppt}],};
\node[annot, above] at (0.8,1.4) {\scriptsize $\NS$ [Prop. \ref{prop:nonsig-unit}]};

\end{tikzpicture}

\caption{Summary of the signaling properties of channels and the bounds on the conditional channel min-entropy $S_{\infty}[A|B]_{\mathcal N}$ for different classes of bipartite  channels $\n_{A'B'\to AB}$. For simplicity, we have assumed that $|A'|=|B'|=|B|=|A|$. The further the conditional entropy of a unitary channel lies to the left of $-\log|A|$, the greater the signaling from $A' \to B$ exhibited by the channel. Maximally entangling unitary channels sit at the left extremal point of $-3\log |A|$, while the controlled unitary channels $\mathcal{C}_\U$ are lower bounded by $-2\log|A|$, indicating a clear gap between the signaling ability of swap-like and controlled unitary channels. $\widehat{\mathcal{C}}_\U$ denotes the controlled unitary with orthogonal unitary operators. The conditional min-entropy of channels $\n$ in the sets ${\rm C{\text -}SEP{\text -}P}$, ${\rm C{\text -}PPT{\text -}P}$, and $\NS$ are all lower bounded by $-\log|A|$. See Section \ref{sec:signaling} for related results in generality.}
\label{fig:channel-hierarchy}
\end{figure}

The dependence of conditional athermality (and purity) distillation and formation rates on conditional channel entropies reflects an intrinsic relation between the channel's thermodynamic resource content and its underlying causal structure. We consider broad classes of channels to illustrate the interdependence between underlying causal structure and conditional entropy (\autoref{fig:channel-hierarchy} and Section \ref{sec:unitaries}). We summarize some of our main results on the conditional min-entropy of a bipartite quantum channel in \autoref{tab:characteristics} in blue text. In the table, we make comparison with analogous observations for conditional min-entropy of states~\cite{T21}. The properties of conditional channel entropies can be of independent interest and serve as tools to analyze other informational aspects of quantum processes in quantum communication and computation, many-body quantum physics, open quantum systems, etc.

\subsection{Related works}\label{sec:prior}
The resource theory of athermality of quantum states has been widely studied under various settings~\cite{Len78,Gou24,Los19}, along with different classes of free operations such as Gibbs-preserving channels, thermal operations, covariant thermal operations, and catalytic assistance~\cite{BHO+13,BHN+15,AHJ18,LBT19,LS21,JGW25}. The resource theory of athermality of states generalizes the resource theory of purity of states~\cite{HO12,WY16,SKW+18}. In parallel, the thermodynamic aspects of quantum channels have been investigated to better understand the energetics of quantum processors, as well as their informational and computational limitations~\cite{NG15,ZA25,DS25,LMB25,HWS+25,BGC+26,BGD25b}.

Our work complements an alternative line of research in which the resourcefulness of athermal channels is studied under the assumption that unitaries commuting with the Hamiltonian are free operations~\cite{NG15,FBB19}. Although the identity channel trivially commutes with the Hamiltonian, preserving a quantum state in isolation from the environment incurs an energy cost. This observation motivates the need for a resource theory in which identity and noiseless quantum channels are not considered free. The resource theory of purity of quantum channels provides a step in this direction for trivial Hamiltonians~\cite{YZGZ20,YHW19}. Recent works~\cite{DS25,BGD25,BGD25b}, including the present one, align with and further advance this ongoing effort for non-trivial Hamiltonians. 

In \cite{BGD25b,BGD25}, thermodynamic cost of the transformation between two channels was determined and the basic thermodynamic concept of the channel free energy was introduced using both axiomatic and operational approaches. The free energy of quantum channels was shown to be dependent on different information processing capabilities of a channel like entanglement-assisted classical capacity~\cite{BSST02} and private randomness distillation capacity~\cite{YHW19}.

We directly generalize the resource theory of athermality for single-input single-output quantum channels under the action of Gibbs-preserving superchannels (GPSs)~\cite{BGD25b,BGD25}. Furthermore, we introduce a resource theory of conditional athermality for bipartite quantum channels under CGPSs. This framework reduces to the standard resource theory of athermality for quantum channels when the bipartite channels are decoupled (i.e., tensor-product channels between the two parties). Our analysis also shows that the resource theory of athermality is asymptotically reversible under the action of GPSs, even when the input and output dimensions of the channels differ. In addition, we generalize the resource theory of purity of quantum channels to a resource theory of conditional purity for bipartite channels under conditional uniformity-preserving superchannels.

Our work provides operational interpretations of the hypothesis-testing, max-, and regularized relative entropies between quantum channels and decoupled quantum channels in thermodynamic settings. Building upon known limitations in channel discrimination tasks~\cite{CMW16,CE16,Das19,WW19b,WBHK20,FFRS20}, we determine asymptotic rates for the distillation and formation of conditional purity for broad families of quantum channels. Our contributions also connect to broader developments in the resource theories of quantum channels, including entanglement, asymmetry, purity, and coherence~\cite{BDWW19,Das19,LW19,SCG20,LY20,TWH20,DBWH21}, as well as to the thermodynamic interpretation of negative entropic quantities~\cite{RAR+11,JGW25,doBL25,BGC+26} and the role of quantum correlations in thermodynamics~\cite{Ben03,PHS15,GMN+15,BRLW17,LMB25,BGD25}.

The clear operational interpretation of the conditional entropy of quantum states~\cite{CA97}, as established in~\cite{HOW05,HOW06}, represents a landmark result in quantum information theory. In a similar spirit, we expect that our operational characterization of conditional channel entropies will yield new insights into quantum processes relevant to quantum information processing and computation, and will prove to be broadly useful across quantum information science.

\subsection{Outline}
The paper is organized as follows. In Section~\ref{section:prelim}, we introduce the necessary preliminaries and notations used throughout the paper. We also review generalized divergences and their channel extensions, which form the basis of entropic quantities. In particular, we introduce the generalized entropy and conditional entropy of bipartite quantum channels, along with their smoothed variants. In Section~\ref{section:athermality}, we introduce the framework for the resource theory of conditional athermality for bipartite quantum channels, where the free objects are conditional absolutely thermalizing channels and the free operations are conditional Gibbs-preserving superchannels. Within this framework, we formalize the operational tasks of conditional athermality distillation and formation, and determine limitations on their one-shot and asymptotic rates. In Section~\ref{section:min-entropy}, we prove several properties of the conditional min-entropy of quantum channels, along with its bounds and the asymptotic equipartition property. In Section~\ref{sec:asym-rate}, we consider bipartite channels with trivial Hamiltonians where the resource theory of conditional athermality reduces to the resource theory of conditional purity under conditional purity-preserving superchannels. We show that the conditional distillation and formation capacities are related to conditional channel entropies, thereby providing an operational interpretation of these quantities. We prove that the resource theory of conditional purity is asymptotically reversible under conditional purity-preserving superchannels. We discuss the relation of the conditional purity capacity of a tele-covariant channel with the superdense coding capacity of its Choi state. In Section~\ref{sec:signaling}, we explore the implications of the causal structure of a bipartite channel on its conditional channel entropy and the conditional purity distillation and formation rates. We prove the asymptotic reversibility of the resource theory of conditional purity for semicausal channels that are no-signaling from non-connditioning input to conditioning output. We consider several examples of bipartite channels of practical interest and discuss bounds on their conditional channel min-entropies. Finally, we conclude in Section~\ref{section:discussions} with a brief discussion of potential open problems.

\section{Preliminaries}\label{section:prelim}
\subsection{Standard notations and definitions}
We consider separable Hilbert spaces of finite-dimensions. Quantum states are described by positive semidefinite operators with unit trace. Quantum channels are physical transformations of quantum states and described by completely positive, trace-preserving linear maps. Quantum superchannels describe physical transformations of quantum channels and its action on a quantum channel is described by the concatenation of the channel with pre-processing and post-processing quantum channels.

A quantum system $A$ and associated the Hilbert space $\mathcal{H}_A$ are both denoted with $A$ for simplicity, $\dim(\mathcal{H}_A)=|A|$ and $AB:=A\otimes B$. ${\rm L}(A,B)$ denotes the set of all linear transformations (or operators) from $A\to B$. $\mathrm{L}(A)$ and $\mathrm{L}_+(A)$ denotes the set of bounded operators and positive semidefinite operators defined on $A$, respectively. $X_A$ denotes $X\in{\rm L}(A)$ and $\mathbbm{1}_A$ is the identity operator on $A$. $\mathrm{St}(A)$ denotes the set of all density operators on $A$; $\mathrm{St}_{\leq}(A)$ denotes the set of all subnormalized density operators on $A$. $\mathrm{Ch}(A',A)$ denotes the set of all quantum channels from $\mathrm{L}(A')\to \mathrm{L}(A)$. A superoperator $\mathcal{N}_{A\to B}$ denotes a linear map from ${\rm L}(A)\to {\rm L}(B)$. $\id_{A'\to A}$ denotes the identity channel, where $A'\simeq A$. Let $\Gamma_{RA'}:=\sum_{i,j=0}^{d-1}\ket{ii}\bra{jj}_{RA'}$, where $d=\min\{|R|,|A'|\}$ and $\{\ket{i}\}_{i=0}^{d-1}$ is an orthonormal set of vectors, denote the maximally entangled operator and $\Phi_{RA'}:=\frac{1}{d}\Gamma_{RA'}$ is a maximally entangled state. The Choi operator of a linear map $\n:{\rm L}(A')\to {\rm L}(A)$ is given by $\Gamma^{\n}_{RA}:=\id_{R}\otimes\n_{A'\to A}(\Gamma_{RA'})$, where $R\simeq A'$. For a quantum channel $\n_{A'\to A}$, its Choi state is $\Phi^{\n}_{RA}:=\id_R\otimes\n_{A'\to A}(\Phi_{RA'})$, where $R\simeq A'$. The Choi state of a bipartite quantum channel $\n_{A'B'\to AB}$ is $\Phi^{\n}_{R_AAR_BB}=\n(\Phi_{R_AA'}\otimes\Phi_{R_BB'})$. For a bipartite state $\rho_{AB}$, we use $\rho_A:=\tr_B(\rho_{AB})$.

An operator $U_{A'\to A}$ is unitary if $U^\dag U=\mathbbm{1}_{A'}$ and $UU^\dag=\mathbbm{1}_A$ and corresponding channel is $\U_{A'\to A}(\cdot)=U(\cdot)U^\dag$, and it holds that $A'\simeq A$. An operator $V_{A'\to A}$ is isometry if  $V^\dag V=\mathbbm{1}_{A'}$ and $VV^\dag=\Pi_{A}$, where $|A'|\leq |A|$ and $\Pi_A$ is projection on $A$ with ${\rm rank}(\Pi_A)=|A'|$; the corresponding isometry channel is $\V_{A'\to A}(\cdot)=V(\cdot)V^\dag$. Let $\mathrm{SCh}((A',A),(B',B))$ denote the set of all quantum superchannels that maps $\mathrm{Ch}(A',A)\to \mathrm{Ch}(B',B)$, and the action of any $\Theta\in\mathrm{SCh}((A',A),(B',B))$ on $\n\in\Ch(A',A)$ can be expressed as~\cite{CDGP08,CDP09}
\begin{equation}
    \Theta(\n)= \mathcal{Q}_{CA\to B}\circ\mathcal{N}_{A'\to A}\circ\mathcal{P}_{B'\to CA'},
\end{equation}
where $\mathcal{P}$ is some channel that pre-processes $\n$ and $\Q$ is some channel that post-processes $\n$.

 Let ${\rm SEP}(A;B)$ denote the set of all separable states $\rho_{AB}$ that is separable between $A,B$, i.e., all $\rho\in\St(AB)$ of the form $\rho_{AB}=\sum_xp_X(x)\omega^x_A\otimes\tau^x_B$ for some probability distribution $\{p_X(x)\}_x$ and $\omega^x\in\St(A)$ and $\tau^x\in\St(B)$ for each $x$. Let ${\rm Ent}(A;B)$ be the set of all states $\rho_{AB}$ that is entangled (not separable) between $A,B$. Let ${\rm PPT}(A;B)$ denote the set of all states $\rho_{AB}$ that remains positive under the partial transposition on $B$, i.e., ${\rm T}_B(\rho_{AB})\geq 0$ for transposition ${\rm T}_B$ with respect to some basis on $B$ (without loss of generality)~\cite{HHHH09}. A bipartite state $\rho_{AB}$ is called NPT if $\rho_{AB}\notin{\rm PPT}(A;B)$.

The action of a swap operator is ${{U}^{\rm SWAP}}_{A'B'\to AB}(\ket{i}_{A'}\otimes\ket{j}_{B'})=\ket{j}_A\otimes\ket{i}_B$ for each vectors in some orthornomal bases $\{\ket{i}_{A'}\}_i$ and $\{\ket{j}_{B'}\}_j$, and $\{\ket{j}_A\}_j$ and $\{\ket{i}_B\}$ are some orthonormal sets of vectors. Its Choi operator has the form $\Gamma^{{\mathcal{S}}_{A:B}}=\Gamma_{R_AB}\otimes\Gamma_{R_BA}=\Gamma_{R_AA:R_BB}$. We call a unitary operation $\U_{A'B'\to AB}$ as maximally entangling or a swap-like operation (channel), denoted by $\U\in{\rm SWAP}_{A;B}$, if its Choi operator $\Phi^{\U}_{R_AAR_BB}\in{\rm Ent}(R_AA;R_BB)$ is maximally entangled of Schmidt rank $\min\{|A'||A|,|B'||B|\}$. Here, ${\rm SWAP}(A;B)$ denotes the set of all swap-like operations from $A'B'\to AB$. A swap operation composed with pre-processing local unitaries on $A',B'$ and post-processing local isometries on $A,B$ is a swap-like operation.

We denote the Hamiltonian of a quantum system $A$ with $\widehat{H}_A$. $\gamma^\beta_A:={\exp(-\beta\widehat{H}_A)}/{Z^\beta_{A}}$, where $Z^\beta_A:=\tr[\exp(-\beta\widehat{H}_A)]$ is the partition function, denotes a thermal or Gibbs state of $A$ associated with inverse temperature $\beta$ and $\widehat{H}_A$. The Hamiltonian $\widehat{H}_{AB}$ of two noninteracting quantum systems $A$ and $B$ is given by $\widehat{H}_A+\widehat{H}_B:=\widehat{H}_A\otimes\mathbbm{1}_B+\mathbbm{1}_A\otimes\widehat{H}_B$ as the interaction Hamiltonian is absent, $\widehat{H}^{\rm int}_{AB}=0$. Let $\mathcal{R}^{\omega}_{A'\to A}$ denote the completely positive linear map that replaces the input with $\omega_A$, $\mathcal{R}^{\omega}_{A'\to A}(X_{A'}):=\tr(X_{A'})\omega_A$. The absolutely thermal channel is defined as $\mathcal{T}^\beta_{A'\to A}(\cdot)=\mathcal{R}^{\gamma^\beta}_{A'\to A}(\cdot)=\tr(\cdot)\gamma^\beta_A$. For a trivial Hamiltonian $\widehat{H}_A\propto \mathbbm{1}_A$, $\gamma^\beta_A$ is equal to the maximally mixed state $\pi_A:=\frac{1}{|A|}\mathbbm{1}_A$.

$\norm{X}_p:=(\tr[\abs{X}^p])^{1/p}$ is the Schatten $p$-norm of an operator $X$ for $p\in[1,\infty)$ (definition is also extended to $p\in(0,1)$ but then $\norm{X}_p$ is not a norm), where $\norm{X}_1$ is also called the trace-norm. Fidelity between $\rho_A\in\St(A)$ and $\sigma_A\in{\rm L}_+(A)$ is defined as $F(\rho_A,\sigma_A):=\norm{\sqrt{\rho}\sqrt{\sigma}}_1^2$, and generalized fidelity is defined as ${F}_{\leq}(\rho,\sigma):=\left[\sqrt{F(\rho,\sigma)}+\sqrt{(1-\tr(\rho))(1-\tr(\sigma))}\right]^2$ for $\rho,\sigma\in{\rm St}_{\leq}(A)$. The generalized purified distance $P(\rho,\sigma):=\sqrt{1-F_{\leq}(\rho,\sigma)}$ for $\rho,\sigma\in\St_{\leq}(A)$.

\subsection{Generalized divergence}
The generalized state divergence ${\bf D}:\St(A)\times\St(A)\to \mathbb{R}$ is a real-valued function on a pair of states that is monotone under the action of an arbitrary quantum channel $\n_{A\to B}$, i.e., ${\bf D}(\rho\Vert\sigma)\geq {\bf D}(\n(\rho)\Vert\n(\sigma))$ for arbitrary $\rho,\sigma\in\St(A)$ and arbitrary $\n\in\Ch(A,B)$~\cite{Ren05,Dat09,CMW16}. Some examples of the generalized state divergences are hypothesis-testing relative entropy $D^{\varepsilon}_H$, quantum relative entropy $D$~\cite{Ume62}, sandwiched R\'enyi relative entropy $D_{\alpha}$ for $\alpha\in[\frac{1}{2},1)\cup(1,\infty)$~\cite{MDSFT13,WWY14}, Petz-R\'enyi relative entropy $\overline{D}_{\alpha}$ for $\alpha\in[0,1)\cup(1,2]$, trace-distance, purified distance, etc.

For arbitrary $\rho\in{\rm L}(A)$ and arbitrary $\sigma\in{\rm L}_+(A)$,
\begin{itemize}
  \item {Quantum Relative Entropy:}
  \begin{equation}
      D(\rho \| \sigma) := \tr[\rho(\log \rho - \log \sigma)] 
  \end{equation}
    provided $\supp(\rho) \subseteq {\supp}(\sigma)$, and is $+\infty$ otherwise.
    
    \item Hypothesis-Testing relative entropy: For $\varepsilon \in [0,1]$,
    \begin{equation}
  D_H^\varepsilon(\rho \| \sigma) = -\log \inf \{ \tr(\Lambda \sigma) : 0 \leq \Lambda \leq \mathbbm{1}, \tr(\Lambda \rho) \geq 1 - \varepsilon \}. 
  \end{equation}
  
    \item Sandwiched R\'enyi relative entropy: For $\alpha \in [\frac{1}{2}, 1) \cup (1, \infty)$,
    \begin{equation}
        {D}_\alpha(\rho \| \sigma) = \frac{1}{\alpha - 1} \log \tr \left[ \left( \sigma^{\frac{1-\alpha}{2\alpha}} \rho \sigma^{\frac{1-\alpha}{2\alpha}} \right)^\alpha \right],
    \end{equation}
    where ${D}_\alpha(\rho \| \sigma) = +\infty$ for $\alpha > 1$ if ${\supp}(\rho) \not\subseteq {\supp}(\sigma)$.

    \item Petz-R\'enyi relative entropy: For $\alpha \in [0, 1) \cup (1, \infty)$,
    \begin{equation}
        \overline{D}_\alpha(\rho \| \sigma) = \frac{1}{\alpha - 1} \log \tr \left( \rho^\alpha \sigma^{1-\alpha} \right)
    \end{equation}    where $\overline{D}_\alpha(\rho \| \sigma) = +\infty$ for $\alpha > 1$ if $\supp(\rho) \not\subseteq \supp(\sigma)$.

    \item {Quantum max-relative entropy:}
    \begin{equation}
        D_{\max}(\rho \| \sigma) = \log \inf \{ \lambda \geq 0 : \rho \leq \lambda \sigma \}
    \end{equation}
    where $D_{\max}(\rho \| \sigma) = +\infty$ if $\supp(\rho) \not\subseteq {\supp}(\sigma)$.
\end{itemize}
We have $\lim_{\alpha\to 1}{D}_\alpha(\rho\|\sigma)=D(\rho\|\sigma)=\lim_{\alpha\to 1}\overline{D}_\alpha(\rho\|\sigma)$, $D_{\infty}(\rho\|\sigma):=\lim_{\alpha\to \infty}{D}_\alpha(\rho\|\sigma)=D_{\max}(\rho\|\sigma)$, and $\overline{D}_0(\rho\|\sigma)=D^{\varepsilon=0}_H(\rho\|\sigma)$. We consider $\log$ with respect to base $2$ unless stated otherwise.

{\it Generalized channel divergence}: The generalized channel divergence $\mathbf{D}:\Ch(A,B)\times\Ch(A,B)\to \mathbb{R}$ is a real-valued function that is based on the generalized state divergence ${\bf D}(\cdot\|\cdot)$ and defined between a quantum channel $\n_{A\to B}$ relative to a completely-positive map $\m_{A\to B}$ as~\cite{LKDW18}(see also~\cite{SPSD25})
\begin{equation}\label{eq:gen-ch-div}
    \mathbf{D}[\n\|\m]:=\sup_{\rho\in\St(RA)}{\bf D}(\id\otimes\n(\rho_{RA})\|\id\otimes\m(\rho_{RA})),
\end{equation}
where supremum is over all possible states with arbitrary $|R|$. The generalized channel divergence between two quantum channels is monotone under the action of quantum superchannels~\cite[Theorem 1]{DGP24}. The generalized channel divergence ${\bf D}[\cdot\|\cdot]$ is called sandwiched R\'enyi relative entropy $D_{\alpha}[\cdot\|\cdot]$, hypothesis-testing relative entropy $D^\varepsilon_H[\cdot\|\cdot]$, quantum relative entropy $D[\cdot\|\cdot]$, purified distance $P[\cdot,\cdot]$, etc.~when the associated generalized state divergence ${\bf D}(\cdot\|\cdot)$ is taken to be $D_{\alpha}(\cdot\|\cdot)$, $D^\varepsilon_H(\cdot\|\cdot)$, $D(\cdot\|\cdot)$, purified distance $P(\cdot,\cdot)$, etc., respectively. It suffices to optimize over pure input states $\rho\in\St(RA)$ such that $R\simeq A$ in Eq.~\eqref{eq:gen-ch-div} for $D^\varepsilon_H[\cdot\|\cdot]$, $D_{\alpha}[\cdot\|\cdot]$ (when $\alpha\in[\frac{1}{2})\cup(1,\infty))$, $D[\cdot\|\cdot]$, and $P[\cdot,\cdot]$. For quantum channels $\n,\m\in\Ch(A,B)$, $D_{\infty}[\n\|\m]=D_{\infty}(\Gamma^\n_{R_AB}\|\Gamma^\m_{R_AB})=D_{\infty}(\Phi^\n_{R_AB}\|\Phi^m_{R_AB})$, where $R_A\simeq A$~\cite{WBHK20}.

The regularized generalized channel divergence ${\bf D}^{\rm reg}[\n\|\m]$ for channels $\n,\m\in\Ch(A,B)$ is defined as~\cite{WBHK20,FFRS20}
\begin{equation}\label{eq:reg-gen-div}
    {\bf D}^{\rm reg}[\n\|\m]:=\lim_{n\to\infty}\frac{1}{n}{\bf D}[\n^{\otimes n}\|\m^{\otimes n}].
\end{equation}

\subsection{Entropies and conditional entropies}
The generalized state and channel divergences have been used to provide a collective method to define entropic and information measures; see, for example,~\cite{T21,BCY26,DGP24}. The entropy functions quantify the amount of randomness or uncertainty present in a state. The conditional entropy functions quantify the amount of randomness or uncertainty present in a state when part of the system, or memory, is accessible to the observer. See \autoref{tab:entropy_bounds_vertical}, where we compare dimensional bounds for entropy and conditional entropy of random variables, states, and channels.

\begin{table}[h]
\centering
\renewcommand{\arraystretch}{1.5}
\begin{tabular}{|l|c|c|}
\hline
& \textbf{Min-entropy} & \textbf{Conditional min-entropy} \\ \hline
Random variables & $0 \leq H_{\infty}(X) \leq \log |X|$~\cite{Ren05,T21} & $\textcolor{orange}{0} \leq H_{\infty}(X|Y) \leq \log |X|$ ~\cite{Ren05,T21} \\ \hline
Quantum states & $\textcolor{orange}{0} \leq S_{\infty}(A)_{\rho} \leq \log |A|$~\cite{Ren05} & $\textcolor{red}{-\log |A|} \leq S_{\infty}(A|B)_{\rho} \leq \log |A|$~\cite{KRS09} \\ \hline
Quantum channels & $\textcolor{red}{-\log |A|} \leq S_{\infty}[A]_{\mathcal{M}} \leq \log |A|$~\cite{GW21} & $\textcolor{blue}{-\log |A'| |B'| |B|} \leq S_{\infty}[A|B]_{\mathcal{N}} \leq \log |A|$ \\ \hline
\end{tabular}
\caption{We compare dimensional bounds for min-entropies~\cite{Ren05,BGC+26} vs.~conditional min-entropies~\cite{T21,DGP24} of random variables, quantum states, and quantum channels. We have min-entropy $H(X)$ and conditional min-entropy $H(X|Y)$ for classical probability distributions, min-entropy $S_{\infty}(A)_{\rho}$ and conditional min-entropy $S_{\infty}(A|B)_{\rho}$ for quantum states, and min-entropy $S_{\infty}[A]_{\mathcal{N}}$ and conditional min-entropy $S[A|B]_{\mathcal{N}}$ of point-to-point channels $\mathcal{M}_{A'\to A}$ and bipartite quantum channels $\mathcal{N}_{A'B'\to AB}$ (\autoref{prop:entropybound1}), respectively.}
\label{tab:entropy_bounds_vertical}
\end{table}

The generalized entropy ${\bf S}(A)_{\rho}$ of a quantum state $\rho_A$ is defined as
\begin{equation}
    {\bf S}(A)_{\rho}:= - \mathbf{D}(\rho_A\|\mathbbm{1}_A),
\end{equation}
for ${\bf D}(\cdot\|\cdot)$ being quantum relative entropy, sandwiched R\'enyi relative entropy, hypothesis-testing relative entropy etc. The generalized entropy reduces to the von Neumann entropy $S(A)_{\rho}:=S(\rho_A):=-\tr[\rho\log\rho]$ when the associated state divergence is quantum relative entropy.

The generalized conditional entropy ${\bf S}(A|B)_{\rho}$ of a bipartite quantum state $\rho_{AB}$ is defined as 
\begin{equation}\label{eq:gen-con-st}
    {\bf S}(A|B)_{\rho}:= {\bf S}^{\uparrow}(A|B)_{\rho}:= -\inf_{\sigma\in\St(B)}{\bf D}(\rho_{AB}\|\mathbbm{1}_A\otimes\sigma_B),
\end{equation}
for ${\bf D}(\cdot\|\cdot)$ being quantum relative entropy, sandwiched R\'enyi relative entropy, hypothesis-testing relative entropy, Petz-R\'enyi relative entropy, etc. We can also define a related quantifier for the sandwiched R\'enyi conditional entropy in the following way: for $\alpha\in[\frac{1}{2},1)\cup(1,\infty)$ and a quantum state $\rho_{AB}$,
\begin{equation}
     {S}^{\downarrow}_{\alpha}(A|B)_{\rho}:= -{D}_{\alpha}(\rho_{AB}\|\mathbbm{1}_A\otimes\rho_B)\leq {S}_{\alpha}(A|B)_{\rho}.
\end{equation}
For the von Neumann conditional entropy of a quantum state $\rho_{AB}$, we have
\begin{align}
    S(A|B)_{\rho}:=-\inf_{\sigma\in\St(B)}D(\rho_{AB}\|\mathbbm{1}_A\otimes\sigma_B)=-D(\rho_{AB}\|\mathbbm{1}_A\otimes\rho_B)=S(AB)_{\rho}-S(B)_{\rho}.
\end{align}

\textit{Generalized conditional channel entropy}: The generalized conditional channel entropy of a bipartite quantum channel $\mathcal{N}_{A'B'\to AB}$ is defined as~\cite{DGP24}
\begin{equation}\label{eq:gen-con-ent}
    {\bf S}[A|B]_{\mathcal{N}}:=-\inf_{\mathcal{Q}\in\Ch(B',B)}{\bf D}[\mathcal{N}_{A'B'\to AB}\Vert\mathcal{R}^{\mathbbm{1}}_{A'\to A}\otimes\mathcal{Q}_{B'\to B}],
\end{equation}
for ${\bf D}[\cdot\|\cdot]$ being quantum relative entropy, sandwiched R\'enyi relative entropy, hypothesis-testing relative entropy, etc.~between completely positive maps. For a replacer channel $\mathcal{R}^{\omega}_{A'B'\to AB}$, i.e., $\mathcal{R}^{\omega}_{A'B'\to AB}(\cdot)=\tr(\cdot)\omega_{AB}$ for a given $\omega\in\St(AB)$, ${\bf S}[A|B]_{\n}={\bf S}(A|B)_{\omega}$ whenever ${\bf D}(\rho\otimes\sigma\|\tau\otimes\sigma)={\bf D}(\rho\|\tau)$. The generalized channel entropy ${\bf S}[A]_{\Q}$ of a quantum channel $\Q_{A'\to A}$ is defined as ${\bf S}[\Q]:=-{\bf D}[\Q\Vert\R]$. 

A channel $\Q_{B'\to B}$ in Eq.~\eqref{eq:gen-con-ent} can be considered to be a side channel in a similar spirit as a state $\sigma_B$ in Eq.~\eqref{eq:gen-con-st} can be considered as a side information. The entropic and conditional entropic functions are named after the generalized state or channel divergences they are based on.

\textit{Smoothened relative-entropy:} The max-relative entropy function $D_{\max}(\rho\Vert\sigma)$ can show sensitive dependence on the tails of the spectrum of $\rho$, such that small changes in $\rho$ can lead to drastic changes in max-relative entropy. Therefore the concept of smoothened max-relative entropy is introduced which is more robust against finite errors that inevitably creep in during any operational implementation. To define smoothed relative entropy we first introduce the $\varepsilon$-ball $\mathcal{B}^\varepsilon(\rho)$ around the state $\rho$ in Hilbert space $A$ as
\begin{align}
    \mathcal{B}^\varepsilon(\rho):=\{\sigma\in\St_\le (A): P(\rho,\sigma)\le \varepsilon\}.
\end{align}
The smoothed max-relative entropy is given as $D_{\infty}^\varepsilon(\rho\Vert\sigma):=\inf_{\rho'\in\mathcal{B}^\varepsilon(\rho)}D_\infty(\rho'\Vert\sigma)$. We can similarly define a $\varepsilon$-ball $\mathcal{B}^\varepsilon[\n]$ around the channel $\n_{A'\to A}$ as 
\begin{align}
    \mathcal{B}^\varepsilon[\n]:=\{\mathcal{Q}\in \mathrm{Ch}(A',A): P[\n,\mathcal{Q}]\le \varepsilon\}.
\end{align}
where $P[\cdot,\cdot]$ is the purified distance. The smoothed max-relative entropy defined with respect to the $\varepsilon$-ball defined above is given as 
\begin{align}
D_{\infty}^\varepsilon[\n\Vert\mathcal{M}]:=\inf_{\mathcal{Q}\in\mathcal{B}[\n]}D_{\infty}^\varepsilon[\mathcal{Q}\Vert\mathcal{M}].
\end{align}
and the smoothed conditional channel min-entropy for the channel $\n_{A'B'\to AB}$ is given as
\begin{align}\label{eq:sm-min-ent}
    S_{\infty}^\varepsilon[A|B]_{\n}:=&\sup_{\mathcal{M}\in \mathcal{B}^\varepsilon[\n]}S_{\infty}[A|B]_{\mathcal{M}}\\
    =&-\inf_{\mathcal{Q}\in\Ch(B',B)}D_\infty^\varepsilon[\mathcal{N}_{A'B'\to AB}\Vert\mathcal{R}^{\mathbbm{1}}_{A'\to A}\otimes\mathcal{Q}_{B'\to B}]
\end{align}

\section{Conditional athermality distillation and formation}\label{section:athermality}
When we look around, we observe that different objects can possess different temperatures; even within a single system, one part may have thermalized while another remains far from thermal equilibrium. From a dynamical perspective, certain processes drive systems toward thermalization, whereas other coexisting processes often don't exhibit such tendency. A simplified (toy-model) representation of such naturally occurring dynamics can be described by a bipartite quantum channel of the form $\mathcal{T}^{\beta}_{A' \to A} \otimes \mathcal{Q}_{B' \to B}$, where $\mathcal{T}^\beta$ tends to thermalize all inputs to $\gamma^\beta_A$. A thermal state $\gamma^\beta_A$ is in equilibrium with a thermal bath at a temperature $\beta\in(0,\infty)$. This decomposition formalizes the idea that two processes, one strictly thermal $\mathcal{T}^{\beta}$ and the other described by an arbitrary channel $\mathcal{Q}$, coexist and act in parallel on distinct subsystems. We are going to refer to bipartite quantum channels of the form $\mathcal{T}^\beta\otimes\Q$ as partially (or conditionally) thermalizing channel. Notice that for a given $\T^\beta_{A'\to A}$, thermal state $\gamma^\beta_A$ of the output gets fixed. Also for $\widehat{H}_A\propto \mathbbm{1}_A$, we have $\gamma^\beta_A=\pi_A=\frac{1}{|A|}\mathbbm{1}_A$.

\textit{Resource theory of conditional athermality}: To formalize these intuitive observations, we introduce the resource theory of conditional athermality of bipartite quantum processes. For simplicity, we consider the conditional athermal resource theory for an arbitrary bipartite quantum channel $\mathcal{N}_{A'B'\to AB}$. The free objects are bipartite quantum channels of the form $\mathcal{T}^\beta_{A'\to A}\otimes\mathcal{Q}_{B'\to B}$. The equilibrium state $\gamma^\beta$ corresponding to the bath at temperature $\beta$ fixes the absolutely thermal channel $\mathcal{T}^\beta_{A'\to A}$ for parts $A', A$ of the system for a given conditional athermal resource theory. That is, $\mathcal{T}^{\beta'}\otimes\mathcal{Q}$ is not a free object but rather a resource when the bath for $A$ is at temperature $\beta$ and $\beta'\neq \beta$~\cite{BGD25}. The free operations are conditional Gibbs-preserving superchannels (CGPSs) $\Theta^{\beta}$. A CGPS, $\Theta^{\beta}\in\SCh((A'B',AB),(C'D',CD))$, is defined as a superchannel that preserves conditionally thermalizing channel $\mathcal{T}^\beta$, i.e., for an arbitrary $\Q\in\Ch(B',B)$ there exists some $\Q'\in\Ch(D',D)$ such that
\begin{equation}
    \Theta^{\beta}(\mathcal{T}^\beta_{A'\to A}\otimes\mathcal{Q}_{B'\to B})=\mathcal{T}^\beta_{C'\to C}\otimes\mathcal{Q}'_{D'\to D}.
\end{equation}
These channels can be characterized as postprocessing with a conditional Gibbs covariant channel~\cite{JGW25}. We discuss these properties in Appendix~\ref{app:freesup}. The resourceful objects are bipartite quantum channels $\mathcal{N}_{A'B'\to AB}$ that are not partially thermalizing on $A$, i.e., channels that are not of the form $\mathcal{T}^\beta_{A'\to A}\otimes\mathcal{Q}_{B'\to B}$. To denote athermal resourcefulness of a quantum channel $\n_{A'B'\to AB}$ with respect to the given free object $\mathcal{T}^\beta_{A'\to A}\otimes\Q_{B'\to B}$, we may write $(\n,\mathcal{T}^\beta\otimes\mathcal{Q})$ whenever required. Two bipartite quantum channels $\n,\m$ are said to be conditional athermality equivalent if there exists some quantum channels $\Q, \Q'$ and a CGPS $\Theta_1^\beta(\mathcal{T}^\beta\otimes\Q)=\mathcal{T}^\beta\otimes\Q'$  that can convert $\n\to \m$, and similarly another CGPS $\Theta_2^\beta(\mathcal{T}^\beta\otimes\Q')=\mathcal{T}^\beta\otimes\Q$ that can convert $\m\to \n$, and we write $(\n,\mathcal{T}^\beta\otimes\Q)\sim (\m,\mathcal{T}^\beta\otimes\Q')$.

\textit{Resourcefulness}: We inspect channel free energies to determine the standard resource unit. The thermal channel free energy $F_{\rm T}^\beta[{\cal P}]$ of an arbitrary quantum channel $\mathcal{P}_{A'\to A}$ is given by $F_{\rm T}^\beta[{\cal P}]:=(-\log Z^{\beta}_A+D[\mathcal{P}\Vert\mathcal{T}^\beta])\beta^{-1} \ln 2$ and its resource-theoretic free energy $F^\beta[{\cal P}]=F_{\rm T}^\beta[{\cal P}]-F_{\rm T}^\beta[{\T^\beta}]$~\cite{BGD25b,BGD25}. For a given $\T^\beta_{A'\to A}$, the free energy of a quantum channel $\mathcal{P}_{A'\to A}$ is maximum only if $\mathcal{P}_{A'\to A}$ is an isometry channel and they are minimum iff $\mathcal{P}=\mathcal{T}^\beta$~\cite{BGD25,BGD25b} (also see \autoref{lem:aap-free-opt}). The resource-theoretic thermal channel free energy $F^\beta$ is additive for tensor-product quantum channels with noninteracting output Hamiltonians~\cite{BGD25}. For a bipartite quantum channel $\mathcal{P}_{A'\to A}\otimes\mathcal{Q}_{B'\to B}$, we have $F_{\rm T}^\beta[\mathcal{P}\otimes\mathcal{Q}]-F^\beta_{\rm T}[\mathcal{T}^\beta\otimes\Q]=F^{\beta}[\mathcal{P}]=\beta^{-1}D[\mathcal{P}\|\mathcal{T}^\beta]\ln2$. This leads to the fact that among tensor-product quantum channels $\mathcal{P}_{A'\to A}\otimes\mathcal{Q}_{B'\to B}$ for a fixed channel $\Q$, $F^\beta$ is maximum only if $\mathcal{P}$ is an isometry channel and minimum iff $\mathcal{P}=\mathcal{T}^\beta$. As Hamiltonian in general can be arbitrary, the choice of trivial Hamiltonian is made for the standard resource unit in the resource theory of athermality. The channel free energy $F^\beta[\mathcal{P}]$ of a unitary process is minimal when its output Hamiltonian is fully degenerate~\cite{BGD25}. In the resource theory of athermality under Gibbs-preserving superchannels~\cite{BGD25}, all unitary channels $\U_{A'\to A}$ with trivial Hamiltonian are resource equivalent to each other. The standard resource unit of size $m$ in the resource theory of athermality is taken to be the identity channel $\id_{m}$ with trivial Hamiltonian and input dimension of $m$, without loss of generality. 

\textit{Conditional athermality resource monotones}: The generalized resource-theoretic conditional thermal free energy ${\bf F}^\beta[A|B]_{\n}$ of an arbitrary bipartite quantum channel $\n_{A'B'\to AB}$ is defined with respect to conditional bath at temperature $\beta$ (or conditionally thermalizing channel on $A'$),
\begin{equation}
    {\bf F}^\beta[A|B]_{\n}:=\beta^{-1}\ln2\inf_{\Q\in\Ch(B',B)}{\bf D}[\n\|\mathcal{T}^\beta\otimes\Q].
\end{equation}
${\bf F}^\beta[A|B]_{\n}$ is nonincreasing under the action of an arbitrary CGPS $\Theta\in\SCh((A'B',AB),(C'D',CD))$. When ${\bf D}[\cdot\|\cdot]$ is channel sandwiched R\'enyi relative entropy ${D}_{\alpha}[\cdot\|\cdot]$ for $\alpha\in[\frac{1}{2},\infty]$, where $D_{\alpha}=D$ at $\alpha=1$ and $D_{\alpha}=D_{\infty}$ at $\alpha=\infty$, then $F_{\alpha}^\beta[A|B]_{\n}\geq 0$ for all $\n\in\Ch(A'B',AB)$. For a tensor-product channel $\mathcal{N}=\mathcal{P}_{A'\to A}\otimes\mathcal{Q}_{B'\to B}$, the sandiwched R\'enyi conditional free energy reduces to the sandwiched R\'enyi free energy of the channel $\mathcal{P}_{A'\to A}$, $F_{\alpha}^\beta[A|B]_{\n}=F^{\beta}_{{\rm T},\alpha}[A]_{\T^\beta}-F^{\beta}_{{\rm T},\alpha}[A|B]_{\n}=F_{\alpha}^\beta[A]_{\mathcal{P}}=\beta^{-1}\ln2~D_{\alpha}[\mathcal{P}\|\mathcal{T}^\beta]$ for all $\alpha\in[\frac{1}{2},\infty]$. For a bipartite replacer channel $\mathcal{R}^{\rho}_{A'B'\to AB}$, $F_{\alpha}^\beta[A|B]_{\mathcal{R}^\rho}=\beta^{-1}\ln2~\inf_{\omega\in\St(B)}{D}_{\alpha}(\rho_{AB}\|\gamma^\beta_A\otimes\omega_B)$ for $\alpha\in[\frac{1}{2},\infty]$. For a conditional isometry channel $\V_{A'\to A}\otimes\Q_{B'\to B}$ with $\widehat{H}_{A}\propto\mathbbm{1}_A$, where $\V$ is an isometry channel, $F^\beta[\V\otimes\Q]=\beta^{-1}\ln|A'||A|$ (see ~\autoref{lem:iso-genH}).

\textit{Standard resource unit}: All conditional unitary channels $\mathcal{U}_{A'\to A}\otimes\Q_{B'\to B}$ with trivial Hamiltonian $\widehat{H}_A\propto\mathbbm{1}_A$ are resource equivalent to each other. That is, $(\mathcal{U}\otimes\Q,\rp\otimes\Q)\sim (\id\otimes\Q,\rp\otimes\Q)$, where $\widehat{H}_A\propto \mathbbm{1}_A$, for any unitary channel $\mathcal{U}_{A'\to A}$, $\id_{A'\to A}$, and some channel $\Q$, since CGPSs can be defined as appropriate unitary pre-processing that can take any unitary channel to any other unitary channel. As Hamiltonian in general can be arbitrary, we set our standard conditional athermal (resource) unit with respect to trivial Hamiltonian. The standard resource unit of size $m$ is a conditional identity channel $(\id_m\otimes\Q,\mathcal{R}^{\pi}\otimes\Q)$, for the $m$-dimensional identity channel and some side channel $\Q_{B'\to B}$.

\textit{Resource conversion}: Resource conversion distance is a conditional athermality resource monotone that captures the distance between input resource channel and output resource channel under the action of CGPS. For the conversion of a resource channel $(\n_{A'B'\to AB},\mathcal{T}^{\beta}_{A'\to A}\otimes\mathcal{Q}_{B'\to B})$ to another resource channel $(\m_{C'D'\to CD},\mathcal{T}^{\beta}_{C'\to C}\otimes\mathcal{Q}'_{D'\to D})$ under the action of CGPS $\Theta^\beta$, the athermal resource conversion distance in terms of purified channel distance $P[\cdot,\cdot]$ is given by
\begin{align}\label{eq:conversion-dist}
d_{\mathrm{CGPS}}\Big((\mathcal{N},\T^{\beta}\otimes\mathcal{Q})\to(\mathcal{M},\T^{\beta}\otimes\mathcal{Q}')\Big):=\min_{\Theta\in\SCh} \{P\left[
\mathcal{N},\Theta(\mathcal{M})\right]:
\Theta(\T^{\beta}\otimes \mathcal{Q})=\T^{\beta}\otimes \mathcal{Q}'\}.
\end{align}

\textit{Conditional athermality distillation and formation}: For brevity, we at times denote the conditional identity channel $\id\otimes\Q$ as $\id_m\otimes \Q$ or $\id_m$ for $m$-dimensional input on $\id$, whenever it is clear from the context. In the task of conditional athermality distillation, we want to estimate the largest possible standard conditional athermality resource unit $(\id_m\otimes\Q',\mathcal{R}^\pi\otimes \Q')$ that a resource channel $(\n,\T^\beta\otimes\Q)$ can be transformed to, under the action of some CGPS up to an error $\varepsilon\in[0,1]$. In the task of conditional athermality dilution, we would like to form a channel $(\n,\mathcal{T}^\beta\otimes\Q)$ using the smallest standard resource unit $(\id_m\otimes\Q',\mathcal{R}^\pi\otimes \Q')$, under the action of CGPS up to an error $\varepsilon\in[0,1]$. To capture these notions, we define
\begin{align}
   & {\rm Dist}^{(1,\varepsilon)}((\n,\T^\beta\otimes\Q)\to (\id_m\otimes\Q',\mathcal{R}^\pi\otimes\Q')) \nonumber\\
    &\qquad\qquad := \sup_{m}\left\{
\log m : d_{\mathrm{CGPS}}\left((
\mathcal{N},\T^{\beta}\otimes\mathcal{Q})\to(\mathrm{id}_m\otimes\mathcal{Q}',\mathcal{R}^{\pi}\otimes\mathcal{Q}')\right)\le \varepsilon\right\}.\\
& \mathrm{Cost}^{(1,\varepsilon)}((\id_m\otimes\Q',\mathcal{R}^\pi\otimes\Q')\to (\n,\T^\beta\otimes\Q)) \nonumber\\
& \qquad\qquad :=\inf_{m}
\Bigl\{
\log m : d_{\mathrm{CGPS}}\left((
\mathrm{id}_m \otimes \mathcal{Q}',
\mathcal{R}^{\pi} \otimes \mathcal{Q}')
\to
(\mathcal{N}, \T^{\beta} \otimes \mathcal{Q})\right) \le \varepsilon\Bigr\}.
\end{align}

\begin{definition}[One-shot conditional athermality distillation and formation]
The one-shot conditional athermality distillation yield and formation cost of a quantum bipartite channel $\n_{A'B'\to AB}$ under the action of some CGPSs up to an error $\varepsilon\in[0,1]$, are defined as
\begin{align}
\mathrm{Dist}^{(1,\varepsilon)}(\mathcal{N}_{A'B'\to AB},\T^\beta_{A'\to A})
&:=\inf_{\Q\in\Ch}{\rm Dist}^{(1,\varepsilon}((\n,\T^\beta\otimes\Q)\to (\id_m\otimes\Q',\mathcal{R}^\pi\otimes\Q')),\\
\mathrm{Cost}^{(1,\varepsilon)}(\mathcal{N}_{A'B'\to AB},\T^\beta_{A'\to A})
&:=\inf_{\Q\in\Ch}\mathrm{Cost}^{(1,\varepsilon)}((\id_m\otimes\Q'',\mathcal{R}^\pi\otimes\Q'')\to (\n,\T^\beta\otimes\Q)),
\end{align}
for some quantum channels $\Q',\Q''$.
\end{definition}

$\mathrm{Dist}^{(1,\varepsilon)}(\mathcal{N}_{A'B'\to AB},\T^\beta_{A'\to A})$ is the maximum number of conditional qubit unitary channel, with trivial output Hamiltonian, that can be extracted from a single use of a resource channel $(\mathcal{N}_{A'B'\to AB},\T^\beta_{A'\to A}\otimes\Q_{B'\to B})$ through some CGPS for the worst-possible side channel $\Q$, up to an error $\varepsilon$. $\mathrm{Cost}^{(1,\varepsilon)}(\mathcal{N}_{A'B'\to AB},\T^\beta_{A'\to A})$ is the minimum number of conditional qubit unitary channel, with trivial output Hamiltonian, that is required to prepare a single copy of a resource channel $(\mathcal{N}_{A'B'\to AB},\T^\beta_{A'\to A}\otimes\Q_{B'\to B})$ through some CGPS with access to the best-possible side channel $\Q$, up to an error $\varepsilon$. 

\begin{definition}[Conditional athermality distillation and formation capacity]
The capacities of the conditional athermality distillation yield and formation cost of a quantum bipartite channel $\n_{A'B'\to AB}$ under the action of some CGPSs, are defined as the optimal distillation and formation rates over asymptotically many i.i.d.~uses of $\n$ are defined respectively as
\begin{align}
\mathrm{Dist}^{(\infty,0)}(\mathcal{N}_{A'B'\to AB},\T^\beta_{A'\to A})
&:=\lim_{\varepsilon\to 0^+}\lim_{n\to\infty}\frac{1}{n}\mathrm{Dist}^{(1,\varepsilon)}(\mathcal{N}^{\otimes n}_{{A'}^n{B'}^n\to A^nB^n},{\T^\beta}^{\otimes n}_{{A'}^{n}\to A^n}),\\
\mathrm{Cost}^{(\infty,0)}(\mathcal{N}_{A'B'\to AB},\T^\beta_{A'\to A})
&:=\lim_{\varepsilon\to 0^+}\lim_{n\to\infty}\frac{1}{n}\mathrm{Cost}^{(1,\varepsilon)}(\mathcal{N}^{\otimes n}_{{A'}^n{B'}^n\to A^nB^n},{\T^\beta}^{\otimes n}_{{A'}^{n}\to A^n}),
\end{align}
where ${\mathcal{T}^{\beta}}^{\otimes n}_{{A'}^n\to A^n}:=(\mathcal{T}^{\beta}_{A'\to A})^{\otimes n}$ is the tensor-product of $n$ $\T^{\beta}_{A'\to A}$, assuming that the channel outputs $A_1',A_2',\ldots,A_n'$ are noninteracting.
\end{definition}

We present one of the main results in the following theorem where we determine the one-shot conditional athermality distillation and formation rates of a quantum channel $\mathcal{N}_{A'B'\to AB}$, see Appendix~\ref{app:proof-dist-for-cost} for the detailed proof.
\begin{theorem}[Single-shot yield and cost]\label{th:dist-for-cost}
The one-shot conditional athermality distillation yield and formation cost of a quantum channel $\mathcal{N}_{A'B'\to AB}$ up to an error $\varepsilon\in[0,1]$, is given by
\begin{align}
\mathrm{Dist}^{(1,\varepsilon)}(\mathcal{N}_{A'B'\to AB},\T^\beta_{A'\to A})
&=\frac{1}{2}\inf_{\Q\in\Ch({B',B})}D_{\mathrm{H}}^{\varepsilon^2}[\n_{A'B'\to AB}\Vert\T^\beta_{A'\to A}\otimes \Q_{B'\to B}],\\
\mathrm{Cost}^{(1,\varepsilon)}(\mathcal{N}_{A'B'\to AB},\T^\beta_{A'\to A})
&=\frac{1}{2}\inf_{\mathcal{Q}\in\Ch({B',B})}D_{\infty}^{\varepsilon}[\n_{A'B'\to AB}\Vert\T^\beta_{A'\to A}\otimes\mathcal{\mathcal{Q}}_{B'\to B}].\label{eq:athermality-cost-rate-a}
\end{align}
\end{theorem}

We note that the one-shot conditional athermality distillation yield and formation cost of a resource channel $(\mathcal{N}_{A'B'\to AB},\T^\beta_{A'\to A})$ can also be expressed as half the resource-theoretic $\varepsilon$-smoothened conditional hypothesis-testing and max free energies,
\begin{align}
\mathrm{Dist}^{(1,\varepsilon)}(\mathcal{N}_{A'B'\to AB},\T^\beta_{A'\to A})
&=\frac{\beta}{2\ln 2}F^{\beta,\varepsilon^2}_H[A|B]_{\n},\\
    \mathrm{Cost}^{(1,\varepsilon)}(\mathcal{N}_{A'B'\to AB},\T^\beta_{A'\to A})& =\frac{\beta}{2\ln 2}F^{\beta,\varepsilon}_\infty[A|B]_{\n},
\end{align}
where 
\begin{align}
    F^{\beta,\varepsilon}_H[A|B]_{\n} & :=\beta^{-1}\ln 2\inf_{\Q\in\Ch({B',B})}D_{\mathrm{H}}^{\varepsilon}[\n_{A'B'\to AB}\Vert\T^\beta_{A'\to A}\otimes \Q_{B'\to B}],\nonumber\\
    F^{\beta,\varepsilon}_\infty[A|B]_{\n}&:=\beta^{-1}\ln 2\inf_{\mathcal{Q}\in\Ch({B',B})}D_{\infty}^{\varepsilon}[\n_{A'B'\to AB}\Vert\T^\beta_{A'\to A}\otimes\mathcal{\mathcal{Q}}_{B'\to B}].
\end{align}

We note that \autoref{th:dist-for-cost} implies that the resource theory of athermality for arbitrary single-input single-output quantum channels is asymptotically reversible under the action of Gibbs-preserving superchannels. The proof arguments in \cite{BGD25,BGD25b} follow for arbitrary channel even when the input and output system dimensions are not the same. Another direct consequence of \autoref{th:dist-for-cost} is the following observation.
\begin{corollary}\label{cor:asym-dist-lb}
For a bipartite quantum channel $\n_{A'B'\to AB}$,
\begin{equation}\label{eq:asym-dist-lb}
    {\rm Dist}^{(\infty,0)}(\n_{A'B'\to AB},\T^\beta_{A'\to A})\leq \frac{1}{2}\inf_{\Q\in\Ch(B',B)}D^{\rm reg}[\n_{A'B'\to AB}\|\T^\beta_{A'\to A}\otimes\Q_{B'\to B}],
\end{equation}
where $D^{\rm reg}[\cdot\|\cdot]$ is defined with channel relative entropy $D[\cdot\|\cdot]$ in Eq.~\eqref{eq:reg-gen-div}.
\end{corollary}
\begin{proof}
    It follows from the definition and \autoref{th:dist-for-cost} that
    \begin{align}
       2~{\rm Dist}^{(\infty,0)}(\mathcal{N},\mathcal{T}^\beta) &= \lim_{\varepsilon\to 0^+}\lim_{n\to\infty}\frac{1}{n}\inf_{\Q\in\Ch({B'}^n,B^n)}D^{\varepsilon^2}_H[\n^{\otimes n}\|{\mathcal{T}^\beta}^{\otimes n}_{{A'}^n\to A^n}\otimes \Q]\\
         &\leq \lim_{\varepsilon\to 0^+}\lim_{n\to\infty}\frac{1}{n}\inf_{\Q\in\Ch({B'},B)}D^{\varepsilon^2}_H[\n^{\otimes n}\|{\mathcal{T}^\beta}^{\otimes n}_{{A'}\to A}\otimes \Q^{\otimes n}].
    \end{align}
    We know that $\lim_{\varepsilon\to 0^+}\lim_{n\to\infty}D^{\varepsilon}_H[\n^{\otimes n}\|\m^{\otimes n}]=D^{\rm reg}[\n\|\m]$~\cite{WW19b}. Therefore, for an arbitrary $\Q\in\Ch(B',B)$,
    \begin{align}
     \lim_{\varepsilon\to 0^+}\lim_{n\to\infty}  \frac{1}{n} D^{\varepsilon^2}_H[\n^{\otimes n}\|{{\mathcal{T}^\beta}}^{\otimes n}_{{A'}\to A}\otimes \Q^{\otimes n}]=D^{\rm reg}[\n\|{\mathcal{T}^\beta}\otimes\Q],
    \end{align}
    and taking infimum over $\Q\in\Ch(B',B)$ on both sides, we arrive at the inequality~\eqref{eq:asym-dist-lb}.
\end{proof}
We make an important remark here that the protocols for the conditional athermality distillation yield and formation cost of a quantum channel $\n$ can be laced with adaptive strategies~\cite{CMW16,Das19}, where adaptive channels are conditional Gibbs-preserving~\cite{JGW25}, for many uses of $\n$. The non-adaptive strategy is special case of adaptive strategy. However, based on the recent results on the limitations on channel discrimination tasks~\cite{WW19b,FFRS20,BKSD23}, we can claim that $\frac{1}{2}\inf_{\Q\in\Ch(B',B)}D^{\rm reg}[\n_{A'B'\to AB}\|\T^\beta_{A'\to A}\otimes\Q_{B'\to B}]$ is an upper bound also for the asymptotic rate of the conditional athermality distillation assisted with CGPSs using adaptive strategy. One can also consider restricted class of free operations instead of CGPSs. The distillation yield won't be larger under restricted class of free operations and the formation cost won't be lesser under restricted class of free operations, when compared to the resource theory under CGPSs.

{\it Distillation yield and formation cost relation}: Let us recall the relation between $D^{\varepsilon}_H(\rho\|\sigma)$ and $D_{\infty}(\rho\|\sigma)$ for $\rho\in\St(A)$, $\sigma\in{\rm L}_+(A)$, and $\varepsilon\in(0,1)$,
\begin{equation}
    D^{\varepsilon}_H(\rho\|\sigma)\leq D_{\infty}(\rho\|\sigma)+\log\left(\frac{1}{1-\varepsilon}\right).
\end{equation}
Utilizing this for channel divergences based on the definitions, for quantum channels $\n_{A'B'\to AB}$, $\T^{\beta}_{A'\to A}$, and $\Q_{B'\to B}$, and $\varepsilon\in(0,1)$, we have
\begin{equation}
     D^{\varepsilon}_H[\n\|\T^\beta\otimes\Q]\leq D_{\infty}[\n\|\T^\beta\otimes\Q]+\log\left(\frac{1}{1-\varepsilon}\right),
\end{equation}
implying that asymptotically the distillation yield is never greater than the formation cost,
\begin{align}
    \mathrm{Dist}^{(1,\varepsilon)}(\mathcal{N}_{A'B'\to AB},\T^\beta_{A'\to A}) &\leq \mathrm{Cost}^{(1,\delta=0)}(\mathcal{N}_{A'B'\to AB},\T^\beta_{A'\to A})+\frac{1}{2}\log\left(\frac{1}{1-\varepsilon^2}\right),\\
      \mathrm{Dist}^{(\infty,0)}(\mathcal{N}_{A'B'\to AB},\T^\beta_{A'\to A})&\leq     \mathrm{Cost}^{(\infty,0)}(\mathcal{N}_{A'B'\to AB},\T^\beta_{A'\to A}).\label{eq:asym-dist-cost-ineq}
\end{align}
In a proper resource theory, the resource formation capacity of an object is never lesser than the resource distillation capacity of that object when the allowed transformations are free operations. This is because if resource distillation capacity of an object was strictly more than its resource formation capacity under free operations, one could first distill resource from that object and then would form the object while gaining standard resource units for free and all this under free operations. Free operations in a proper resource theory are those that cannot generate resource units, thereby implying inequality~\eqref{eq:asym-dist-cost-ineq}, i.e., the conditional athermality distillation capacity of a channel is always lesser than its conditional formation capacity. 

We remark that the zero-error one-shot conditional athermality distillation and formation rates of a bipartite quantum channel $\mathcal{N}_{A'B'\to AB}$ that is tensor-product of two local quantum channels, $\n=\mathcal{P}_{A'\to A}\otimes\Q_{B'\to B}$, are equal to the zero-error one-shot athermality distillation and formation rates of a quantum channel $\mathcal{P}$, respectively. This follows from \autoref{th:dist-for-cost} and \cite[Theorem~2]{DGP24}.

If we consider $B',B$ in a bipartite quantum channel $\n_{A'B'\to AB}$ to be trivial ($B'\simeq B\simeq \mathbb{C}$), then $\n$ is a single-input, single-output quantum channel from $A'\to A$. Conditional athermal resource channels $(\n_{A'B'\to AB},\mathcal{T}^\beta_{A'\to A}\otimes\Q_{B'\to B})$ and $(\id_{A'\to A}\otimes\mathcal{Q}'_{B'\to B},\mathcal{R}^\pi_{A'\to A}\otimes\Q'_{B'\to B})$ reduce to athermal resource channels $(\n_{A'\to A},\T^\beta)$ and $(\id_{A'\to A},\mathcal{R}^\pi_{B'\to B})$, respectively, for trivial $B',B$. This leads to the observation that the resource theory of conditional athermality of bipartite quantum channels provides a unifying framework to the resource theory of athermality of quantum channels~\cite{BGD25,BGD25b}, the resource theory of purity of quantum channels~\cite{YZGZ20,BGD25}, the resource theory of athermality of quantum states~\cite{WW19,JGW25}, and the resource theory of purity of quantum states~\cite{HHO03,WY16}, when the set of free operations contains all free operations (see~\cite[Section IV.A.]{BGD25} for relevant discussions).

The causal structure of a bipartite quantum channel $\n_{A'B'\to AB}$ in general can be more intricate than a tensor-product channel $\mathcal{P}_{A'\to A}\otimes\Q_{B'\to B}$. There are several questions regarding the kind of causal structure a bipartite quantum channel $\n_{A'B'\to AB}$ has: signaling or no-signaling from $A'\to B$, separability preserving, entangling, entanglement-breaking, etc. A bipartite channel $\mathcal{P}_{A'\to A}\otimes\Q_{B'\to B}$ is no-signaling from $A'\to B$ as well as separability preserving (non-entangling). If we consider Ash holds $A',A$ and Bob holds $B',B$, then a channel $\mathcal{P}_{A'\to A}\otimes\Q_{B'\to B}$ represents quantum channels locally acted by Ash and Bob on their respective systems. The conditional athermality formation rate of a quantum channel $\n_{A'B'\to AB}$ allows us to estimate the amount of conditional athermality $\id_m\otimes\Q'$ required to form the channel $\n$ which could have intricate causal structure. We discuss these aspects in coming sections in detail for quantum channels $\mathcal{N}_{A'B'\to AB}$ whose $\widehat{H}_{A}\propto \mathbbm{1}_A$.

\section{Conditional channel min-entropy}\label{section:min-entropy}
In this section, we prove some of the main properties of conditional min-entropy, such as asymptotic equipartition property for tele-covariant channels, continuity, dimensional bounds, etc. Asymptotic equipartition property of conditional min-entropy will be applied to determine the asymptotic rates of the conditional purity formation cost of tele-covariant bipartite quantum channels. 

The sandwiched R\'enyi conditional entropy $S_{\alpha}[A|B]_{\n}$ of an arbitrary quantum channel $\n_{A'B'\to AB}$, for $\alpha\in[\frac{1}{2},\infty]$, is defined as~\cite{DGP24}
\begin{equation}
    S_{\alpha}[A|B]_{\n}:=-\inf_{\Q\in\Ch(B',B)}D_{\alpha}[\mathcal{N}_{A'\to A}\Vert\R_{A'\to A}\otimes\Q_{B'\to B}],
\end{equation}
where $\R_{A'\to A}(\cdot):=\tr(\cdot)\mathbbm{1}_A$, $\lim_{\alpha\to 1}S_{\alpha}[A|B]_{\n}=S[A|B]_{\n}$ is von Neumann conditional entropy, and $\lim_{\alpha\to\infty}S_{\alpha}[A|B]_{\n}=S_{\infty}[A|B]_{\n}$ is the conditional min-entropy. The sandwiched conditional entropy $S_{\alpha}[A|B]_{\n}$, for $\alpha\in[\frac{1}{2},\infty]$, of a bipartite quantum channel $\n_{A'B'\to AB}$ satisfies the following properties~\cite{DGP24}:
\begin{enumerate}
    \item It is nondecreasing under the actions actions of an arbitrary conditional uniformity preserving superchannels $\Theta^{\pi}$, 
    \begin{equation}
        S_{\alpha}[A|B]_{\n}\geq S_{\alpha}[C|D]_{\Theta^{\pi}(\n)},
    \end{equation}
 for $\Theta^{\pi}\in\SCh((A'B',AB),(C'D',CD))$ such that $ \Theta^{\pi}(\mathcal{R}^\pi_{A'\to A}\otimes\Q_{B'\to B})=\mathcal{R}^\pi_{C'\to C}\otimes\Q'_{D'\to D}$ (see~Eq.~\eqref{eq:cups} and discussions around it).
 \item Conditioning doesn't increase entropy,
 \begin{equation}\label{eq:conditioning-prop}
     S_{\alpha}[A|B]_{\n}\leq S_{\alpha}[A]_{\n},
 \end{equation}
 where $S_{\alpha}[A]_{\n}:=S_{\alpha}[\tr_B\circ\n_{A'B'\to AB}]=-D_{\alpha}[\tr_B\circ\n\Vert\R_{A'B'\to A}]$. $\tr_B\circ\n_{A'B'\to AB}$ is a reduced channel of $\n$~\cite{DGP24} and $S_{\alpha}[\Q]:=-D_{\alpha}[\Q_{C'\to C}\Vert\R_{C'\to C}]$ is the sandwiched R\'enyi entropy of a quantum channel $\Q_{C'\to C}$~\cite{GW21,Yua19} (see also~\cite{SPSD25}). 
    \item For an arbitrary replacer channel $\mathcal{R}^{\rho}_{A'B'\to AB}$,
    \begin{equation}
        S_{\alpha}[A|B]_{\mathcal{R}^\rho}=S_{\alpha}(A|B)_{\rho}.
    \end{equation}
    \item For a tensor-product bipartite channel $\n=\mathcal{P}_{A'\to A}\otimes\Q_{B'\to B}$, the sandwiched R\'enyi conditional entropy reduces to the sandwiched R\'enyi entropy of ${\cal P}\in\Ch(A',A)$,
    \begin{equation}
        S_{\alpha}[A|B]_{\n}=S_{\alpha}[A]_{\mathcal{P}}.
    \end{equation}
    \item For $\alpha_1,\alpha_2\in[\frac{1}{2},\infty]$ and $\alpha_1\leq \alpha_2$, we have $S_{\alpha_1}[A|B]_\n\geq S_{\alpha_2}[A|B]_{\n}$. In particular, $S_{\infty}[A|B]_{\n}\leq S[A|B]_{\n}$.
\end{enumerate}

\begin{lemma}
    The sandwiched R\'enyi conditional entropy $S_{\alpha}[A|B]_{\n}$, for $\alpha\in[\frac{1}{2},\infty]$, of a biparite channel $\n_{A'B'\to AB}$ is maximum, $S_{\alpha}[A|B]_{\n}=\log|A|$ iff $\n$ is conditional uniformly mixing channel, $\n=\mathcal{R}^{\pi}_{A'\to A}\otimes\Q_{B'\to B}$ for arbitrary $\Q\in\Ch(B',B)$. Also, $S_{\alpha}[A|B]_{\n}\leq \log |A|$.
\end{lemma}
\begin{proof}
    \begin{align}
        S_{\alpha}[A|B]_{\n}&=-\inf_{\Q\in\Ch(B',B)}D_{\alpha}[\n_{A'B'\to AB}\|\R_{A'\to A}\otimes\Q_{B'\to B}]\\
        & =\log|A|-\inf_{\Q\in\Ch(B',B)}D_{\alpha}[\n_{A'B'\to AB}\|\mathcal{R}^{\pi}_{A'\to A}\otimes\Q_{B'\to B}].
    \end{align}
     For $\alpha\in[\frac{1}{2},\infty]$, we know that $D_{\alpha}[\mathcal{N}\|\mathcal{M}]\geq 0$ for any two arbitrary channels $\mathcal{N},\m$ and the inequality saturates ($=0$) iff $\n=\m$. We can then conclude that $ S_{\alpha}[A|B]_{\n}=\log |A|$ iff $\n=\mathcal{R}^{\pi}_{A'\to A}\otimes\Q_{B'\to B}$ for some $\Q\in\Ch(B',B)$ and also $S_{\alpha}[A|B]\leq \log|A|$.
\end{proof}

The Choi operator $\Gamma^{\n}$ of a bipartite channel $\mathcal{N}_{A'B'\to AB}$ is $\Gamma^\n_{R_AAR_BB}=\n(\Gamma_{R_AA'}\otimes\Gamma_{R_BB'})$, where $\Gamma_{R_AA'}$ denotes a maximally entangled operator. A maximally entangled operator $\Gamma_{R:A'B'}$ of $R:A'B'$ can be written as $\Gamma_{R:A'B'}=\Gamma_{R_AA'}\otimes\Gamma_{R_BB'}$, where $R=R_AR_B$ for $R\simeq A'B'$. We recall that for a bipartite quantum channel $\n_{A'B'\to AB}$~\cite{DGP24} (see also~\cite{CE16}),
\begin{align}
    S_{\infty}[A|B]_\n &:=-\inf_{\Q\in\Ch(B',B)}D_{\infty}[\n_{A'B'\to AB}\|\R_{A'\to A}\otimes\Q_{B'\to B}]\\
    &=-\inf_{\Q\in\Ch(B',B)}D_{\infty}(\Gamma^{\n}_{R_AAR_BB}\|\mathbbm{1}_{R_AA}\otimes\Gamma^{\Q}_{R_BB})\\
    &=-\inf_{\Q\in\Ch(B',B)}D_{\infty}(\Phi^{\n}_{R_AAR_BB}\|\mathbbm{1}_{R_AA}\otimes\Phi^{\Q}_{R_BB})-\log|A'|\\
    &=-\log \inf_{\Q\in\Ch(B',B)}\norm{(\Gamma^{\Q}_{R_BB})^{-\frac{1}{2}}\Gamma^\n_{R_AAR_BB}(\Gamma^{\Q}_{R_BB})^{-\frac{1}{2}}}_{\infty}\\
    &=-\log 
 \min\left\{\frac{\tr{M}}{|B'|}: 
 \Gamma^\n_{R_AAR_BB}\leq \mathbbm{1}_{R_AA}\otimes M, 0\leq M, \tr_B M=\frac{\tr M}{|B'|}\mathbbm{1}_{R_B}\right\},\label{eq:con-min-ent-sdp}
\end{align}
where $M\in{\rm L}_+(R_BB)$. Eq.~\eqref{eq:con-min-ent-sdp} provides semidefinite program (SDP) representation of the conditional channel min-entropy.

We first show that the conditional min-entropy of $\n_{A'B'\to AB}$ can be bounded by the conditional min-entropy of its Choi state $\Phi^{\n}_{R_AAR_BB}=\frac{1}{|A'||B'|}\Gamma^\n_{R_AAR_BB}$. This fact is crucial in making important observations like asymptotic equipartition property of the conditional min-entropy and proving that the conditional channel min-entropy is minimum iff the channel is a swap-like operation.

\begin{proposition}\label{thm:min-ent-bounds}
The conditional min-entropy $S_{\infty}[A|B]_{\mathcal{N}}$ of an arbitrary bipartite quantum channel $\mathcal{N}_{A'B'\to AB}$, with $|A'|<+\infty$, is bounded as
    \begin{align}
    S^{\downarrow}_{\infty}(R_AA|R_BB)_{\Phi^{\mathcal{N}}}-\log|A'|  \leq S_{\infty}[A|B]_{\mathcal{N}} & \leq  S_{\infty}(R_AA|R_BB)_{\Phi^{\mathcal{N}}}-\log|A'|,
    \end{align}
    where $\mathbf{S}^{\downarrow}(A|B)_{\rho}:=-{\bf D}(\rho_{AB}\Vert\mathbbm{1}_A\otimes\rho_B)$ and $\mathbf{S}(A|B)_{\rho}:=\inf_{\sigma\in\St(B)}{\bf D}(\rho_{AB}\|\mathbbm{1}_A\otimes\sigma_B)$.
\end{proposition}
\begin{proof}
We first prove the upper bound, using definition,
    \begin{align}
        - S_{\infty}[A|B]_{\mathcal{N}} 
        & = \inf_{\mathcal{Q}\in\mathrm{Ch}(B',B)}D_{\infty}(\Phi^{\mathcal{N}}_{R_AAR_BB}\Vert\pi_{R_A}\otimes\mathbbm{1}_{A}\otimes\Phi^{\mathcal{Q}}_{R_BB})\nonumber\\
          & = \inf_{\mathcal{Q}\in\mathrm{Ch}(B',B)}D_{\infty}(\Phi^{\mathcal{N}}_{R_AAR_BB}\Vert\mathbbm{1}_{R_A}\otimes\mathbbm{1}_{A}\otimes\Phi^{\mathcal{Q}}_{R_BB})+\log|A'|\nonumber\\
         &\geq   \inf_{\sigma\in{\mathrm St}(R_BB)}D_{\infty}(\Phi^{\mathcal{N}}_{R_AAR_BB}\Vert\mathbbm{1}_{R_A}\otimes\mathbbm{1}_{A}\otimes\sigma_{R_BB})+\log|A'|\nonumber\\
        & = - S_{\infty}(R_AA|R_BB)_{\Phi^{\mathcal{N}}}+\log|A'|.
    \end{align}
To prove the lower bound, consider a reduced channel $\mathcal{T}^{\mathcal{N}}_{B'\to B}$ of $\mathcal{N}_{A'B'\to AB}$ defined as
\begin{equation}
    \mathcal{T}^{\mathcal{N}}_{B'\to B}(\cdot):= \tr_A\circ\mathcal{N}_{A'B'\to AB}(\pi_{A'}\otimes\cdot).
\end{equation}
Notice that $\Phi^{\mathcal{T}^{\mathcal{N}}}_{R_BB}=\Phi^{\mathcal{N}}_{R_BB}$. Then,
\begin{align}
  - S_{\infty}[A|B]_{\mathcal{N}} &= \inf_{\mathcal{Q}\in\mathrm{Ch}(B',B)}D_{\infty}[\mathcal{N}_{A'B'\to AB}\Vert\R_{A'\to A}\otimes\mathcal{Q}_{B'\to B}]\nonumber\\
  & \leq D_{\infty}[\mathcal{N}_{A'B'\to AB}\Vert\R_{A'\to A}\otimes\mathcal{T}^{\mathcal{N}}_{B'\to B}]\nonumber\\
  & = D_{\infty}(\Phi^{\mathcal{N}}_{R_AAR_BB}\Vert\mathbbm{1}_{R_A}\otimes\mathbbm{1}_{A}\otimes\Phi^{\mathcal{N}}_{R_BB})+\log|A'|\nonumber\\
  & = - S^{\downarrow}_{\infty}(R_AA|R_BB)_{\Phi^{\mathcal{N}}}+\log|A'|.
\end{align}
That is,
\begin{equation}
      S^{\downarrow}_{\infty}(R_AA|R_BB)_{\Phi^{\mathcal{N}}} \leq S_{\infty}[A|B]_{\mathcal{N}}+ \log|A'|  \leq  S_{\infty}(R_AA|R_BB)_{\Phi^{\mathcal{N}}}.
\end{equation}
\end{proof}

We have the following dimensional bound on the conditional min-entropy, which will also holds true for the conditional von Neumann entropy of a bipartite quantum channel.
\begin{proposition}\label{prop:entropybound1}
The conditional min-entropy of a bipartite channel $\mathcal{N}_{A'B'\to AB}$ is lower bounded as
     \begin{equation}
\max\{ -\log|B|+ S_{\infty}[AB]_{\mathcal{N}}, -\log|A'|-\log\min\{{\rm rank}(\Phi^\mathcal{N}_{R_AA}), {\rm rank}(\Phi^\mathcal{N}_{R_BB})\} \} \leq S_{\infty}[A|B]_{\mathcal{N}}.
     \end{equation}
In terms of the dimension of the systems of $\mathcal{N}_{A'B'\to AB}$, we have
     \begin{equation}\label{eq:dim-bound}
-\log \min\{ |A'||B'||B|,|A'|^2|A|\} \leq S_{\infty}[A|B]_{\mathcal{N}} \leq \log |A|.
     \end{equation}
\end{proposition}
\begin{proof}
     The upper bound on $S_{\infty}[A|B]_{\mathcal{N}}$ follows directly from inequality~\eqref{eq:conditioning-prop}~\cite{DGP24},
    \begin{equation}
S_{\infty}[A|B]_{\mathcal{N}}\leq S[A|B]_{\mathcal{N}}\leq \log|A|.
    \end{equation}
    We may also observe that $S_{\infty}[A|B]_{\mathcal{N}}\leq  S_{\infty}(A|R_AR_BB)_{\Phi^{\mathcal{N}}}\leq \log|A|$.

  To derive the lower bounds on $S_{\infty}[A|B]_{\mathcal{N}}$, we notice that 
\begin{align}
    S_{\infty}[A|B]_{\mathcal{N}} & \geq   S^{\downarrow}_{\infty}(R_AA|R_BB)_{\Phi^{\mathcal{N}}}-\log|A'|  \\
    & \geq -\log\min\{{\rm rank}(\Phi^\n_{R_AA}), {\rm rank}(\Phi^\n_{R_BB})\}-\log|A'|\\
    &\geq -\log \min\{|A'|^2|A|,|A'||B'||B|\},
\end{align}
where we used $S^{\downarrow}(A|B)_{\rho}\geq -\log \min\{{\rm rank} (\rho_A), {\rm rank}(\rho_B)\}$~\cite[Lema 5.11]{T21} to arrive the second inequality.
\end{proof}

The theorem below states that the lower bound in inequality~\eqref{eq:dim-bound} is saturated if and only if the channel is a swap-like or maximally entangling operation, that is, if the Choi state of the channel is a maximally entangled state, $\Phi_{R_AA:R_BB}\in{\rm Ent}(R_AA;R_BB)$, of Schmidt rank $\min\{|A'||A|,|B'||B|\}$. We provide the detailed proof in Appendix~\ref{app:proof-SwapBound} where we show that the upper and lower bounds of $S_{\infty}[A|B]_\n$, as given in \autoref{thm:min-ent-bounds} coincide if and only if the channel is a maximally entangling operation.
\begin{theorem}\label{thm:SwapBound} 
The conditional min-entropy of a bipartite quantum channel $\n_{A'B'\to AB}$ achieves the minimum if and only if the channel is a maximally entangling unitary operation, i.e., $\n\in{\rm SWAP}(A;B)$,
\begin{equation}
    S_{\infty}[A|B]_{\n\in\mathrm{SWAP}(A;B)}=-\log \min\{|A'|^2|A|,|A'||B'||B|\}.
\end{equation}
\end{theorem}
The smooth conditional channel min-entropy for the swap-like operation is close to the extremal bound. In particular, we have
\begin{align}\label{eq:smooth-swap}
S_{\infty}^\varepsilon[A|B]_{\n\in\mathrm{SWAP}(A;B)}\le -\log \min\{|A'|^2|A|,|A'||B'||B|\} +\log\frac{1}{1-\varepsilon^{2}},
\end{align}
where we have used the fact that $S^\varepsilon_\infty[A|B]_\n\le S^\varepsilon_\infty(R_AA|R_BB)_{\Phi^\n}-\log|A'|$ (see inequality~\eqref{eq:smooth-channel-ineq}) and the following inequality from \cite[Lemma 5]{VDT+13} for a pure state $\psi_{AB}$,
\begin{align}
    S^{\varepsilon}_{\infty}(A|B)_\psi\le -S_{\infty}(A)_{\psi}+\log\frac{1}{1-\epsilon^2}.
\end{align}
Since the Choi state of a swap-like operation is maximally entangled, the inequality~\eqref{eq:smooth-swap} follows.

We now prove that the conditional min-entropy satisfies asymptotic continuity.
\begin{theorem}[Asymptotic continuity]\label{thm:conti}
Consider two arbitrary bipartite channels $\n_{A'B'\to AB}$ and $\m_{A'B'\to AB}$. For $\frac{1}{2}\norm{\n-\m}_\diamond\le\delta\in[0,1]$, the difference of their conditional min-entropy is bounded as
\begin{equation}
    \abs{S_{\infty}[A|B]_\n-S_{\infty}[A|B]_\m}\le \frac{1}{\ln 2}{|A|\min\{|A'||A|,|B'||B|\}\delta}.
\end{equation}
\end{theorem}
We provide the detailed proof in Appendix~\ref{app:proof-conti}. 

\textit{Tele-covariant channels}~\cite{DBW20,DBWH21}: Let ${\rm G}$ and ${\rm H}$ be finite groups. Let $g\mapsto U_{A'}^{(g)}$, $h\mapsto V_{B'}^{(h)}$ be unitary representations of $G$ and $H$ acting on input systems $A'$ and $B'$, respectively. Let $(g,h) \mapsto W_A^{(g,h)}$, $(g,h) \mapsto Z_B^{(g,h)}$ be unitary representations acting on output systems $A$ and $B$, respectively. A bipartite quantum channel $\mathcal{N}_{A'B' \to AB}$ is said to be bicovariant with respect to these representations
if, for all input states $\rho_{A'B'}$ and all $g \in {\rm G}$, $h \in {\rm H}$~\cite{GC99,STM11},
\begin{align}
\mathcal{N}_{A'B' \to AB}
\Big(
  \mathcal{U}_{A'}^{(g)} \otimes \mathcal{V}_{B'}^{(h)} (\rho_{A'B'})
\Big)
=
\mathcal{W}_A^{(g,h)} \otimes \mathcal{Z}_B^{(g,h)}
\left(\mathcal{N}_{A'B' \to AB}(\rho_{A'B'})\right),
\end{align}
where $\mathcal{U}^{(g)}(\cdot):=U^{(g)}(\cdot){U^{(g)}}^\dag$, $\mathcal{V}^{(g)}(\cdot):=V^{(g)}(\cdot){V^{(g)}}^\dag$, $\mathcal{W}^{(g,h)}(\cdot):=W^{(g,h)}(\cdot){W^{(g,h)}}^\dag$, $\mathcal{Z}^{(g,h)}(\cdot):=Z^{(g,h)}(\cdot){Z^{(g,h)}}^\dag$. The channel $\n$ is called tele-covariant if it is bicovariant such that the input representations form unitary one-designs, i.e.,
\begin{align}
\frac{1}{|G|}
\sum_{g \in G}
U_{A'}(g)\,\rho_{A'}\,U_{A'}(g)^\dagger
= \pi_{A'}, \quad
\frac{1}{|H|}
\sum_{h \in H}
V_{B'}(h)\,\rho_{B'}\,V_{B'}(h)^\dagger
= \pi_{B'},
\end{align}
for all states $\rho_{A'}$ and $\rho_{B'}$. We refer the readers to \cite[Section V.C.]{DBWH21} for the definition and relevant results on multipartite tele-covariant channels, which generalize the concept of covariance to multipartite quantum channels with respect to the input representations forming unitary one-designs.

In general, for an arbitrary quantum channel $\n_{A'B'\to AB}$, it holds that~\cite{DGP24}
\begin{equation}
    S[A|B]_{\n}\leq S(R_AA|R_BB)_{\Phi^\n}-\log|A'|,
\end{equation}
and the inequality is known to be saturated, $S[A|B]_{\n}= S(R_AA|R_BB)_{\Phi^\n}-\log|A'|$, if $\n$ is tele-covariant~\cite[Theorem 4]{DGP24}.

We now prove that the conditional channel min-entropy exhibits asymptotic equipartition property: on an average, the conditional min-entropy of the asymptotically many uses of a tele-covariant channel approaches its von Neumann conditional channel entropy.

\begin{theorem}[Asymptotic equipartition property]\label{th:equipartition}
    For $\varepsilon\in(0,1)$ and asymptotically many i.i.d.~uses of an bipartite channel $\n_{A'B'\to AB}$, the smoothed conditional channel min-entropy satisfies the following bounds in terms of the conditional von Neumann entropy of the Choi state $\Phi^{\n}$,
\begin{align}
    \lim_{\varepsilon\to 0^+} \lim_{n\to \infty}\frac{1}{n} S_{\infty}^\varepsilon [A^n|B^n]_{\mathcal{N}^{\otimes n}}  \leq  S(R_AA|R_BB)_{\Phi^{\mathcal{N}}}-\log|A'|.
\end{align}
If $\n_{A'B'\to AB}$ is tele-covariant, then
\begin{align}
    \lim_{\varepsilon\to 0^+} \lim_{n\to \infty}\frac{1}{n} S_{\infty}^\varepsilon [A^n|B^n]_{\mathcal{N}^{\otimes n}} & \leq S[A|B]_{\n},\\
    \lim_{n\to \infty}\frac{1}{n} S_{\infty}^\varepsilon [A^n|B^n]_{\mathcal{N}^{\otimes n}}&\geq S[A|B]_{\n}.
\end{align}
\end{theorem}
We provide the detailed proof in Appendix~\ref{app:proof_equipartition}. The following identity follows from the above theorem and \cite[Theorem 4]{DGP24}, for a tele-covariant quantum channel $\n_{A'B'\to AB}$,
\begin{equation}
     \lim_{\varepsilon\to 0^+} \lim_{n\to \infty}\frac{1}{n} S_{\infty}^\varepsilon [A^n|B^n]_{\mathcal{N}^{\otimes n}}  = S[A|B]_{\n}= S(R_AA|R_BB)_{\Phi^{\mathcal{N}}}-\log|A'|.
\end{equation}
The asymptotic equipartition property also holds for quantum channels $\n_{A'B'\to AB}$ that are no-signaling from $A'\to B$ (see \autoref{thm:ns-asym-con-ent}). We discuss some consequences of the asymptotic equipartition property of the conditional channel min-entropy for tele-covariant channels in the following section and for no-signaling channels in \autoref{sec:signaling}.

\section{Conditional purity distillation and formation capacity}\label{sec:asym-rate}
 If a quantum channel $\n_{A'B'\to AB}$ has output $A$ with fully degenerate Hamiltonian $\widehat{H}_{A}\propto \mathbbm{1}_A$ then the resource theory of conditional athermality reduces to the resource theory of conditional purity, given $\gamma^\beta_A=\pi_A$. The free operations are conditional uniformity preserving superchannels (CUPS) $\Theta^{\pi}$. A CUPS $\Theta^{\pi}\in\SCh((A'B',AB),(C'D',CD)))$ is defined as a superchannel that preserves conditional uniformly mixing channel $\mathcal{R}^{\pi}$, i.e., for an arbitrary $\Q\in\Ch(B',B)$ there exists some $\Q'\in\Ch(D',D)$ such that
\begin{equation}\label{eq:cups}
    \Theta^{\pi}(\mathcal{R}^\pi_{A'\to A}\otimes\Q_{B'\to B})=\mathcal{R}^\pi_{C'\to C}\otimes\Q'_{D'\to D}.
\end{equation}
We observe that a conditional athermal resource channel $(\n,\T^\beta\otimes\Q)$ becomes $(\n,\mathcal{R}^\pi\otimes\Q)$ when $\widehat{H}_A\propto \mathbbm{1}_A$, which is also a conditional purity resource channel. The conditional identity channel $\id_{A'\to A}\otimes\Q_{B'\to B}$ is conditional purity resource equivalent to all conditional unitary channels $\mathcal{U}_{A'\to A}\otimes\Q_{B'\to B}$ and can be deemed as the standard resource unit for the conditional purity. Following \autoref{th:dist-for-cost}, we have the following corollary that determines the one-shot conditional purity distillation and formation.

\begin{corollary}[One-shot conditional purity distillation and cost]\label{cor:one-shot-con-purity-rates}
    The one-shot conditional purity distillation yield $\mathrm{Dist}^{(1,\varepsilon)}(\mathcal{N},\mathcal{R}^\pi)$ and formation cost $\mathrm{Cost}^{(1,\varepsilon)}(\mathcal{N},\mathcal{R}^\pi)$ of a quantum channel $\n_{A'B'\to AB}$ are given as 
\begin{align}
\mathrm{Dist}^{(1,\varepsilon)}(\mathcal{N}_{A'B'\to AB},\mathcal{R}^\pi_{A'\to A})
&=\frac{1}{2}\inf_{\Q\in\Ch({B',B})}D_{\mathrm{H}}^{\varepsilon^2}[\n\Vert\mathcal{R}^\pi\otimes \Q]=\frac{1}{2}\left(\log |A|-S^{\varepsilon^2}_H[A|B]_{\n}\right),\\
\mathrm{Cost}^{(1,\varepsilon)}(\mathcal{N}_{A'B'\to AB},\mathcal{R}^\pi_{A'\to A})
&=\frac{1}{2}\inf_{\mathcal{Q}\in\Ch({B',B})}D_{\infty}^{\varepsilon}[\n\Vert\mathcal{R}^\pi\otimes\mathcal{Q}]=\frac{1}{2}\left(\log|A|-S^\varepsilon_{\infty}[A|B]_{\n}\right),\label{eq:one-shot-purity-cost}
\end{align}
where for a bipartite channel $\m_{A'B'\to AB}$, $S^\varepsilon_H[A|B]_{\m}$ is the hypothesis-testing conditional entropy obtained by taking the generalized channel divergence in Eq.~\eqref{eq:gen-con-ent} to be hypothesis-testing relative entropy and $S^{\varepsilon}_{\infty}[A|B]_{\m}$ is the $\varepsilon$-smoothed conditional min-entropy (Eq.~\eqref{eq:sm-min-ent}).
\end{corollary}
 Thus, we have operational meaning of $S^\varepsilon_H[A|B]_{\n}$ and $S^\varepsilon_\infty[A|B]_{\n}$ for a bipartite quantum channel $\n_{A'B'\to AB}$, as for $\varepsilon\in[0,1]$
\begin{align}
   S^{\varepsilon}_H[A|B]_{\n} &=\log|A|-2~\mathrm{Dist}^{(1,\sqrt{\varepsilon})}(\mathcal{N}_{A'B'\to AB},\mathcal{R}^\pi_{A'\to A}),\nonumber\\
   S^\varepsilon_{\infty}[A|B]_{\n}&=\log|A|-2~\mathrm{Cost}^{(1,\varepsilon)}(\mathcal{N}_{A'B'\to AB},\mathcal{R}^\pi_{A'\to A}).
\end{align}
We determine the exact optimal asymptotic rate (capacity) of the conditional purity formation cost of a tele-covariant quantum channel in \autoref{cor:asym-con-purity}.

We define the conditional purity distillation capacity ${\rm Dist}^{(\infty,0)}(\mathcal{N},\mathcal{R}^\pi)$ of a quantum channel $\n_{A'B'\to AB}$ as the optimal asymptotic rate of the conditional purity distillation yield of the channel $\n$.
\begin{proposition}[Conditional purity distillation capacity (Converse bound)]\label{prop:asy-dist-cap}
      The conditional purity distillation capacity of a bipartite quantum channel $\n_{A'B'\to AB}$ is upper bounded as 
    \begin{align}
    {\rm Dist}^{(\infty,0)}(\mathcal{N},\mathcal{R}^\pi):= \lim_{\varepsilon\to 0^+}\lim_{n\to \infty} \frac{1}{n}\mathrm{Dist}^{(1,\varepsilon)}(\mathcal{N}^{\otimes n}_{A'B'\to AB},{\mathcal{R}^{\pi}}_{{A'}^n\to A^n})&\leq \frac{1}{2}\left(\log|A|-S^{\rm reg}[A|B]_{\n}\right),
    \end{align}
    where $S^{\rm reg}[A|B]_{\m}$ is the regularized von Neumann conditional channel entropy of an arbitrary $\m\in\Ch({A'B',AB})$ based on the regularized relative entropy (Eq.~\eqref{eq:reg-gen-div})
    \begin{equation}
      S^{\rm reg}[A|B]_{\n}:=-\inf_{\Q\in\Ch(B',B)}D^{\rm reg}[\n\|\R\otimes\Q].
    \end{equation}
\end{proposition}
\begin{proof}
It follows from \autoref{cor:asym-dist-lb},
    \begin{align}
         {\rm Dist}^{(\infty,0)}(\mathcal{N},\mathcal{R}^\pi) & \leq \frac{1}{2}\inf_{\Q\in\Ch(B',B)}D^{\rm reg}[\n\|{\mathcal{R}^\pi}\otimes\Q]\\
         &= \frac{1}{2}\inf_{\Q\in\Ch(B',B)}(\log|A|+D^{\rm reg}[\n\|{\R}\otimes\Q])\\
         &=\frac{1}{2}(\log|A|-S^{\rm reg}[A|B]_{\n}).
    \end{align}
\end{proof}

 \begin{proposition}[Conditional purity distillation capacity (achievability)]\label{prop:asy-dist-lb}
      The conditional purity distillation capacity of a bipartite quantum channel $\n_{A'B'\to AB}$ is lower bounded as 
    \begin{align}
    {\rm Dist}^{(\infty,0)}(\mathcal{N},\mathcal{R}^\pi)\geq \frac{1}{2}\left[\log|A'||A|-{S}(R_AA|R_BB)_{\Phi^{\n}}\right].
    \end{align}
\end{proposition}
\begin{proof}
Consider $\varepsilon\in(0,1)$ and $\alpha\in(0,1)$, and let $f(\alpha,\varepsilon):=\frac{\alpha}{1-\alpha}\log\frac{1}{\varepsilon^2}$. For an arbitrary quantum channel $\Q\in\Ch({B'}^n,B^n)$,
    \begin{align}
        D^{\varepsilon^2}_H[\n^{\otimes n}\|\mathcal{R}^\pi_{{A'}^n\to A^n}\otimes \Q_{{B'}^n\to B^n}] & \geq D^{\varepsilon^2}_H({\Phi^\n}^{\otimes n}_{R_A^nA^nR_B^nB^n}\Vert{\Phi^{\mathcal{R}^{\pi}}}^{\otimes n}_{R_A^nA^n}\otimes\Phi^{\Q}_{R_B^nB^n})\\
        & \geq \overline{D}_{\alpha}({\Phi^\n}^{\otimes n}_{R_A^nA^nR_B^nB^n}\Vert{\Phi^{\mathcal{R}^{\pi}}}^{\otimes n}_{R_A^nA^n}\otimes\Phi^{\Q}_{R_B^nB^n})-f(\alpha,\varepsilon)\\
        & \geq \inf_{\sigma\in\St(R_B^nB^n)}\overline{D}_{\alpha}({\Phi^\n}^{\otimes n}_{R_A^nA^nR_B^nB^n}\Vert{\Phi^{\mathcal{R}^{\pi}}}^{\otimes n}_{R_A^nA^n}\otimes\sigma_{R_B^nB^n})-f(\alpha,\varepsilon)\\
        & = n(\log|A'||A|-\overline{S}_{\alpha}(R_AA|R_BB)_{\Phi^\n})-f(\alpha,\varepsilon),
    \end{align}
    where the last equality holds because the Petz-R\'enyi conditional entropy is additive for tensor-product states (see \cite[Lemma 7]{HT16}).
    Therefore, the following holds for all $\alpha\in(0,1)$,
    \begin{align}
          {\rm Dist}^{(\infty,0)}(\mathcal{N},\mathcal{R}^\pi)&:= \lim_{\varepsilon\to 0^+}\lim_{n\to \infty} \frac{1}{n}\mathrm{Dist}^{(1,\varepsilon)}(\mathcal{N}^{\otimes n}_{A'B'\to AB},{\mathcal{R}^{\pi}}_{{A'}^n\to A^n})\\
          &\geq \frac{1}{2}\left[\log|A'||A|-\overline{S}_{\alpha}(R_AA|R_BB)_{\Phi^\n}- \lim_{\varepsilon\to 0^+}\lim_{n\to \infty} \frac{1}{n}\frac{\alpha}{1-\alpha}\log\frac{1}{\varepsilon^2}\right]\\
          &=\frac{1}{2}\left[\log|A'||A|-\overline{S}_{\alpha}(R_AA|R_BB)_{\Phi^\n}\right].
    \end{align}
    Taking limit $\alpha\to 1^{-}$, we arrive at ${\rm Dist}^{(\infty,0)}(\mathcal{N},\mathcal{R}^\pi)\geq\frac{1}{2}\left[\log|A'||A|-{S}(R_AA|R_BB)_{\Phi^{\n}}\right]$.
\end{proof}

We define the conditional purity formation capacity ${\rm Cost}^{(\infty,0)}(\mathcal{N},\mathcal{R}^\pi)$ of a quantum channel $\n_{A'B'\to AB}$ as the optimal asymptotic rate of the conditional purity formation cost of the channel $\n$. An application of \autoref{th:equipartition} and \autoref{cor:one-shot-con-purity-rates} is the following corollary which determines ${\rm Cost}^{(\infty,0)}(\mathcal{N},\mathcal{R}^\pi)$ exactly when $\n$ is tele-covariant and a single-letter lower bound for arbitrary $\n$ in terms of its Choi state.
\begin{corollary}[Conditional purity formation capacity]\label{cor:asym-con-purity}
    The conditional purity formation capacity of a bipartite quantum channel $\n_{A'B'\to AB}$ is lower bounded as
    \begin{align}
    {\rm Cost}^{(\infty,0)}(\mathcal{N},\mathcal{R}^\pi):= \lim_{\varepsilon\to 0^+}\lim_{n\to \infty} \frac{1}{n}\mathrm{Cost}^{(1,\varepsilon)}(\mathcal{N}^{\otimes n}_{A'B'\to AB},{\mathcal{R}^{\pi}}_{{A'}^n\to A^n})&\geq \frac{1}{2}\left[\log|A'||A|-S(R_AA|R_BB)_{\Phi^\n}\right].
    \end{align}
\end{corollary}
We now state one of the main results that showcases the asymptotic reversibility of the conditional purity for tele-covariant quantum channels $\n_{A'B'\to AB}$.
\begin{theorem}[Purity capacity of tele-covariant channels]\label{thm:cap-tel}
        The conditional purity distillation and formation capacities of an arbitrary tele-covariant quantum channel $\n_{A'B'\to AB}$ are the same and exactly determined in terms of its conditional entropy,
    \begin{align}
{\rm Dist}^{(\infty,0)}(\mathcal{N},\mathcal{R}^\pi)={\rm Cost}^{(\infty,0)}(\mathcal{N},\mathcal{R}^\pi)&=\frac{1}{2}\left(\log|A|-S[A|B]_{\n}\right)=\frac{1}{2}\left[\log|A'||A|-S(R_AA|R_BB)_{\Phi^\n}\right].
\end{align}
That is, the resource theory of conditional purity for tele-covariant channels under the action of CUPSs is asymptotically reversible.
\end{theorem}
\begin{proof}
Let $\n_{A'B'\to AB}$ be a tele-covariant channel. Asymptotic equipartition property (\autoref{th:equipartition}) of the conditional channel min-entropy implies that
\begin{equation}
    \lim_{\varepsilon\to 0^+}\lim_{n\to \infty}\frac{1}{n}S^{\varepsilon}_{\infty}[A^n|B^n]_{\n}=S[A|B]_{\n}=S(R_AA|R_BB)_{\Phi^{\n}}-\log|A'|,
\end{equation}
and from \autoref{cor:one-shot-con-purity-rates},
\begin{equation}
    {\rm Cost}^{(\infty,0)}(\mathcal{N},\mathcal{R}^\pi)=\frac{1}{2}\left(\log|A|-S[A|B]_{\n}\right)=\frac{1}{2}\left[\log|A'||A|-S(R_AA|R_BB)_{\Phi^\n}\right].
\end{equation}
Whereas, from \autoref{prop:asy-dist-lb}, we have 
\begin{equation}
    {\rm Dist}^{(\infty,0)}(\mathcal{N},\mathcal{R}^\pi)\geq \frac{1}{2}(\log|A'||A|-S(R_AA|R_BB)_{\Phi^{\n}}).
\end{equation}
Observing that ${\rm Dist}^{(\infty,0)}(\mathcal{N},\mathcal{R}^\pi)\leq {\rm Cost}^{(\infty,0)}(\mathcal{N},\mathcal{R}^\pi)$ always holds, we conclude the proof.
\end{proof}

\textit{Superdense coding}: Let Ash and Bob be connected through a tele-covariant quantum channel $\n_{A'B'\to AB}$, where $A',A$ are held by Ash and $B',B$ are held by Bob and we assume $|A'|=|A|=d$. Ash and Bob use this channel to share its Choi state $\Phi^{\n}_{R_AAR_BB}$ among them, where $R_AA$ is with Ash and $R_BB$ is with Bob. Ash and Bob carry out superdense coding protocol~\cite{Hir01,BDL+04}, Ash performs the unitary operations $\U^\n_i$ on ${R_AA}$ with probability $p_i$, following which she can send her system $R_AA$ to Bob through a noiseless channel. The superdense coding capacity ${\rm Sdc}({\Phi^\n})$ captures the maximum amount of classical information that can be transmitted from Ash to Bob by sending quantum systems, utilizing pre-shared entanglement ($\Phi^\n$). For a two-qu$d$it quantum state $\rho_{AB}$, where $|A|=d$, its superdense coding capacity ${\rm Sdc}(A;B)_{\rho}=\log d-S(A|B)_{\rho}$~\cite{BDL+04}; this shows quantum advantage over classical protocol only when $S(A|B)_{\rho}<0$. Using this and \autoref{cor:asym-con-purity}, we find that the superdense coding capacity ${\rm Sdc}({\Phi^\n})$ of the Choi state $\Phi^{\n}_{R_AAR_BB}$ of a tele-covariant $\n$, with $|A'|=|A|$, is
\begin{align}
  {\rm Sdc}(R_AA;R_BB)_{{\Phi^\n}}   =&  \log|A'||A|-S(R_AA\vert R_BB)_{\Phi^\n}\nonumber\\
    = & \log|A|- S[A|B]_{\n}\nonumber\\
    =& 2~{\rm Dist}^{(\infty,0)}(\n,\mathcal{R}^{\pi})= 2~{\rm Cost}^{(\infty,0)}(\n,\mathcal{R}^{\pi}).\label{eq:sdc-dist-cost}
\end{align}
That is, the superdense coding capacity of the Choi state of a tele-covariant channel $\n$ is exactly the double of the conditional purity capacity of $\n$.

\section{Causal structure and conditional channel entropy}\label{sec:signaling}
We now examine relations between the causal structure of a bipartite channel and its conditional entropies, and in turn will provide insights for its conditional athermality formation cost.

\textit{Semicausal channels}: A quantum channel $\n_{A'B'\to AB}$ is called semicausal from $A'\to B$ (or no signaling from $A'\to B$) if operations on $A'$ cannot influence output $B$. Let ${\rm S}_{A'\not\to B}$, where ${\rm S}_{A'\not\to B}\subset\Ch(A'B',AB)$, denote the set of all bipartite channels $\n_{A'B'\to AB}$ that is semicausal (no-signaling) from $A'\to B$. For a bipartite quantum channel $\n\in\NS$, we have a reduced channel $\Q^{\n}_{B'\to B}$ derived from $\n$ such that~\cite{BGNP01,PHHH06}
\begin{equation}\label{eq:ns-con1}
    \tr_A\circ\n_{A'B'\to AB}=\tr_{A'}\otimes\Q^\n_{B'\to B},
\end{equation}
where $\Q^\n_{B'\to B}(\cdot)=\tr_A\circ\n(\rho_{A'}\otimes\cdot)$ holds true for all $\rho\in\St(A')$. All semicasual channels have the semilocalizable structure~\cite{ESW02}: each $\n\in\NS$ can be written as a composition of some quantum channels in $\Ch(B',BC)$ and $\Ch(A'C,A)$,
\begin{equation}
    \n_{A'B'\to AB}=(\widetilde{\n}_{A'C\to A}\otimes\id_{B})\circ(\id_{A'}\otimes\widetilde{\n}_{B'\to BC}).
\end{equation}
For all $\rho\in\St(A')$, we have $\tr_{C}\circ \widetilde{N}_{B'\to BC}=\tr_A\circ\n(\rho_{A'}\otimes\cdot)=\Q^{\n}_{B'\to B}$. We can write following equivalent condition for Eq.~\eqref{eq:ns-con1}: For an arbitrary $\n\in\NS$,
\begin{equation}\label{eq:ns-r-map}
    \mathcal{R}^\omega_{A\to A}\circ\n_{A'B'\to AB}=\mathcal{R}^\omega_{A'\to A}\otimes\Q^\n_{B'\to B}, 
\end{equation}
where $\mathcal{R}^\omega$ denotes a replacer channel.

{\it NS conditional channel entropy}: The sandwiched R\'enyi NS conditional channel entropy $S^{\not\to}_{\alpha}$ of a bipartite quantum channel $\n_{A'B'\to AB}$, for $\alpha\in[\frac{1}{2},\infty]$, is defined as~\cite{DGP24}
\begin{equation}\label{eq:sand-ns-con-ent}
    S^{\not\to}_{\alpha}[A|B]_\n:=-D_{\alpha}[\n_{A'B'\to AB}\|\R_{A\to A}\circ\n_{A'B'\to AB}].
\end{equation}
The von Neumann NS conditional entropy $S^{\not\to}[A|B]_{\n}$ (in $\lim\alpha\to1$ in Eq.~\eqref{eq:sand-ns-con-ent}) of a bipartite quantum channel $\n_{A'B'\to AB}$ satisfies~\cite{DGP24}
\begin{align}
     S^{\not\to}[A|B]_\n & =\inf_{\psi\in\St(RA'B')}S(A|RB)_{\n(\psi)},\label{eq:ns-ns}\\
     S^{\not\to}[A|B]_{\n} &\geq S[A|B]_{\n}.\label{eq:ns-s}
\end{align}
where it suffices to optimize over pure states $\psi_{RA'B'}$, $R\simeq A'B'$, in Eq.~\eqref{eq:ns-s}. The inequality~\eqref{eq:ns-s} is saturated, i.e., $S^{\not\to}[A|B]_\n =S[A|B]_\n$ iff $\n\in\NS$~\cite[Theorem 5]{DGP24}. $S^{\not\to}[A|B]_{\n}$ and related entropic quantities are also studied in the context of entropy accumulation theorem and cryptography in \cite{FKR+26,AT25} and conditional channel merging in \cite{DGP24}. 

\begin{proposition}\label{prop:asy-form-lb}
     For an arbitrary quantum channel $\n_{A'B'\to AB}$ its conditional purity distillation and formation capacities under CUPSs are bounded from below as
    \begin{align}
       \frac{1}{2}\left(\log|A|-S^{\not\to}[A|B]_{\n}\right) \leq   {\rm Dist}^{(\infty,0)}(\n,\mathcal{R}^{\pi})\leq {\rm Cost}^{(\infty,0)}(\n,\mathcal{R}^{\pi}).
    \end{align}
\end{proposition}
\begin{proof}
To prove the lower bound, consider $\varepsilon\in(0,1)$ and $\alpha\in(0,1)$, and let $f(\alpha,\varepsilon):=\frac{\alpha}{1-\alpha}\log\frac{1}{\varepsilon^2}$. For an arbitrary quantum channel $\Q\in\Ch({B'}^n,B^n)$,
    \begin{align}
        D^{\varepsilon^2}_H[\n^{\otimes n}\|\mathcal{R}^\pi_{{A'}^n\to A^n}\otimes \Q_{{B'}^n\to B^n}] & \geq \sup_{\psi\in\St(RA'B')} D^{\varepsilon^2}_H((\n(\psi_{RA'B'}))^{\otimes n}\Vert\pi_{A}^{\otimes n}\otimes{\Q(\psi_{RB'}^{\otimes n})})\\
        & \geq \sup_{\psi\in\St(RA'B')} \overline{D}_{\alpha}((\n(\psi_{RA'B'}))^{\otimes n}\Vert\pi_{A}^{\otimes n}\otimes{\Q(\psi_{RB'}^{\otimes n})})-f(\alpha,\varepsilon)\\
        & \geq \inf_{\sigma\in\St(R_B^nB^n)}\sup_{\psi\in\St(RA'B')} \overline{D}_{\alpha}((\n(\psi_{RA'B'}))^{\otimes n}\Vert\pi_{A}^{\otimes n}\otimes{\sigma_{R^nB^n}})-f(\alpha,\varepsilon)\\
        & = n\left(\log|A|-\inf_{\psi\in\St(RA'B')}\overline{S}_{\alpha}(A|RB)_{{\n(\psi)}}\right)-f(\alpha,\varepsilon),
    \end{align}
    where the last equality holds because the Petz-R\'enyi conditional entropy is additive for tensor-product states. Therefore, the following holds for all $\alpha\in(0,1)$,
    \begin{align}
         {\rm Dist}^{(\infty,0)}(\mathcal{N},\mathcal{R}^\pi)&=\lim_{\varepsilon\to 0^+}\lim_{n\to \infty} \frac{1}{n}\mathrm{Dist}^{(1,\varepsilon)}(\mathcal{N}^{\otimes n}_{A'B'\to AB},{\mathcal{R}^{\pi}}_{{A'}^n\to A^n})\\
          &\geq \frac{1}{2}\left(\log|A|-\inf_{\psi\in\St(RA'B')}\overline{S}_{\alpha}(A|RB)_{{\n(\psi)}}- \lim_{\varepsilon\to 0^+}\lim_{n\to \infty} \frac{1}{n}\frac{\alpha}{1-\alpha}\log\frac{1}{\varepsilon^2}\right)\\
          &=\frac{1}{2}\left(\log|A|-\inf_{\psi\in\St(RA'B')}\overline{S}_{\alpha}(A|RB)_{{\n(\psi)}}\right).
    \end{align}
    Taking limit $\alpha\to 1^{-}$, ${\rm Dist}^{(\infty,0)}(\mathcal{N},\mathcal{R}^\pi)\geq\frac{1}{2}\left(\log|A|-\inf_{\psi\in\St(RA'B')}{S}(A|RB)_{{\n(\psi)}}\right)=\frac{1}{2}(\log|A|-S^{\not\to}[A|B]_{\n})$. The other inequality follows from Eq.~\eqref{eq:asym-dist-cost-ineq}.
    \end{proof}
    
\begin{theorem}\label{thm:asym-ns}
For an arbitrary quantum channel $\n_{A'B'\to AB}$, $A'\simeq A$, that is no-signaling from $A'\to B$ ($\n\in\NS$), its conditional purity distillation and formation capacities are the same and exactly determined in terms of its conditional entropy,
    \begin{align}\label{eq:dist-ns-cap}
        {\rm Dist}^{(\infty,0)}(\n,\mathcal{R}^{\pi})={\rm Cost}^{(\infty,0)}(\n,\mathcal{R}^{\pi})=\frac{1}{2}(\log|A|-S[A|B]_{\n})&=\frac{1}{2}(\log|A|-S^{\not\to}[A|B]_{\n}).
    \end{align}
That is, the resource theory of conditional purity for no-signaling quantum channles under the action of CGPSs is asymptotically reversible.
\end{theorem}
We provide detailed proof of \autoref{thm:asym-ns} in Appendix~\ref{app:ns-asym}. We write the asymptotic equipartition property of the conditional channel min-entropy for no-signaling quantum channels as a theorem below, proof of which follows from the proof arguments of \autoref{thm:asym-ns}.
\begin{theorem}[Asymptotic equipartition property]\label{thm:ns-asym-con-ent}
    For $\varepsilon\in(0,1)$ and asymptotically many i.i.d.~uses of an arbitrary bipartite channel $\n_{A'B'\to AB}$, the smoothed conditional channel min-entropy satsifies the following bounds,
    \begin{align}
        \lim_{\varepsilon\to 0^+}\lim_{n\to\infty}\frac{1}{n}S^{\varepsilon}_{\infty}[A^n|B^n]_{\n^{\otimes n}}\leq S^{\not\to}[A|B]_{\n}=\inf_{\psi\in\St(RA'B')}S(A|RB)_{\n(\psi)}.
    \end{align}
    If $\n_{A'B'\to AB}$ is no-signaling, $\n\in\NS$, then
    \begin{align}
        \lim_{\varepsilon\to 0^+}\lim_{n\to\infty}\frac{1}{n}S^{\varepsilon}_{\infty}[A^n|B^n]_{\n^{\otimes n}}&\leq S[A|B]_{\n},\\
        \lim_{n\to\infty}\frac{1}{n}S^{\varepsilon}_{\infty}[A^n|B^n]_{\n^{\otimes n}}&\geq S[A|B]_{\n}.
    \end{align}
\end{theorem}
A consequence of \autoref{thm:asym-ns} for $\n\in\NS$ is that
\begin{equation}
     \lim_{\varepsilon\to 0^+}\lim_{n\to\infty}\frac{1}{n}S^{\varepsilon}_{\infty}[A^n|B^n]_{\n^{\otimes n}}=S[A|B]_{\n}=S^{\not\to}[A|B]_{\n}=\inf_{\psi\in\St(RA'B')}S(A|RB)_{\n(\psi)}.
\end{equation}
We now inspect the relations between the causal structure of a quantum channel and its conditional min-entropy.

\begin{proposition}\label{prop:ns-downmin-eq}
For an arbitrary bipartite channel $\mathcal{N}_{A'B'\to AB}$, its NS conditional min-entropy satisfies
\begin{equation}
    S^{\not\to}_{\infty}[A|B]_{\mathcal{N}}=S^{\downarrow}_{\infty}(A|R_AR_BB)_{\Phi^{\mathcal{N}}}.
\end{equation}
\end{proposition}
\begin{proof}
Note that for any channel $\n_{A'B'\to AB}$, we have $
\mathcal{R}^{\mathbbm{1}}_{A\to A}\circ \n(\cdot)=\mathbbm{1}_{A}\otimes \tr_A\circ\n(\cdot)$, and therefore, $\Phi^{\mathcal{R}_{A\to A}\circ\n}=\mathbbm{1}_A\otimes\tr_A\Phi^\n$. From the definition of NS conditional min-entropy, we get 
\begin{align}
S^{\not\to}_{\infty}[A|B]_{\mathcal{N}}&=-D_{\infty}[\mathcal{N}_{A'B'\to AB}\Vert\mathcal{R}^{\mathbbm{1}}_{A\to A}\circ\mathcal{N}_{A'B'\to AB}]\\
&=-D_{\infty}(\Phi^{\mathcal{N}}_{R_AAR_BB}\Vert\mathbbm{1}_A\otimes\Phi^{\mathcal{N}}_{R_AR_BB})\\
&=S_{\infty}^\downarrow(A|R_AR_BB)_{\Phi^{\n}}.
\end{align}
\end{proof}
In the following lemma we show that for a bipartite semicausal causal channel, its NS conditional min-entropy lower bounds its conditional min-entropy.
\begin{lemma}\label{lem:ns-ineq}
    For an arbitrary bipartite semicausal channel $\mathcal{N}_{A'B'\to AB}\in\NS$, we have
    \begin{align}
    S^{\downarrow}_{\infty}(R_AA|R_BB)_{\Phi^{\mathcal{N}}}-\log|A'|= S^{\not\to}_{\infty}[A|B]_{\mathcal{N}}\leq  S_{\infty}[A|B]_{\mathcal{N}}.
    \end{align}
\end{lemma}
\begin{proof}
For a bipartite channel $\mathcal{N}_{A'B'\to AB}\in{\NS}$, there exists a quantum channel $\Q^{\n}_{B'\to B}$ such that $\mathcal{R}^{\mathbbm{1}}_{A\to A}\circ\mathcal{N}=\mathcal{R}^{\mathbbm{1}}_{A'\to A}\otimes\Q^\n$. This implies $\Phi^{\Q^{\mathcal{N}}}_{R_BB}=\Phi^{\mathcal{N}}_{R_BB}$ and
 \begin{align}
       S^{\not\to}_{\infty}[A|B]_{\mathcal{N}}
        &=-D_{\infty}[\mathcal{N}_{A'B'\to AB}\Vert\mathcal{R}^{\mathbbm{1}}_{A\to A}\circ\mathcal{N}_{A'B'\to AB}]\\
        & =- D_{\infty}[\mathcal{N}_{A'B'\to AB}\Vert\mathcal{R}^{\mathbbm{1}}_{A'\to A}\otimes\Q^{\mathcal{N}}_{B'\to B}]\\
         & =- D_{\infty}(\Phi^{\mathcal{N}}\Vert\mathbbm{1}_{R_AA}\otimes \Phi^{\mathcal{N}}_{R_BB})-\log|A'|\\
        & = S^{\downarrow}_{\infty}(R_AA|R_BB)_{\Phi^{\mathcal{N}}}-\log|A'|\\
        &\leq S_{\infty}[A|B],
    \end{align}
 where the last inequality follows from \autoref{thm:min-ent-bounds}.
\end{proof}
A direct consequence of the above two results, \autoref{lem:ns-ineq} and \autoref{prop:ns-downmin-eq}, is the fundamental lower bound on the conditional min-entropy of a semicausal bipartite channel. 
\begin{proposition}\label{prop:nonsig-unit}
For an arbitrary semicausal bipartite channel, its conditional min-entropy is always lower bounded by the negative logarithm of the dimension of the non-connditioning output system. That is, for each $\n_{A'B'\to AB}\in \NS$, we always have $S_{\infty}[A|B]_\n\ge -\log|A|$.
\end{proposition}
\begin{proof}
From  \autoref{lem:ns-ineq} and \autoref{prop:ns-downmin-eq}, we have
\begin{align}
  S_{\infty}[A|B]_\n
  & \ge S_{\infty}^{\not\to}[A|B]_\n \\&=S_{\infty}^\downarrow(A|R_AR_BB)_{\Phi^\n}\\
  &\ge-\log\mathrm{rank}(\Phi^\n_A)\ge-\log|A|.
\end{align}
\end{proof}

Now we inspect the relations between entangling capabilities~\cite{Das19} of a bipartite quantum channel with its conditional entropies.

{\it Completely PPT-preserving channels}: A bipartite channel $\mathcal{N}_{A'B'\to AB}$ is called completely PPT-preserving~\cite{Rains99} if it completely preserves positivity of states under partial transposition, i.e., $\id_{R_AR_B}\otimes\mathcal{N}_{A'B'\to AB}(\rho_{R_AA'R_BB'})\in\mathrm{PPT}(R_AA;R_BB)$ for all $\rho_{R_AA'R_BB'}\in\mathrm{PPT}(R_AA';R_BB')$, where $R_A\simeq A', R_B\simeq B'$. A bipartite channel $\n_{A'B'\to AB}$ is completely PPT-preserving iff its Choi state $\Phi^{\mathcal{N}}\in\mathrm{PPT}(R_AA';R_BB')$~\cite{Rains01}. These channels do not generate distillable entanglement~\cite{Das19,DBW20} and appears as free operations in the resource theory of entanglement of bipartite quantum channels~\cite{BDW18,BDWW19,GS21}. In Appendix~\ref{app:proof-ppt}, we provide proof of the following proposition that provides lower bound on the conditional min-entropy a completely PPT-preserving channel $\n_{A'B'\to AB}$.
\begin{proposition}\label{prop:ppt}
    The conditional min-entropy of a completely PPT-preserving bipartite channel $\n_{A'B'\to AB}$ is always lower bounded by the negative of logarithm of its non-connditioning input system,
    \begin{align}
        S_{\infty}[A|B]_\n\ge -\log|A'|.
    \end{align}
\end{proposition}

{\it Entangling and separable channels}: A subclass of completely PPT-preserving channels are completely separability-preserving channels. A bipartite channel $\mathcal{N}_{A'B'\to AB}$ is called completely separability-preserving if its operator-sum representation is given by~\cite{BDF+99}
\begin{align}
    \n(\rho_{A'B'})=\sum_i (K^i_{A'}\otimes M^i_{B'})\rho_{A'B'}( K^i_{A'}\otimes M^i_{B'})^\dagger,
\end{align}
such that for each $i$, $K_{A'}^i\in\mathrm{L}(A',A)$ and $M_{B'}^i\in\mathrm{L}(B', B)$, and   $\sum_i (K^i_{A'}\otimes M^i_{B'})^\dagger( K^i_{A'}\otimes M^i_{B'})^\dagger=\mathbbm{1}_{A'B'}$. That is, Kraus operators of a completely separability-preserving channel $\n_{A'B'\to AB}$ can be written as tensor-product of operators on $A'\to A$ and $B'\to B$. These channels preserve the set of separable states, and in fact they form the largest set of completely entanglement non-generating channels~\cite{CG19}. The Choi state $\Phi^{\mathcal{N}}_{R_AA:R_BB}\in{\rm SEP}(R_AA;R_BB)$ of a bipartite channel $\mathcal{N}_{A'B'\to AB}$ iff the channel is completely separability-preserving. Any bipartite channel that is not completely separability-preserving is called entangling. The Choi state $\Phi^\n_{R_AAR_BB}$ of an entangling channel $\n_{A'B'\to AB}$ is entangled $\Phi^\n_{R_AAR_BB}\in{\rm Ent}(R_AA;R_BB)$. A quantum channel $\n_{A'B'\to AB}$ is said to be NPT entangling if its Choi state $\Phi^{\n}_{R_AAR_BB}\not\in{\rm PPT}(R_AA;R_BB)$. 

A related observation below follows directly from \autoref{prop:ppt}.
\begin{corollary}\label{prop:seperable}
    For any completely separability-preserving channel $\mathcal{N}_{A'B'\to AB}$, $S_{\infty}[A|B]_{\mathcal{N}}\geq -\log|A'|$ always holds. $S_{\infty}[A|B]_{\mathcal{N}}<-\log|A'|$ implies that the channel $\n_{A'B'\to AB}$ is NPT entangling.
\end{corollary}
A direct consequence of the above observation is that for any local quantum operations and classical communication (LOCC) channel~\cite{HHH99} $\mathcal{L}_{A'B'\to AB}$, its conditional min-entropy is lower bounded as $S_{\infty}[A|B]_{\mathcal{L}}\geq -\log|A'|$. This follows because ${\rm LOCC}\subset{\rm C\text{-}SEP\text{-}P}\subset{\rm C\text{-}PPT\text{-}P}$, where ${\rm LOCC}$ is the set of all LOCC channels, 
${\rm C\text{-}SEP\text{-}P}$ is the set of all completely separability-preserving channels, and ${\rm C\text{-}PPT\text{-}P}$ is the set of all completely PPT-preserving channels.
\begin{corollary}\label{cor:ent-unit}
    Given a bipartite unitary channel $\U_{A'B'\to AB}$, $S[A|B]_\U<-\log|A'|$ if and only if $\U_{A'B'\to AB}$ is an entangling channel.
\end{corollary}
\begin{proof}
    Note that $S[A|B]_{\U}\le S(R_AA|R_BB)_{\Phi^\U}-\log|A'|$, and given a pure state $\psi$, $S(A|B)_\psi<0$ if and only if $\psi$ is entangled. Therefore, $S[A|B]_{\U}<-\log|A'|$ if and only if $\Phi^\U$ is entangled, equivalently, if and only if $\U$ is entangling.       
\end{proof}
\subsection{Examples: Bipartite unitary channels}\label{sec:unitaries}
In this section, we determine the conditional entropies either analytically or numerically for a broad class of bipartite unitaries, as well as for noisy unitary channels of practical interest in quantum information processing~\cite{Kni05,GC99,CDGP08,Das19,TW25,DLL26}.

Before we begin with our results, let us recall that for a pure state $\psi_{AB}$, its conditional min-entropy satisfies $S_{\infty}(A|B)_{\psi} =-2\log \tr(\sqrt{\psi_B})$~\cite{KRS09}, where $\psi_B:=\tr_B\psi_{AB}$. For a pure state $\psi_{AB}$, the optimal conditional state on $B$ for $S(A|B)_{\psi}$ is $(\tr\sqrt{\psi_B})^{-1}\sqrt{\psi_B}$. That is, for pure state $\psi_{AB}$,
      \begin{align}\label{eq:SminupPure}
              S_{\infty}(A|B)_{\psi}
              & =-\inf_{\sigma\in\St(B)}D_{\infty}(\psi_{AB}\|\mathbbm{1}_{A}\otimes\sigma_B)=- D_{\infty}\left(\psi_{AB}\,\middle\|\,\mathbbm{1}_{A}\otimes
(\tr\sqrt{\psi_B})^{-1}\sqrt{\psi_B}\right)\\
&= -\log \tr(\sqrt{\psi_B})-D_{\infty}\left(\psi_{AB}\,\middle\|\,\mathbbm{1}_{A}\otimes\sqrt{\psi_B}\right).
      \end{align}
The conditional min-entropy $S_{\infty}^\downarrow(A|B)_\psi$ of a pure state $\psi_{AB}$ can also be calculated from the duality relations~\cite{T21}: $S^{\downarrow}_{\infty}(A|B)_\psi=-S_{\infty}(A)_{\psi}=-\log\mathrm{rank}(\psi_A)$. For a pure state $\psi_{AB}$, its Schmidt rank is equal to the rank of the reduced states, $\mathrm{rank}(\psi_A)=\mathrm{rank}(\psi_B)$.  

{\it Bipartite unitary channels}: For an arbitrary unitary channel $\mathcal{U}_{A'B'\to AB}$ corresponding to a unitary operator $U_{A'B'\to AB}$, its Choi state $\Phi^{\U}$ is also a pure state, given by $\Phi^\U=\dket{U}\dbra{U}$, where $\dket{U}=\sum_i \mathbbm{1}_{R_B}\otimes U\ket{ii}$ is the vectorization of the operator $U$. The conditional min-entropy of its Choi state is given by 
\begin{equation}\label{eq:SminupUnit}
    S_{\infty}(R_AA|R_BB)_{\Phi^{\U}}=-2\log\tr\sqrt{\Phi^\U_{BR_B}}.
\end{equation}
Consider a class of unitary channels $\mathcal{U}_{AB\to AB}$ called the controlled unitary channels and denoted as $\mathcal{C}_\U$, with control on $A$ and unitary operations on $B$. The corresponding unitary operator is given by $C_U=\sum_j \op{j}_A\otimes {U_j}$, where $\{\ket{j}\}_j$ are orthonormal basis in $A$ and $\{U_j\}_j$ are unitary operators on $B$. When the operators $\{U_i\}_i$ form a complete set of mutually orthonormal unitary operators, we denote such channels as $\widehat{\mathcal{C}}_\U$. For such channels conditional min-entropy of the Choi state has a relatively simple form, as we show in the following. 
\begin{lemma}\label{lem:SmindownUnit}
    Given an arbitrary controlled unitary channel $\mathcal{C}_\U$, we have 
    \begin{equation}
        S^\downarrow_{\infty}(AR_A|BR_B)_{\Phi^{\mathcal{C}_\U}}=-\log \abs{\mathrm{span}\{\dket{U_j}\}}.
    \end{equation}
\end{lemma}
\begin{proof}
The Choi state of the controlled unitary defined above is given by 
\begin{equation}
    \Phi^{\mathcal{C}_\U}=\dfrac{1}{|A|}\Big[ \sum_j \op{jj}\otimes \Phi^{\U_j}+\sum_{i\neq k}\ket{ii}\bra{kk}\otimes \dket{U_i}\dbra{U_k}\Big].
\end{equation}
 The reduced state of $\Phi^{\mathcal{C}_U}$ is given by $\Phi^{\mathcal{C}_U}_{R_BB}=\dfrac{\sum_j\Phi^{\U_j}}{|A|}$.
Therefore,
\begin{align}
    S^\downarrow_{\infty}(AR_A|BR_B)_{\Phi^{\mathcal{C}_U}}&=-D_{\infty}(\Phi^{\mathcal{C}_U}\Vert  \mathbbm{1}_{R_AA}\otimes \dfrac{\sum_j \Phi^{\U_j}}{|A|})\\
    &=-\log\mathrm{rank}(\sum_j\Phi^{\U_j})\\
    &=-\log |\mathrm{span}\{\dket{U_j}\}|.
    \end{align}
\end{proof}
As an example, consider the case where $\{U_j\}_j$ is an irreducible set of unitary operators (they do not have a non-trivial invariant subspace), then $|\mathrm{span}\{\dket{U_j}\}|=|A|$. As a simple corollary of the above Lemma, if the control space is 2-dimensional with the controlled unitary operator $C_U=\op{0}_A\otimes \mathbbm{1}_B+\op{1}_A\otimes U_B$, the min-entropy is $S^\downarrow_{\infty}(AR_A|BR_B)_{\Phi^{\mathcal{C}_U}}=-1$ with the obvious assumption that $U_B\neq \mathbbm{1}_B$.
\begin{proposition}\label{thm:contr-unit}
Let $\mathcal{C}_{\mathcal{U}}$ be a controlled unitary channel,such that the corresponding unitary operators are $C_U=\sum_j \op{j}_A\otimes {U_j}$ where $\{\ket{j}\}_j$ are orthonormal basis in $A$ and $\{U_j\}_j$ are unitary operators on $B$. We have
\begin{align}
    S_{\infty}[A|B]_{\mathcal{C}_\U}\ge -2\log|A|.
\end{align}
The inequality saturates for $\widehat{\mathcal{C}}_{\mathcal{U}}$, i.e., when $\{\dket{U_j}\}_j$ is an orthonormal set, $S_{\infty}[A|B]_{\widehat{\mathcal{C}}_\U}= -2\log|A|$.
\end{proposition}
\begin{proof}
From \autoref{lem:SmindownUnit} and the bounds on channel min-entropy from \autoref{thm:min-ent-bounds}, we have $S_{\infty}[A|B]_{\mathcal{C}_{\mathcal{U}}}\ge-2\log|A|$ for all controlled unitary channels since $|\mathrm{span}\{\dket{U_j}\}|\le |A|$. The inverse relations follows for channels $\widehat{\mathcal{C}}_\U$ since $\{\dket{U_i}\}\_i$ form an orthonormal basis for system $A$.  From Equation~\eqref{eq:SminupUnit} we have
    \begin{align}
        S_{\infty}(AR_A|BR_B)_{\Phi^{\widehat{\mathcal{C}}_\U}}&=-2\log\tr\sqrt{\frac{\sum_i\Phi^{\U_i}}{|A|}} \\
        &=-\log|A|.
    \end{align}
Again using \autoref{thm:min-ent-bounds}, we get $S_{\infty}[A|B]_{\widehat{\mathcal{C}}_{\mathcal{U}}}=-2\log|A|$. This proves the proposition.
\end{proof}
This bound also holds for von Neumann conditional entropy of $\widehat{\mathcal{C}}_\U$ as we show below
\begin{corollary}\label{cor:contunit-von}
    Given the controlled unitary channel $\widehat{\mathcal{C}}_\U$,  the von Neumann conditional entropy is $S[A|B]_{\widehat{\mathcal{C}}_U}=-2\log|A|$.
\end{corollary}
\begin{proof}
    Note that $S[A|B]_{\widehat{\mathcal{C}}_\U}\le S(R_AA|R_BB)_{\Phi^{\widehat{\mathcal{C}}_\U}}-\log|A|$. Since $\Phi^{\widehat{\mathcal{C}}_\U}$ is a pure state, we have
    \begin{align}   S(R_AA|R_BB)_{\Phi^{\widehat{\mathcal{C}}_\U}}&=-S(\Phi^{\widehat{\mathcal{C}}_\U}_{R_AA})\\
        &=-S\Big(\frac{1}{|A|}\Pi_A\Big)\\
        &=-\log|A|.
    \end{align}
    Therefore, $S[A|B]_{\widehat{\mathcal{C}}_\U}\le -2\log|A|$. Since $S[A|B]_{\widehat{\mathcal{C}}_\U}\ge S_\infty[A|B]_{\widehat{\mathcal{C}}_\U}=-2\log|A|$, we have  $S[A|B]_{\widehat{\mathcal{C}}_U}=-2\log|A|$.
\end{proof}
\begin{corollary}
Given the controlled unitary channel $\widehat{\mathcal{C}}_\U$, we have $S_\infty[B|A]_{\widehat{\mathcal{C}}_\U}=-2\log|A|$.
\end{corollary}
\begin{proof}
Note that $\Phi^{\widehat{\mathcal{C}}_\U}_{R_AA}=\frac{\sum_{i=1}^{|A|} \op{i}}{|A|}$, and the following the steps in~\autoref{lem:SmindownUnit} we see that $S_{\infty}^{\downarrow}(R_BB|R_AA)_{\widehat{\mathcal{C}}_\U}=-\log|A|$. Similarly, we have $S_{\infty}(R_BB|R_AA)_{\widehat{\mathcal{C}}_\U}=-\log|A|$, which proves the assertion that $S_\infty[B|A]_{\widehat{\mathcal{C}}_\U}=-2\log|A|$. Using the arguments similar to \autoref{cor:contunit-von} we can also conclude that $S[A|B]_{\widehat{\mathcal{C}}_\U}=-2\log|A|$.
\end{proof}
As a result, for controlled unitary channels $\widehat{\mathcal{C}}_\U$, we see that the conditional entropy of $A$ conditioned on $B$ is equal to the conditional entropy of $B$ conditioned on $A$,
\begin{align}    S_\infty[A|B]_{\widehat{\mathcal{C}}_\U}=S_\infty[B|A]_{\widehat{\mathcal{C}}_\U},
\end{align}
indicating that such channels have similar signaling power in both directions.

The NS conditional min-entropy of a unitary channel $\U_{A'B'\to AB}$ is completely characterized by $|A|$. Using the unitality of $\U$ we have $\Phi^\U_{AB}=\dfrac{1}{|A'\Vert B'|}\mathbbm{1}_{AB}$, therefore $    \Phi^\U_{A}=\dfrac{|B|}{|A'\Vert B'|}\mathbbm{1}_A$. Since $|A'\Vert B'|=|A\Vert B|$, we have $\Phi^\U_{A}=\pi_A$. Using~\autoref{prop:ns-downmin-eq}, we have $S_{\infty}^{\not\to}[A|B]_\n=-\log|A|.$ 

For no-signaling channel $\n_{A'B'\to AB}\in\NS$, we have the inequality $S_{\infty}[A|B]_\n\ge S^{\not\to}_{\infty}[A|B]_\n$ (\autoref{lem:ns-ineq}). For no-signaling unitary channels $\U_{A'B'\to AB}\in \NS$, we have $S_{\infty}[A|B]_\U\ge -\log|A|$. If we restrict dimensional condition on the unitary channels $\U$ to $|A'|=|A|$, then that puts constrain also on $B$, $|B'|=|B|$. The unitary channels $\U\in \NS$ with $A'\simeq A$ are tensor-product of local unitaries, $\U=\mathcal{U}_{A'\to A}\otimes\mathcal{U}_{B'\to B}$, and
\begin{align}\label{eq:nsu-product}
    S_{\infty}[A|B]_{\mathcal{U}_A\otimes\mathcal{U}_B}=S_{\infty}^{\not\to}[A|B]_{\mathcal{U}_A\otimes\mathcal{U}_B}=S_{\infty}[A]_{\U}=-\log|A'|.
\end{align}
For all conditional isometry channels $\V_{A'\to A}\otimes\Q_{B'\to B}$, where $\V_{A'\to A}$ is an isometry and $\Q\in\Ch(B',B)$, we have $S_{\infty}[A|B]_{\V\otimes\Q}=S_{\infty}[A]_{\V}=-\log|A'|$. All entangling unitaries are capable of generating distillable entanglement~\cite{HHHH09,BDW18} and non-entangling unitaries are tensor-product of local unitaries which cannot generate any kind of entanglement.

When the dimensions of conditioning system is set to $|B'|=|B|=2$, we can show that the conditional min-entropy of the bipartite unitary always saturates its upper limit, and therefore, is completely determined by its Choi state.
\begin{proposition}\label{prop:two-qubit-un}
For an arbitrary unitary channel $\U_{A'B'\to AB}$ such that $|B'|=|B|=2$, we have
\begin{align}
    S_{\infty}[A|B]_\U &= S_{\infty}(R_AA|R_BB)_{\Phi^\U}-\log|A'|.
\end{align}    
\end{proposition}
\begin{proof}
From Equation~\eqref{eq:SminupPure}, we see that the optimal state for the relative min-entropy of the Choi state is given by the state $(\tr\sqrt{\Phi^\U_{R_BB}})^{-1}\sqrt{\Phi^U_{R_BB}}$.
\begin{align}
    S_{\infty}(R_AA|R_BB)_{\phi^\U}&=-\inf_{\sigma}D_{\infty}(\Phi^\U|\mathbbm{1}_{R_AA}\otimes\sigma)\\
    &=-D_{\infty}\big(\Phi^\U|\mathbbm{1}_{R_AA}\otimes(\tr\sqrt{\Phi^\U_{R_BB}})^{-1}\sqrt{\Phi^U_{R_BB}}\big).
\end{align}
Due to \autoref{lemma:unitaryroot} presented in Appendix~\ref{app:c4maps}, we see that this state is also a Choi state of a CPTP map when $|B'|=|B|=2$. Therefore, we have
\begin{align}
    \inf_{\m\in\mathrm{Ch}(B',B)}D_{\infty}(\Phi^\U\|\mathbbm{1}_{R_AA}\otimes\Phi^\m)&=\inf_{\sigma}D_{\infty}(\Phi^\U\|\mathbbm{1}_{R_AA}\otimes\sigma)\\
    &=S_{\infty}(R_AA|R_BB)_{\Phi^\U}.
\end{align} 
Since $    S_{\infty}[A|B]_\n=-\inf_{\m\in\mathrm{Ch}(B',B)}D_{\infty}(\Phi^\U\|\mathbbm{1}_{R_AA}\otimes\Phi^\m)-\log|A'|$, the above identities prove the proposition. 
\end{proof}

\begin{figure}[ht]
    \centering
   \includegraphics[height=0.5\columnwidth]{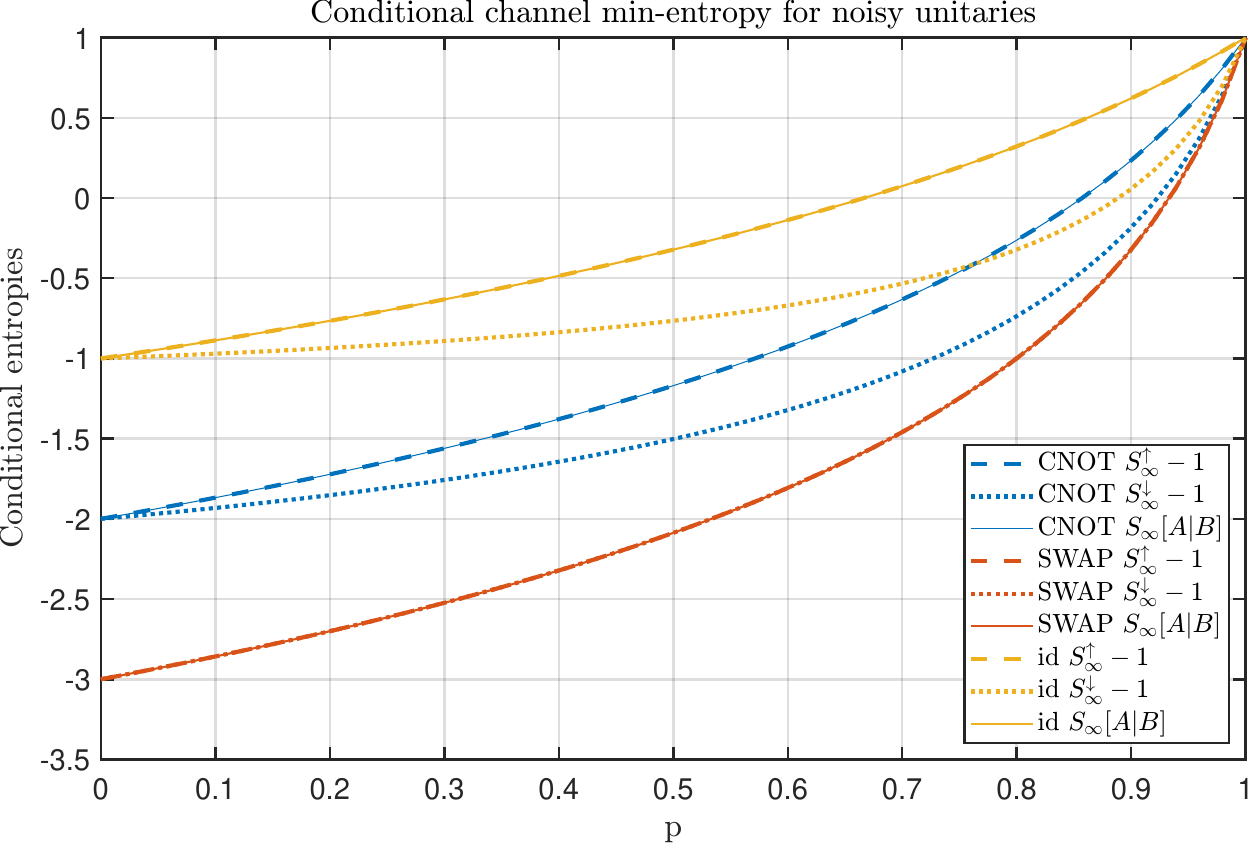}
    \caption{Plot showing the conditional min-entropy $S_{\infty}[A|B]_{\n}$ along with the upper and lower bounds $S_{\infty}(AR_A|BR_B)_{\Phi^\n}-1$ and $S^\downarrow_{\infty}(AR_A|BR_B)_{\Phi^\n}-1$ for two-qubit channels parametrized as $\n_p:=p\mathcal{R}^{\pi}+(1-p)\n$ for $p\in[0,1]$, where $\n$ is chosen to be $\mathrm{CNOT}$, $\mathrm{SWAP}$, $\mathrm{id}$. We see that the conditional channel min-entropy coincides with the upper bound. For the case of a noisy two-qubit swap channel, both the upper and lower bounds coincide for all values of $p$.}
    \label{fig:plots}
\end{figure}

Based on numerical results, we observe that the \autoref{prop:two-qubit-un} will hold for a broad class of channels that can be represented as noisy unitaries on $\mathbbm{C}^2\times\mathbbm{C}^2$: let $\n_p:=p\mathcal{R}^{\pi}+(1-p)\U$ for $p\in[0,1]$, then $S_\infty[A|B]_{\n_p}=S_{\infty} (R_AA|R_BB)_{\Phi^{\n_p}}-\log|A'|$. We plot these instances for some special two-qubit unitaries in \autoref{fig:plots}. We conjecture that \autoref{prop:two-qubit-un} can be extended to all unitaries with added white noise. We leave this as an open problem.

\section{Discussions}\label{section:discussions}
The fundamental information-theoretic notions of conditional entropy and conditional information for quantum states have enabled substantial conceptual and technological advances in quantum information science. These quantities also provide a unifying framework connecting computational, informational, and thermodynamic aspects of quantum systems. A particularly striking feature is that the conditional entropy of quantum states can be negative, which is impossible for classical joint probability distributions~\cite{CA97,HOW05}. This counterintuitive property is now well understood as a signature of entanglement~\cite{CA97}. Another cornerstone result is the strong subadditivity of quantum entropy~\cite{LR73}, which establishes that conditional entropy does not increase under further conditioning.

In this work, we determine the operational meanings of conditional channel entropies in the context of channel transformation under thermodynamic constraints. Too negative conditional channel min-entropy implies that the channel is signaling from non-conditioning input to conditioning output. It also reflects that the channel is able to generate distillable entanglement. We provide framework and tools to study energetics of quantum information processing. More broadly, our results highlight a fundamental connection between thermodynamic resources, causal structure, and the quantum correlations generating abilities in quantum processes.

There are several open questions emerging from our study here. One direction is to consider the conditional athermality and purity distillation and formation under physically motivated free operations that would be subset of CGPSs. It is also of interest to determine the exact optimal rates for other broad classes of bipartite channels and shed light on the asymptotic reversibility (or irreversibility) of the conditional athermality under CGPSs. We leave further study of the properties of conditional channel free energy for future work. The properties of the conditional channel entropies and the strong subadditivity of the channel entropy has potential to provide novel insight and tools for the study of quantum communication and computation, quantum error correction, open quantum systems, and many-body quantum systems, where quantum processes and their causal structure play critical role~\cite{BM01,ADHW09,CDP09,ZCZ+15,DW19a,DBWH21,BGG+25,BGD25b,DYS26,GMZ+25,BBC26}.

\begin{acknowledgments} 
The authors thank Pranab Sen, Uttam Singh, and Mark M Wilde for discussions. H.B. thanks the Ministry of Electronics and Information Technology (MeitY), Government of India for the Visvesvaraya Post-Doctoral Fellowship. S.D. acknowledges support from the Science and Engineering Research Board (now ANRF), Department of Science and Technology, Government of India, under Grant No. SRG/2023/000217 and MeitY, Government of India, under Grant No. 4(3)/2024-ITEA. S.D. also acknowledges partial support from the National Science Centre, Poland, grant Opus 25, 2023/49/B/ST2/02468.
\end{acknowledgments}

\begin{appendix}
\section*{Appendix}
In Appendix~\ref{app:gen-con-ent}, we briefly discuss some properties of the generalized conditional channel entropies. We discuss some results on the channel free energies of single-input, single-output quantum channels in Appendix~\ref{app:properties}. We provide detailed proofs of the results in the main content in Appendix~\ref{app:proofs}: The proof of \autoref{th:dist-for-cost} is in Appendix~\ref{app:proof-dist-for-cost}, the proof of~\autoref{prop:ppt} is in Appendix~\ref{app:proof-ppt}, the proof of~\autoref{thm:conti} is in Appendix~\ref{app:proof-conti}, the proof of~\autoref{th:equipartition} is in Appendix~\ref{app:proof_equipartition}, the proof of~\autoref{thm:asym-ns} is in Appendix~\ref{app:ns-asym}, and the proof of~\autoref{thm:SwapBound} is in Appendix~\ref{app:proof-SwapBound}. In Appendix~\ref{app:c4maps}, we discuss a relevant result on the Choi states of two-qubit unitary channels. In Appendix~\ref{app:freesup}, we discuss the form of conditional Gibbs-preserving superchannels.
 
\section{Generalized conditional channel entropies}\label{app:gen-con-ent}
Let us denote $\mathbb{D}(\cdot\|\cdot)$ denote the generalized divergence for families of Petz-R\'enyi relative entropy $\overline{D}_{\alpha}(\cdot\|\cdot)$ for $\alpha\in[0,2]$ and sandwiched R\'enyi relative entropy $D_{\alpha}(\cdot\|\cdot)$ for $\alpha\in[\frac{1}{2},\infty]$, where $\alpha=1$ means $\lim\alpha\to 1$, i.e., $\overline{D}_1(\cdot\|\cdot)=D_1(\cdot\|\cdot)=D(\cdot\|\cdot)$, and $\alpha=\infty$ means $\lim\alpha\to\infty$. For an arbitrary quantum channel $\n_{A'B'\to AB}$, the conditional channel entropy $\mathbb{S}[A|B]_{\n}$ is~\cite{DGP24} 
\begin{equation}
    \mathbb{S}[A|B]_{\n}=-\inf_{\Q\in\Ch(B',B)}\mathbb{D}[\n_{A'B'\to AB}\|\R_{A'\to A}\otimes\Q_{B'\to B}].
\end{equation}
and the NS conditional channel entropy $\mathbb{S}^{\not\to}[A|B]_{\n}$ is~\cite{DGP24}
\begin{equation}
    \mathbb{S}^{\not\to}[A|B]_{\n}:=\mathbb{D}[\n_{A'B'\to AB}\|\R_{A\to A}\circ\n_{A'B'\to AB}].
\end{equation}

\begin{lemma}\label{lem:gen-con-ent-bounds}
    For a quantum channel $\n_{A'B'\to AB}$,
    \begin{equation}
        \mathbb{S}[A|B]_{\n}\leq \mathbb{S}(R_AA|R_BB)_{\Phi^{\n}}-\log|A'|,
    \end{equation}
    where $\mathbb{S}^{\downarrow}(R_AA|R_BB)_{\Phi^{\n}}:=-\inf_{\sigma\in\St(R_BB)}\mathbb{D}(\Phi^\n_{R_AAR_BB}\|\mathbbm{1}_{R_AA}\otimes\sigma_{R_BB})$. Furthermore, if $\n_{A'B'\to AB}$ is tele-covariant, then 
    \begin{equation}
      \mathbb{S}^{\downarrow}(R_AA|R_BB)_{\Phi^{\n}}-\log|A'|  \leq \mathbb{S}[A|B]_{\n}
    \end{equation}
\end{lemma}
\begin{proof}
    The proof argument is similar to the proof of bounds in \autoref{thm:min-ent-bounds}.
    We arrive at the upper bound observing that
    \begin{align}
  - \bS[A|B]_{\mathcal{N}} &= \inf_{\mathcal{Q}\in\mathrm{Ch}(B',B)}\bD[\mathcal{N}_{A'B'\to AB}\Vert\mathcal{R}_{A'\to A}\otimes\mathcal{Q}_{B'\to B}]\nonumber\\
  & \geq \bD(\Phi^\n_{R_AAR_BB}\Vert\mathbbm{1}_{R_AA}\otimes\Phi^\Q_{R_BB})-\log|A'|\nonumber\\
  & \geq \inf_{\sigma\in\St(R_BB)}\bD(\Phi^\n_{R_AAR_BB}\Vert\mathbbm{1}_{R_AA}\otimes\Phi^\Q_{R_BB})-\log|A'|\\
  & = -\bS(R_AA|R_BB)_{\Phi^{\mathcal{N}}}+\log|A'|.
\end{align}

To arrive at the lower bound, we now consider that $\n$ is tele-covariant. Let a reduced channel $\mathcal{T}^{\mathcal{N}}_{B'\to B}$ of $\mathcal{N}_{A'B'\to AB}$ defined as
\begin{equation}
    \mathcal{T}^{\mathcal{N}}_{B'\to B}(\cdot):= \tr_A\circ\mathcal{N}_{A'B'\to AB}(\pi_{A'}\otimes\cdot~).
\end{equation}
Notice that $\Phi^{\mathcal{T}^{\mathcal{N}}}_{R_BB}=\Phi^{\mathcal{N}}_{R_BB}$. Then,
\begin{align}
  - \bS[A|B]_{\mathcal{N}} &= \inf_{\mathcal{Q}\in\mathrm{Ch}(B',B)}\bD[\mathcal{N}_{A'B'\to AB}\Vert\mathcal{R}^{\mathbbm{1}}_{A'\to A}\otimes\mathcal{Q}_{B'\to B}]\nonumber\\
  & \leq \bD[\mathcal{N}_{A'B'\to AB}\Vert\mathcal{R}^{\mathbbm{1}}_{A'\to A}\otimes\mathcal{T}^{\mathcal{N}}_{B'\to B}]\nonumber\\
  & = \bD(\Phi^{\mathcal{N}}_{R_AAR_BB}\Vert\mathbbm{1}_{R_A}\otimes\mathbbm{1}_{A}\otimes\Phi^{\mathcal{N}}_{R_BB})+\log|A'|\nonumber\\
  & = - \bS^{\downarrow}(R_AA|R_BB)_{\Phi^{\mathcal{N}}}+\log|A'|,
\end{align}
where the second equality holds because a maximally entangled state is an optimal input state for the generalized channel divergence between two jointly tele-covariant channels~\cite{LKDW18} (see also~\cite{DW19}).
\end{proof}

\begin{proposition}
    For a quantum channel $\n_{A'B'\to AB}$, the dimensional lower bound is achieved,
    \begin{equation}
        \bS[A|B]_{\n}=-\log\min\{|A'||B'||B|,|A'|^2|A|\},
    \end{equation}
    iff $\n$ is a maximally entangling (swap-like) operation.
\end{proposition}
\begin{proof}
    For arbitrary $\rho,\sigma\in\St(A)$, we have $\overline{D}_{\alpha_1}(\rho\|\sigma)\leq \overline{D}_{\alpha_2}(\rho\|\sigma)$ and ${D}_{\alpha_1}(\cdot\|\cdot)\leq D_{\alpha_2}(\rho\|\sigma)$ whenever $\alpha_1\leq \alpha_2$~\cite{T21}. Also, in limit $\alpha\to\infty$, $\overline{D}_{\alpha}(\rho\|\sigma)=D_{\infty}(\rho\|\sigma)$. For an arbitrary channel $\n_{A'B'\to AB}$,
    \begin{equation}
        S_{\infty}[A|B]_{\n}\leq \bS[A|B]_{\n}.
    \end{equation}
    
     Swap is a tele-covariant channel and its Choi state is $\Phi_{R_AB}\otimes\Phi_{R_BA}$. We also have $\bS^{\downarrow}(A|B)_{\rho}=-\log\min\{|A|,|B|\}=\bS(A|B)_{\rho}$ if and only if $\rho_{AB}$ is maximally entangled. From \autoref{thm:SwapBound} and \autoref{lem:gen-con-ent-bounds}, we conclude that $\bS[A|B]_{\n}=-\log\min\{|A'||B'||B|,|A'|^2|A|\}$ iff $\n$ is a swap operation ${\rm SWAP}_{A:B}$.
\end{proof}

We remark that the conditional channel min-entropy can, in general, lie strictly between the bounds in \autoref{thm:min-ent-bounds}; see \autoref{fig:plotRandCh} for an example.

\begin{figure}[h]
    \centering
    \includegraphics[width=0.8\linewidth]{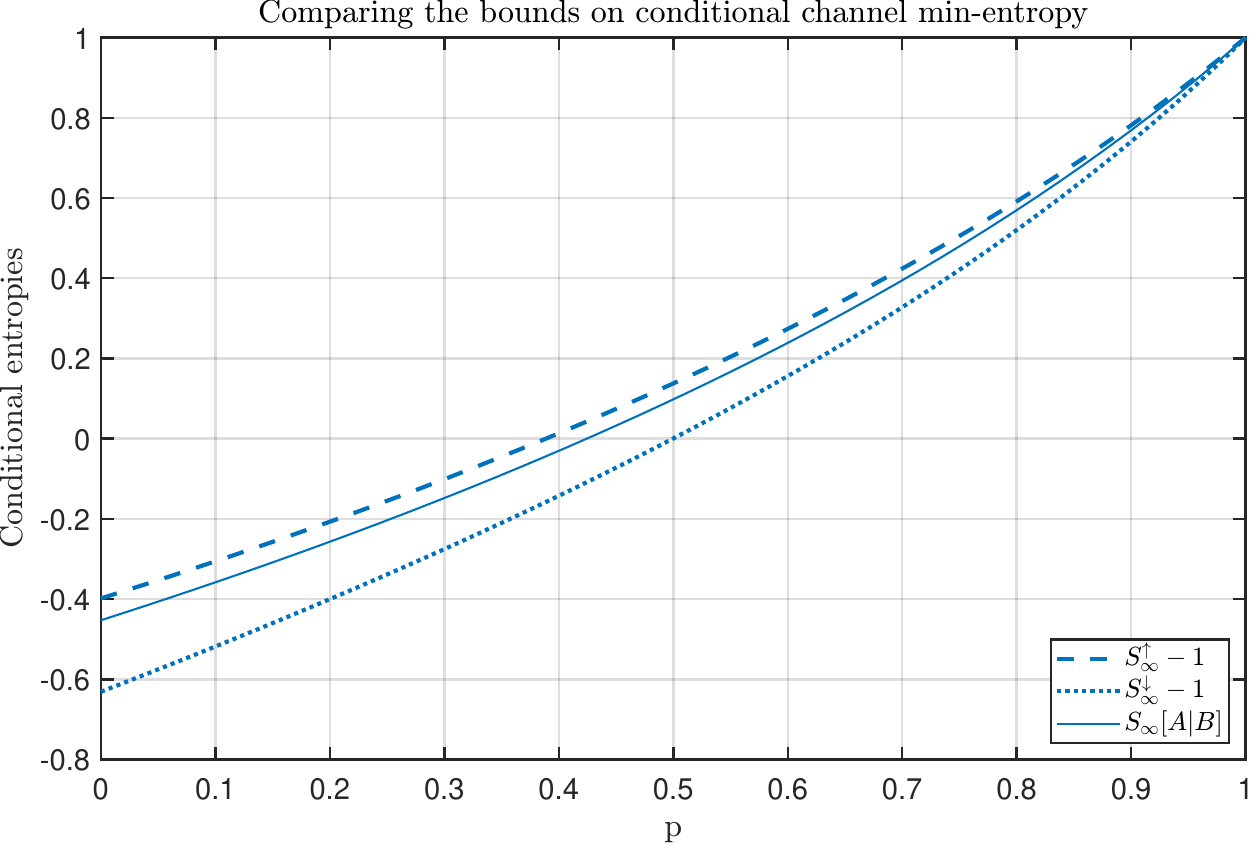}
     \caption{Plot showing the conditional min-entropy $S_{\infty}[A|B]_{\n}$ along with the upper and lower bounds $S{\infty}(AR_A|BR_B)_{\Phi^\n}-1$ and $S_{\infty}(AR_A|BR_B)_{\Phi^\n}-1$ for two-qubit channels parametrized as $\n_p:=p\mathcal{R}^{\pi}+(1-p)\n$ for $p\in[0,1]$, where $\n$ is a numerically generated random channel using Quantinf package~\cite{Cub10}. There is a clear separation between conditional min-entropy and its bounds in \autoref{thm:min-ent-bounds}.}
    \label{fig:plotRandCh}
\end{figure}

\section{Channel free energy}~\label{app:properties}
The channel free energy $F^\beta[{\cal P}]=\beta^{-1}D[\mathcal{P}\|\T^\beta]\ln 2$ of an arbitrary quantum channel $\mathcal{P}_{A'\to A}$ can be expressed as~\cite{BGD25,BGD25b}
 \begin{align}
        \frac{\beta}{\ln2} F^{\beta}[\mathcal{P}]=D[\mathcal{P}\|\mathcal{T}^{\beta}]&=\sup_{\psi\in\St(RA')}D(\mathcal{P}(\psi_{RA'})\|\mathcal{T}^{\beta}(\psi_{RA'}))\label{eq:lem1}\\
        &=\sup_{{\psi}\in\St(RA')}D(\mathcal{P}(\psi_{RA'})\|\psi_R\otimes\gamma^\beta_A)\\
        &=\sup_{\op{\psi}\in\St(RA')}\left[I(R;A)_{\mathcal{P}(\psi)}+D(\mathcal{P}(\psi_{A'})\|\gamma^\beta_A)\right],
    \end{align}
    where it suffices to optimize over pure states $\psi\in\St(RA')$ with $R\simeq A'$ in Eq.~\eqref{eq:lem1}.
    
\begin{lemma}[\cite{BGD25}]\label{lem:aap-free-opt}
The resource-theoretic thermal channel free energy $F^{\beta}[{\cal P}]:=\beta^{-1}D[\mathcal{P}\|\mathcal{T}^\beta]\ln 2$ of an arbitrary quantum channel ${\cal P}_{A'\to A}$ is maximum if it is an isometry channel. 
\end{lemma}
\begin{proof}
 Notice that it suffices to optimize over pure states $\psi_{RA'}=\op{\psi}_{RA'}$ with $R\simeq A'$ in the RHS of Eq.~\eqref{eq:lem1}. Due to joint-convexity (quasi-convexity suffices) of relative entropy between states, maximum value of the state relative entropy $D(\rho\Vert\sigma)$ for a fixed state $\sigma$ is achieved only if $\rho$ is a pure state~\cite{Rus02}. This implies that the maximum value for $D[\mathcal{P}\|\mathcal{T}^{\beta}]$ can be achieved only if the channel $\mathcal{P}$ is an isometry channel. 
\end{proof}

\begin{lemma}\label{lem:iso-genH}
  Consider an isometry channel $\V_{A'\to A}$ with $\V(\cdot):=V(\cdot)V^\dag$ and let $\beta$ be finite. The channel free energy of $\V_{A'\to A}$ is $F_{\infty}^\beta[\V]=\beta^{-1}\ln\tr[\Pi^{\mathcal{V}}_A(\gamma^{\beta}_A)^{-1}]$, where $\Pi^\V_A:=VV^\dag$ is a projector on $A$ and ${\rm rank}(\Pi^\V_{A})=|A'|$.
\end{lemma}
\begin{proof}
Given a pure state $\psi$, we have $D_\infty(\psi\Vert\sigma)=\log\bra{\psi}\sigma^{-1}\ket{\psi}$ for a positive operator $\sigma$. Using this, we have
\begin{align}
{\beta} F_{\infty}^\beta [\mathcal{V}]  &= D_{\infty}(\Phi^{\mathcal{V}} \|\Phi^\mathcal{T})~{\ln2} \\
 &= \ln \left(\frac{1}{|A'|} \sum_{j} 
(\bra{j} \otimes \bra{j} V^\dagger) 
\left( \frac{1}{|A'|} \mathbbm{1}_R\otimes \gamma^\beta_A \right)^{-1} \sum_{i} (\ket{i} \otimes V \ket{i}) \right) \\
&= \ln \left( \sum_{i,j} 
\delta_{ji}  \bra{j} V^\dagger (\gamma^\beta_A)^{-1} V \ket{i} \right)\\
&= \ln \left( \tr \left[ \gamma_\beta^{-1} VV^\dagger \right] \right).
\end{align}
\end{proof}
A direct consequence of the above lemma is the following corollary.
\begin{corollary}
    For an isometry channel $\V_{A'\to A}$ with $\widehat{H}_A\propto \mathbbm{1}_A$, we have $F^\beta[\V]=F^{\beta}_{\infty}[\mathcal{V}]=\beta^{-1}\ln(|A'||A|)$.
\end{corollary}
\begin{proof}
    For the trivial Hamiltonian, we have $(\gamma^\beta_A)^{-1}=\mathbbm{1}_A|A|$ and $ \tr \left[ (\gamma^\beta_A)^{-1}\Pi_A^\mathcal{V} \right] = \tr (|A| \Pi_A^\mathcal{V} ) = |A'||A|$ which proves $F^{\beta}_{\infty}[\mathcal{V}]=\beta^{-1}\ln(|A'||A|)$. We also know,
    \begin{align}
        F^\beta_{\infty}[\V]\geq F^\beta[\V]&\geq \beta^{-1}D[\Phi^\V\|\mathbbm{1}_R\otimes\pi_A]\ln 2\\ 
        &=\beta^{-1}\ln|A'||A|,
    \end{align}
    which concludes the proof.
\end{proof}
Our \autoref{lem:iso-genH} also provides an upper bound on the free energy of a quantum channel $\mathcal{P}_{A'\to A}$,
\begin{align}
    F^\beta[\mathcal{P}]\leq F^{\beta}_{\infty}[\mathcal{P}]\leq F^\beta_{\infty}[\mathcal{V}]\leq \beta^{-1}\ln\tr[(\gamma^\beta_A)^{-1}].
\end{align}
We note that the above bound can also be derived from \cite[Theorem A5]{LMB25} due to the relation between the the robustness of athermality and channel max-free energy~\cite{BGD25}.

\section{Proofs}\label{app:proofs}
\subsection{Proof of \autoref{th:dist-for-cost}}\label{app:proof-dist-for-cost}
\begin{theorem*}
    The one-shot conditional athermality distillation yield and formation cost of a quantum channel $\mathcal{N}_{A'B'\to AB}$, up to an error $\varepsilon\in[0,1]$, is given by
\begin{align}
\mathrm{Dist}^{(1,\varepsilon)}(\mathcal{N}_{A'B'\to AB},\T^\beta_{A'\to A})
&=\frac{1}{2}\inf_{\Q\in\Ch({B',B})}D_{\mathrm{H}}^{\varepsilon^2}[\n_{A'B'\to AB}\Vert\T^\beta_{A'\to A}\otimes \Q_{B'\to B}],\\
\mathrm{Cost}^{(1,\varepsilon)}(\mathcal{N}_{A'B'\to AB},\T^\beta_{A'\to A})
&=\frac{1}{2}\inf_{\mathcal{Q}\in\Ch({B',B})}D_{\infty}^{\varepsilon}[\n_{A'B'\to AB}\Vert\T^\beta_{A'\to A}\otimes\mathcal{\mathcal{Q}}_{B'\to B}].\label{eq:athermality-cost-rate}
\end{align}
\end{theorem*}

\begin{proof}
Before we begin the proof, let us denote the uniformly mixing channel $\rp_{C'\to C}$ as $\mathcal{R}^{\pi}_m$ whenever $|C'|=|C|=m$.  We can decompose a uniformly mixing channel $\mathcal{R}^\pi_m$ by taking a uniform mixture of Weyl unitary channels $\mathcal{W}^i(\cdot)=W^i(\cdot){W^i}^\dagger$,
\begin{align}
    \mathcal{R}^\pi_m(\cdot) &=\frac{1}{m^2}\sum_{i=0}^{m^2-1} {W}^i(\cdot){W^i}^\dag,
\end{align}
where $\{W^i\}_{i=0}^{m^2-1}$ is the set (group) of Weyl unitaries that form the complete orthonormal basis for the space of all linear operators acting on $m$-dimensional Hilbert space. Let $\mathcal{W}^0=\mathrm{\id}_m$ be the identity operation and $\id_m^{\perp}:=\sum_{i=1}^{m^2-1}\mathcal{W}^i$, then
\begin{align}\label{eq:R-decom}
 \mathcal{R}^\pi_m  & =\frac{1}{m^2}\id_m+\frac{1}{m^2}\sum_{i=1}^{m^2-1}\mathcal{W}^i=\frac{1}{m^2}(\id_m+\id_m^{\perp}).
\end{align}

The single-shot distillation cost of the resource $(\mathcal{N}_{A'B'\to AB},
\Gamma^{\beta}_{A'\to A}\otimes \mathcal{Q}_{B'\to B})$  is given by
\begin{align}
\mathrm{Dist}^{(1,\varepsilon)}
(\mathcal{N}_{A'B'\to AB},&\Gamma^{\beta}_{A'\to A}\otimes \mathcal{Q}_{B'\to B})\nonumber\\[1.5em]
&:=
\sup
\left\{
\log m :
d_{\mathrm{CGPS}}\!\left((
\mathcal{N},
\Gamma^{\beta}\!\otimes\!\mathcal{Q})
\;\to\;
(\mathrm{id}_m\otimes\mathcal{Q}',
\mathcal{R}^{\pi}_m\otimes\mathcal{Q}')\right)
\le \varepsilon\right\}\label{eq:dist}
\\[1.5em]
&=\sup\left\{\log m :
P[\mathrm{id}_m\!\otimes\!\mathcal{Q}',\Theta(\mathcal{N})]
\le \varepsilon,
\;\Theta(\Gamma^{\beta}\!\otimes\!\mathcal{Q})=\mathcal{R}^{\pi}_m\!\otimes\!\mathcal{Q}'\right\}.
\end{align}

A resource distillation process is effectively a measure and prepare channel where an object is processed with a free operation and then a free measurement is performed; the optimal distillation rate is when the processing with free operation and measurement is optimized for maximal distillation of the resource unit. A conditional athermality distillation process from  $(\mathcal{N}_{A'B'\to AB},
\Gamma^{\beta}_{A'\to A}\otimes \mathcal{Q}_{B'\to B})$ under the action of a conditional Gibbs-preserving superchannel $\Theta^\Lambda_\psi$, where $\psi_{RA'B'}$ is a state and $\Lambda_{RAB}$ is an effect operator, s.t. $0\le \Lambda\le \mathbbm{1}$, can be described as 
\begin{align}
\Theta^{\Lambda}_{\psi}(\mathcal{N})
&=
\operatorname{tr}\!\left(
\Lambda\Theta(\mathcal{N})(\psi)
\right)\,
\mathrm{id}_m\!\otimes\!\mathcal{Q}'
+
\operatorname{tr}\!\left(
(\mathbbm{1}-\Lambda)\Theta(\mathcal{N})(\psi)
\right)\,
\frac{\mathrm{id}_m^{\perp}\!\otimes\!\mathcal{Q}'}{m^2-1},
\end{align}
where $\Theta$ is a conditional Gibbs-preserving superchannel. The conditional Gibbs-preserving condition $\Theta^{\Lambda}_{\psi}\bigl(\Gamma^{\beta}\!\otimes\!\mathcal{Q}\bigr)=\mathcal{R}^{\pi}_m\!\otimes\!\mathcal{Q}'$ implies
\begin{align}
\operatorname{tr}
\!\left(\Lambda\Theta(\Gamma^{\beta}\!\otimes\!\mathcal{Q})(\psi)\right)&=\frac{1}{m^2}.
\end{align}
Furthermore, the generalized purified distance of the standard resource unit with respect to output of this superchannel is given by $P[\mathrm{id}_m\otimes\mathcal{Q}',\Theta^{\Lambda}_{\psi}(\mathcal{N})]=\sqrt{1-\operatorname{tr}(\Lambda\Theta(\mathcal{N})(\psi))}$. Hence the conversion distance under such operations can be written as
\begin{align}
d_{\mathrm{CGPS}}
\Big((
\mathcal{N},
\Gamma^{\beta}\!\otimes\!\mathcal{Q}),
(\mathrm{id}_m\!\otimes\!\mathcal{Q}',
\mathcal{R}^{\pi}\!\otimes\!\mathcal{Q}')\Big)
&=
\min_{\Theta,\Lambda,\psi}
\left\{
\sqrt{1-\operatorname{tr}\!\left(\Lambda\Theta(\mathcal{N})(\psi)\right)}:
\operatorname{tr}\!\left(
\Lambda\Theta(\Gamma^{\beta}\!\otimes\!\mathcal{Q})(\psi)
\right)
=
\frac{1}{m^2}
\right\}.
\end{align}
Using this expression in Eq.~\eqref{eq:dist}, we have for a conditional Gibbs-preserving supperchannel $\Theta\in{\rm CGPS}$,
\begin{align}
\mathrm{Dist}^{(1,\varepsilon)}(\mathcal{N},\Gamma^{\beta}\otimes\mathcal{Q})
&=\sup
\left\{
\log m :
\sqrt{1-\operatorname{tr}\left(\Lambda\Theta(\mathcal{N})(\psi)\right)}
\le \varepsilon,
\operatorname{tr}(
(\Lambda\Theta(\Gamma^{\beta}\otimes\mathcal{Q})(\psi)
)
=
\frac{1}{m^2}
\right\}\\
&=
\sup_{\Theta\in{\rm CGPS},\Lambda,\psi}
\left\{
-\frac{1}{2}
\log
\operatorname{tr}(
\Lambda\Theta(\Gamma^{\beta}\otimes\mathcal{Q})(\psi)):
1-\operatorname{tr}(\Lambda\Theta(\mathcal{N})(\psi))
\le \varepsilon^2
\right\}\\
&=
\frac{1}{2}
\sup_{\Theta\in{\rm CGPS}}\sup_{\psi}
\sup_{\Lambda}
\left\{
-\log
\operatorname{tr}(
(\Gamma^{\beta}\otimes\mathcal{Q})(\psi)\Lambda
):
1-\operatorname{tr}(\mathcal{N}(\psi)\Lambda)
\le \varepsilon^2
\right\}\\
&=
\frac{1}{2}
\sup_{\Theta\in{\rm CGPS}}\sup_{\psi}
D_{\mathrm H}^{\varepsilon^2}
\left(
\Theta(\mathcal{N})(\psi)
\;\big\|\;
\Theta(\Gamma^{\beta}\otimes\mathcal{Q})(\psi)
\right)\\
&=
\frac{1}{2}
D_{\mathrm H}^{\varepsilon^2}
\left[
\mathcal{N}
\;\big\|\;
\Gamma^{\beta}\otimes\mathcal{Q}
\right].
\end{align}

The one-shot cost of preparation of a channel resource $(\mathcal{N}_{A'B'\to AB},\Gamma^\beta_{A'\to A}\otimes \mathcal{Q}_{B'\to B})$ under the free operation is defined as
\begin{align}
\mathrm{Cost}^{(1,\varepsilon)}(\mathcal{N},\Gamma^\beta\otimes\mathcal{Q})
=\inf_{m}\Bigl\{\log m :d_{\mathrm{CGPS}}\!\left((
\mathrm{id}_m \otimes \mathcal{Q}',
\mathcal{R}^{\pi} \otimes \mathcal{Q}')
\notag\to(\mathcal{N}, \Gamma^{\beta} \otimes \mathcal{Q})\right)\le \varepsilon\Bigr\}.
\end{align}
Following the definition of conversion distance in Eq.~\eqref{eq:conversion-dist} and using the decomposition of $\mathcal{R}^{\pi}_m$ from Eq.~\eqref{eq:R-decom}, we have
\begin{align}
\Theta(\mathcal{R}^{\pi}_m\otimes \mathcal{Q}')
&=
\frac{1}{m^2}\Theta\left(
\mathrm{id}_m \otimes \mathcal{Q}'
\right)
+
\frac{1}{m^2}\Theta\!\left(
\mathrm{id}_m^{\perp} \otimes \mathcal{Q}'
\right)\\
&=\Gamma^{\beta}\otimes \mathcal{Q}
\end{align}
This yields the bound $m^2 \,\Gamma^{\beta} \otimes \mathcal{Q}\;\ge\;\Theta(\mathrm{id}_m \otimes \mathcal{Q}')$. Replacing $\mathcal{E}=\Theta(\mathrm{id}_m \otimes \mathcal{Q}')$, we have
\begin{align}
d_{\mathrm{CGPS}}\Big(
\bigl(
\mathrm{id} \otimes \mathcal{Q}',\,
\mathcal{R}^{\pi} \otimes \mathcal{Q}'
\bigr)\to\bigl(
\mathcal{N},\,
\Gamma^{\beta}\otimes \mathcal{Q}
\bigr)\Big)
&=\min_{\mathcal{E}}\Bigl\{P[\mathcal{N}, \mathcal{E}]:m^2 \,\Gamma^{\beta} \otimes \mathcal{Q}\;\ge\;\mathcal{E}\Bigr\}\\
&=\min_{\mathcal{E}}\Bigl\{P[\mathcal{N}, \mathcal{E}]:2 \log m\ge D_{\infty}[\mathcal{E} \,\|\,\Gamma^{\beta}\otimes\mathcal{Q}]\Bigr\}.
\end{align}
The conversion cost can now be written as
\begin{align}
\mathrm{Cost}^{(1,\varepsilon)}\bigl(
\mathcal{N}, \Gamma^{\beta} \otimes \mathcal{Q}
\bigr)
&=\inf
\Bigl\{\log m:
P(\mathcal{N}, \mathcal{E}) \le \varepsilon,\log m \ge \frac{1}{2}
D_{\infty}\!\left(
\mathcal{E} \,\|\, \Gamma^{\beta} \otimes \mathcal{Q}
\right) 
\Bigr\}\\
&=\inf
\Bigl\{
\frac{1}{2}
D_{\infty}\!\left(
\mathcal{E} \,\|\, \Gamma^{\beta} \otimes \mathcal{Q}
\right)
:
P(\mathcal{N}, \mathcal{E}) \le \varepsilon
\Bigr\}\\
&=
\frac{1}{2}
D_{\infty}^{\varepsilon}
\bigl[
\mathcal{N}
\,\|\,
\Gamma^{\beta} \otimes \mathcal{Q}
\bigr].
\end{align}
The statements for the optimal distillation and preparation cost follow from their definitions.
\end{proof}
\subsection{Proof of \autoref{prop:ppt}}\label{app:proof-ppt}
\begin{proposition*}
   The conditional min-entropy of a completely PPT-preserving bipartite channel $\n_{A'B'\to AB}$ is always lower bounded by the negative of logarithm of its non-connditioning input system,
    \begin{align}
        S_{\infty}[A|B]_\n\ge -\log|A'|.
    \end{align}
\end{proposition*}
As we will show, the proof of the proposition is a staright forward consequence of the following Lemma.
\begin{lemma}
    Given a bipartite PPT state $\rho_{AB}$, we have
    \begin{equation}
        S_{\infty}^\downarrow(A|B)_\rho\ge 0.
    \end{equation}
\end{lemma}
\begin{proof}
    A bipartite PPT state $\rho_{AB}$ satisfies the following reduction criteria~\cite{HH99},
    \begin{align}
        \mathbbm{1}_A\otimes\rho_B\ge \rho_{AB}.
    \end{align}
    From the primal problem for $S_{\infty}^\downarrow(A|B)_\rho$ given by
    \begin{align}
         2^{-  S_{\infty}^\downarrow(A|B)_{\rho}}&= \inf \{\lambda;\ \lambda \mathbbm{1}_A\otimes \rho_B\ge \rho,\ \lambda > 0 \},
    \end{align}
    we see that for PPT states, the optimal value satisfies $\lambda^*\le 1$, or $S_{\infty}^\downarrow(A|B)_{\rho}\ge 0$. 
\end{proof}
Given a completely PPT-preserving channel $\n_{A'B'\to AB}$, the corresponding Choi state $\phi^{\n}_{R_AAR_BB}$ is a PPT state such that $S^\downarrow_{\infty}(R_AA|R_BB)_{\phi^\n}\ge 0$. Using~\autoref{thm:min-ent-bounds}, we have
\begin{align}
    S_{\infty}[A|B]_\n&\ge S^\downarrow_{\infty}(R_AA|R_BB)_{\phi^\n}-\log|A'|\\
    &\ge -\log|A'|.
\end{align}
\subsection{Proof of \autoref{thm:conti}}\label{app:proof-conti}
\begin{theorem*}
  Consider two arbitrary bipartite channels $\n_{A'B'\to AB}$ and $\m_{A'B'\to AB}$. For $\frac{1}{2}\norm{\n-\m}_\diamond\le\delta\in[0,1]$, the difference of their conditional min-entropy is bounded as
\begin{equation}
    \abs{S_{\infty}[A|B]_\n-S_{\infty}[A|B]_\m}\le \frac{1}{\ln 2}{|A|\min\{|A'||A|,|B'||B|\}\delta}.
\end{equation}
\end{theorem*}
\begin{proof}
Consider the functional $\mathcal{E}:\mathrm{Ch}(A'B',AB)\to \mathbb{R}^+$ given by  $\mathcal{E}(\n):= |A'|^{-1}2^{-S_{\infty}[A|B]_\n}$, or equivalently
\begin{equation}
\begin{aligned}
    \mathcal{E}(\n)&=\inf_{\m\in \mathrm{Ch}(B',B) } \{\lambda,\ \lambda(\mathbbm{1}_{AA'}\otimes \Phi^{\m})\ge \Phi^\n \}.
\end{aligned}
\end{equation}
One can extend this functional to the space of Hermiticity preserving maps, 
\begin{equation}
\begin{aligned}
    \widetilde{\mathcal{E}}(\n)&=\inf_{\m\in \mathrm{HPTP}(B',B) } \{\lambda,\ \lambda(\mathbbm{1}_{AA'}\otimes\Phi^\m)\ge \Phi^\n \}.
\end{aligned}
\end{equation}
Note that we can equate $\widetilde{\mathcal{E}}(\n)=\widehat{\mathcal{E}}(\Phi^\n)$ as a map on the Choi matrices of the Hermiticity preserving maps $\n$. The map $\widetilde{\mathcal{E}}(\n)$ satisfies the following properties:
\begin{enumerate}
    \item $\widetilde{\mathcal{E}}(a\n)=a\widetilde{\mathcal{E}}(\n)$ $\forall\ a\in \mathbb{R}^+$, since, $\Phi^{(a\n)}=a\Phi^{\n}$. 
    \item $\Phi^\n\ge \Phi^\m\implies \widetilde{\mathcal{E}}(\n)\ge \widetilde{\mathcal{E}}(\m)$. 
    \item  $\widetilde{\mathcal{E}}(\n+\m)\le \widetilde{\mathcal{E}}(\n)+\widetilde{\mathcal{E}}(\m)$.
    \item  $\widetilde{\mathcal{E}}(\n)\le  \min\{|AA'|,|BB'| \}\tr\{\Phi^\n P^+\}$. Here, $\tr(XP^+)$ is the sum of the positive eigenvalues of $X$. To show this, we note that for any Hermiticity preserving map $\n$, we have the spectral decomposition of its Choi matrix as $\Phi^\n=\sum_i\lambda_i\Phi^{\mathcal{K}_i}$ where $\{\mathcal{K}_i\}_i$ are rank 1 CP maps. This is because every bipartite pure state can be written as a Choi matrix of a CP map. Using the property 1 and 3, we have the convexity of $\widetilde{\mathcal{E}}(\n)$, such that,
\begin{align}
    \widetilde{\mathcal{E}}(\n)\le \sum_i\lambda_i\widetilde{\mathcal{E}}(\mathcal{K}_i)\\
    \le \sum_{\lambda_i\ge 0}\lambda_i\widetilde{\mathcal{E}}(\mathcal{K}_i)
\end{align}
We have
\begin{align}
    \widetilde{\mathcal{E}}(\mathcal{K}_i)&=\inf_{\m\in \mathrm{HPTP}(B',B) } \{\lambda,\ \lambda(\mathbbm{1}_{A'A}\otimes\Phi^\m)\ge \Phi^{\mathcal{K}_i} \}\\
    &\le \inf~\{\lambda,\ \lambda(\mathbbm{1}_{A'A}\otimes\Phi^{\m_i})\ge \Phi^{\mathcal{K}_i} \}\\
    &=2^{-S_{\infty}^{\downarrow}(R_AA|R_BB)_{\Phi^{\mathcal{K}_i}}}\\
    &\le\min\{|A'||A|,|B'||B|\}.
\end{align}
Here $\m_i(\cdot)=\tr_A\circ\mathcal{K}_i(\mathbbm{1}_A\otimes \cdot)$ is the reduced CP map and $\Phi^{\m_i}=\tr_{A'A}(\Phi^{\mathcal{K}_i})$.  In the last inequality, we have used the lower bound on the min-conditional entropy: $S_{\infty}^{\downarrow}(A|B)_\rho\ge-\log\min\{|A|,|B|\}$. Furthermore, the sum of the positive part of the spectrum of the state $\Phi^{\n}$ is given by $\sum_{\lambda_i\ge 0}\lambda_i=\tr(\Phi^\n P_+)$, which gives us the required bound.
\end{enumerate}
The diamond norm between two maps is given as
\begin{equation}
\Vert \n-\m\Vert_\diamond=\max_{\rho} \Vert \id\otimes\n(\rho)-\id\otimes\m(\rho) \Vert_1\ge \Vert\Phi^\n-\Phi^\m\Vert_1.
\end{equation}
Therefore, $\frac{1}{2}\Vert \n-\m\Vert_\diamond\le \delta$ implies $\frac{1}{2}\Vert_1\ge \Vert\Phi^\n-\Phi^\m\Vert_1\le \delta$. Using the properties of the functional $\widetilde{\mathcal{E}}$, we have the following for two arbitrary CPTP map $\n$ and $\m\in \mathrm{Ch}(A'B',AB)$
        \begin{align}
            \widetilde{\mathcal{E}}(\n)&\le\widetilde{\mathcal{E}}(\m)+\widetilde{\mathcal{E}}(\n-\m)\\
            &\le \widetilde{\mathcal{E}}(\m)+\min\{|AA'|,|BB'|\}\tr((\Phi^\n-\Phi^\m)P^+)\\
            &= \widetilde{\mathcal{E}}(\m)+\min\{|AA'|,|BB'|\}\Vert\Phi^\n-\Phi^\m\Vert_1\\
             &\le \widetilde{\mathcal{E}}(\m)+\min\{|AA'|,|BB'|\}\delta
        \end{align}
Here we have used equality $\tr(XP^+)=\dfrac{1}{2}(\Vert X \Vert_1+\tr(X))$. Consequently, if $\n$ and $\m$ are quantum channels, $\tr((\Phi^\n-\Phi^\m)P^+)=\dfrac{\Vert\Phi^\n-\Phi^\m\Vert_1}{2}$. Using $\ln(a+x)-\ln(a)\le \dfrac{x}{a}$ and substituting $a+x=\widetilde{\mathcal{E}}(\n)$, and $a=\widetilde{\mathcal{E}}(\m)$, we have 
\begin{align}
     \widetilde{\mathcal{E}}(\n)-\widetilde{\mathcal{E}}(\m)& \le \min\{|AA'|,|BB'|\}\delta\\
     \implies \dfrac{2^{-S_{\infty}[A|B]_\n}-2^{-S_{\infty}[A|B]_\m}}{2^{-S_{\infty}[A|B]_\n}}&\le \dfrac{\delta}{2^{-S_{\infty}[A|B]_\n}}\min\{|AA'|,|BB'|\}\\
     \implies S_{\infty}[A|B]_\n-S_{\infty}[A|B]_\m & \le  \dfrac{|A|\min\{|AA'|,|BB'|\}\delta}{\tr(\Phi^\m)\ln 2}
\end{align}
In the last step, we have used the inequality $2^{-S_{\infty}[A|B]_\n}\ge \dfrac{\tr(\Phi^\n)}{|A|}$.  
\end{proof}

\subsection{Proof of \autoref{th:equipartition}}\label{app:proof_equipartition}
\begin{theorem*}
   For $\varepsilon\in(0,1)$ and asymptotically many i.i.d.~uses of a bipartite channel $\n_{A'B'\to AB}$, the smoothed conditional channel min-entropy satisfies the following bounds in terms of the conditional von Neumann entropy of the Choi state $\Phi^{\n}$,
\begin{align}
    \lim_{\varepsilon\to 0^+} \lim_{n\to \infty}\frac{1}{n} S_{\infty}^\varepsilon [A^n|B^n]_{\mathcal{N}^{\otimes n}}  \leq  S(R_AA|R_BB)_{\Phi^{\mathcal{N}}}-\log|A'|.\label{tp:eq1}
\end{align}
If $\n$ is tele-covariant, then
\begin{align}
    \lim_{\varepsilon\to 0^+} \lim_{n\to \infty}\frac{1}{n} S_{\infty}^\varepsilon [A^n|B^n]_{\mathcal{N}^{\otimes n}}  \leq  S[A|B]_{\n},\label{tp:eq2}\\
   \lim_{n\to \infty}\frac{1}{n} S_{\infty}^\varepsilon [A^n|B^n]_{\mathcal{N}^{\otimes n}}  \geq  S[A|B]_{\n}.\label{tp:eq3}
\end{align}
\end{theorem*}
\begin{proof}
For an arbitrary channel $\n_{A'B'\to AB}$, using the upper bound on conditional channel min-entropy from \autoref{thm:min-ent-bounds}, we get 
\begin{equation}
S^\varepsilon_{\infty} [A|B]_{\mathcal{N}}  \leq \sup_{\m\in B^\varepsilon[\n]} S_{\infty}(R_AA|R_BB)_{\Phi^{\mathcal{M}}}-\log|A'|,
\end{equation}
Note that
\begin{align}\label{eq:smooth-channel-ineq}
    \sup_{\m\in \mathcal{B}^{\varepsilon}[\n]} S_{\infty}(R_AA|R_BB)_{\Phi^\m}&\le\sup_{\Phi^\m\in \mathcal{B}^{\varepsilon}(\Phi^\n)} S_{\infty}(R_AA|R_BB)_{\Phi^\m}\\
   & \le \sup_{\rho\in \mathcal{B}^{\varepsilon}(\Phi^\n)} S_{\infty}(R_AA|R_BB)_{\Phi^\m}\\
   &= S_{\infty}^{\varepsilon}(R_AA|R_BB)_{\Phi^\n},
\end{align}
where the smoothed conditional min-entropy is defined for $\varepsilon\in[0,1]$ as
\begin{align}
    S_{\infty}^{\varepsilon}(A|B)_\rho &:=\sup_{\widetilde{\rho}\in B^\varepsilon(\rho)} S_{\infty}(R_AA|R_BB)_{\widetilde{\rho}},\\
    \text{and},~~S_{\infty}^{\downarrow,\varepsilon}(A|B)_\rho &:=-\inf_{\widetilde{\rho}\in B^\varepsilon(\rho)} D_{\infty}(\widetilde{\rho}\Vert \mathbbm{1}_{A}\otimes \rho_B).
\end{align}
We now use the asymptotic convergence of the conditional min-entropy to prove the first inequality of the theorem. For $\varepsilon\in (0,1)$, the smoothed conditional min-entropy functions follows the asymptotic equipartition property~\cite[Theroem 1]{TCR09},
\begin{align}
\lim_{\varepsilon\to 0^+}\lim_{n\to \infty} &    \dfrac{1}{n}S_{\infty}^{\varepsilon}(A^n|B^n)_{\rho^{\otimes n}}=S(A|B)_{\rho}.
\end{align}
Noting $\Phi^{\n^{\otimes n}}=({\Phi^\n})^{\otimes n}$, we get
\begin{align}
    \lim_{\varepsilon\to 0^+}\lim_{n\to \infty}\frac{1}{n} S_{\infty}^\varepsilon [A^n|B^n]_{\mathcal{N}^{\otimes n}}  \leq  S(R_AA|R_BB)_{\Phi^{\mathcal{N}}}-\log|A'|,
\end{align}
which concludes the proof of inequality~\eqref{tp:eq1}.

Inequality~\eqref{tp:eq2} follows from inequality~\eqref{tp:eq1} and the fact that for a tele-covariant channel $\n$,
\begin{equation}
    S[A|B]_{\n}=S(R_AA|R_BB)_{\Phi^{\n}}-\log|A'|.
\end{equation}

To show the inequality~\eqref{tp:eq3}, we adapt the techniques used in the proof argument of~\cite[Theorem 21]{GW21}. Let $\m_{{A'}^n{B'}^n\to A^nB^n}$ be a channel in $\varepsilon$-neighborhood of $\n^{\otimes n}$, then
\begin{align}
    \sup_{\m\in \mathcal{B}^\varepsilon[\n^{\otimes n}]}S_{\infty}[A|B]_\m =  \sup_{\m\in \mathcal{B}^\varepsilon[\n^{\otimes n}]}-\inf_{\mathcal{P}}D_\infty [\m\Vert \R_{{A'}^n\to A^n}\otimes \mathcal{P}_{{B'}^n\to B^n}]
\end{align}
Consider the following choice of the channel $\mathcal{P}_{{B'}^n\to B^n}=\mathcal{K}_{B'\to B}^{\otimes n}$ defined as
\begin{align}
    \mathcal{K}_{B'\to B}(\cdot):=\tr_A\circ \n_{A'B'\to AB}(\pi_{A'}\otimes\cdot).
\end{align}
A de Finetti state $\omega_{R^n{A'}^n{B'}^n}$, where $R\cong A'B'$, is a state of a multipartite system that is invariant under any permutation of its subsystems. It can be decomposed as 
\begin{equation}
    \omega_{R^n{A'}^n{B'}^n}=\int d\mu(\sigma)\sigma_{RA'B'}^{\otimes n},
\end{equation}
where $\sigma_{RA'B'}$ are pure states and $d\mu(\sigma)$ is the Haar measure on the space of pure states. Let $\mathrm{Perm}({A'}^n,A^n)$ denote the set of permutation covaraint channels from ${A'}^n\to A^n$. Action of such channels on de Finetti states would leave them permutation invariant.
\begin{align}
    \sup_{\m\in \mathcal{B}^\varepsilon[\n^{\otimes n}]}S_{\infty}[A|B]_\m &\ge \sup_{\m\in \mathcal{B}^\varepsilon[\n^{\otimes n}]}-D_\infty [\m\Vert \R_{{A'}^n\to A^n}\otimes {\mathcal{K}}^{\otimes n}]\\
    &\ge \sup_{\m\in \mathcal{B}^\varepsilon[\n^{\otimes n}],~\m\in \mathrm{Perm}((A'B')^n,(AB)^n)}-D_\infty [\m\Vert \R_{{A'}^n\to A^n}\otimes {\mathcal{K}}^{\otimes n}]\\
    &=\sup_{\m\in \mathcal{B}^\varepsilon[\n^{\otimes n}],~\m\in \mathrm{Perm}((A'B')^n,(AB)^n)}-D_\infty (\m(\omega)\Vert \R_{{A'}^n\to A^n}\otimes {\mathcal{K}}^{\otimes n}(\omega)),\label{eq:three}
\end{align}
where $\omega_{R'R^n{A'}^n{B'}^n}$ is the purification of the de Finette state $\omega_{R^n{A'}^n{B'}^n}$. The equality holds since $\omega_{{A'}^n{B'}^n}$ is a full rank state. From~\cite{CKR09}, we note that if $\m\in \mathrm{Perm}((A'B')^n,(AB)^n)$ such that $\m\in \mathcal{B}^\varepsilon[\n^{\otimes n}]$, then $\m(\omega)\in \mathcal{B}^{\varepsilon'}(\n^{\otimes n}(\omega))$ where $\varepsilon'=\varepsilon(n+1)^{2(|A|^2-1)}$. We have the final expression on the RHS satisfy the following
\begin{align}
    \eqref{eq:three}&\ge \sup_{\m(\omega)\in \mathcal{B}^{\varepsilon'}(\n^{\otimes n}(\omega)),~\m\in \mathrm{Perm}((A'B')^n,(AB)^n)}-D_\infty (\m(\omega)\Vert \R\otimes {\mathcal{K}}^{\otimes n}(\omega))\label{eq:one}\\
    &=\sup_{\sigma\in \mathcal{B}^{\varepsilon'}(\n^{\otimes n}(\omega)),~\sigma_{R{R'}^n}=\omega_{R{R'}^n}}-D_\infty (\sigma\Vert \R\otimes {\mathcal{K}}^{\otimes n}(\omega))\label{eq:two}
\end{align}
This equality can be shown as follows: clearly, \eqref{eq:one}~$\le$~\eqref{eq:two} since the reduced state of $\m(\omega)$ is $\omega_{R'R}$. \eqref{eq:one}~$\ge$~\eqref{eq:two} follows because $\n^{\otimes n}$ is invariant under the action of the superchannel $\Theta(\m)=\frac{1}{n!}\mathrm{\pi}_{(AB)^n}\circ\m\circ\mathrm{\pi}_{{(A'B')}^n}\in \mathrm{Perm}((A'B')^n,(AB)^n)$ and 
\begin{align}
    \m(\omega)\in \mathcal{B}^{\varepsilon'}(\n^{\otimes n}(\omega))\implies\Theta(\m)(\omega)\in \mathcal{B}^{\varepsilon'}(\n^{\otimes n}(\omega)).
\end{align}
Let $\sigma$ be the state that optimizes the expression in Eq.~\eqref{eq:two}. From the~\cite[Lemma 10]{FWT+18}, there exists a CPTP map $\widetilde{\m}$ s.t. $\sigma=\widetilde\m(\omega)\in \mathcal{B}^\varepsilon(\n^{\otimes n}(\omega))$. Due to the $\Theta$-invariance of $\n^{\otimes n}$, we can show that $\Theta(\widetilde{\m})(\omega)\in \mathcal{B}^\varepsilon(\n^{\otimes n}(\omega))$ and since $\Theta(\widetilde{\m})\in \mathrm{Perm}((A'B')^n,(AB)^n)$, we have 
\begin{align}
    \eqref{eq:one}&\ge -D_{\infty}(\Theta(\widetilde{\m})(\omega)\Vert \R\otimes\mathcal{K}^{\otimes n}(\omega))\\
    &=-D_{\infty}(\Theta(\widetilde{\m})(\omega)\Vert \Theta(\R\otimes\mathcal{K}^{\otimes n})(\omega))\\
    &\ge -D_{\infty}(\widetilde{\m}(\omega)\Vert \R\otimes\mathcal{K}^{\otimes n}(\omega))=\eqref{eq:two}.
\end{align}
The equality follows since the channel $\mathcal{R}^\pi_{A^n}\otimes\mathcal{K}^{\otimes n}$ is invariant under $\Theta$. This shows \eqref{eq:one}=\eqref{eq:two}. Using~\cite[Theorem 3]{ABJ+20}, for $\varepsilon'/2-\delta\le 1$ and $\delta\ge 0$, we get 
\begin{align}
      \eqref{eq:two}&=\sup_{\sigma\in \mathcal{B}^{\varepsilon'/2-\delta}(\n^{\otimes n}(\omega))}-D_\infty (\sigma\Vert \R_{{A'}^n\to A^n}\otimes {\mathcal{K}}^{\otimes n}(\omega))-\log\dfrac{8+\delta^2}{\delta^2}\\
      &\ge -D_\alpha(\n^{\otimes n}(\omega)\Vert \R_{{A'}^n\to A^n}\otimes {\mathcal{K}}^{\otimes n}(\omega)))\nonumber\\
      &\qquad -\log\dfrac{8+\delta^2}{\delta^2}-\frac{1}{\alpha-1}\log\frac{1}{(\varepsilon'/2-\delta)^2}-\log\frac{1}{1-(\varepsilon'/2-\delta)^2}\\
      &\ge -nD_\alpha[\n\Vert\R_{A'\to A}\otimes\mathcal{K}] +f(\varepsilon',\delta),
\end{align}
where $\alpha>1$ and we have used the results from~\cite[Theorem 3]{ABJT19}. Then, for $\alpha>1$, we have
\begin{align}
    S_{\infty}^\varepsilon[A^n|B^n]_{\n^{\otimes n}}\ge  -nD_\alpha[\n\Vert\R_{A'\to A}\otimes\mathcal{K}] +f(\varepsilon',\delta),\\
     \frac{1}{n} S_{\infty}^\varepsilon[A^n|B^n]_{\n^{\otimes n}}\ge  -D_\alpha[\n\Vert\R_{A'\to A}\otimes\mathcal{K}] +\frac{1}{n}f(\varepsilon',\delta).\label{eq:lb-asym-min-smth}
\end{align}
Taking the limit $n\to \infty$, then $\lim \alpha\to 1^+$ and noting that $\Phi^{\mathcal{K}}=\tr_{R_AA}\Phi^{\mathcal{N}}$ for a tele-covariant channel $\n$, we get
\begin{align}
   \lim_{n\to \infty}\frac{1}{n} S_{\infty}^\varepsilon [A^n|B^n]_{\mathcal{N}^{\otimes n}}  \geq  S(R_AA|R_BB)_{\Phi^{\mathcal{N}}}-\log|A'|=S[A|B]_\n.
\end{align}
\end{proof}

\subsection{Proof of \autoref{thm:asym-ns}}\label{app:ns-asym}
\begin{theorem*}[Asymptotic equipartition property]
    For $\varepsilon\in(0,1)$ and asymptotically many i.i.d.~uses of an arbitrary bipartite channel $\n_{A'B'\to AB}$, the smoothed conditional channel min-entropy satsifies the following bounds,
    \begin{align}
        \lim_{\varepsilon\to 0^+}\lim_{n\to\infty}\frac{1}{n}S^{\varepsilon}_{\infty}[A^n|B^n]_{\n^{\otimes n}}\leq S^{\not\to}[A|B]_{\n}=\inf_{\psi\in\St(RA'B')}S(A|RB)_{\n(\psi)}.
    \end{align}
    If $\n_{A'B'\to AB}$ is no-signaling, $\n\in\NS$, then
    \begin{align}
        \lim_{\varepsilon\to 0^+}\lim_{n\to\infty}\frac{1}{n}S^{\varepsilon}_{\infty}[A^n|B^n]_{\n^{\otimes n}}&\leq S[A|B]_{\n},\\
        \lim_{n\to\infty}\frac{1}{n}S^{\varepsilon}_{\infty}[A^n|B^n]_{\n^{\otimes n}}&\geq S[A|B]_{\n}.
    \end{align}
\end{theorem*}
\begin{proof}
        We prove Eq.~\eqref{eq:dist-ns-cap} in parts. Let us observe that $S^{\rm reg}[A|B]_{\n}$ for a channel $\n\in\NS$ is additive. It follows because, for any quantum channel $\m_{C'_i\to A_iB_i}$, $i\in\{1,\ldots,n\}$, we have for $\alpha\geq 1$~\cite{FKR+26}
    \begin{equation}
        \inf_{\rho\in\St(R{C'}^n)}S_{\alpha}(A^n|RB^n)_{\m^{\otimes n}(\rho)}=\sum_i\inf_{\rho\in\St(R_iC'_i)}S_{\alpha}(A_i|R_iB_i)_{\m(\rho)}.
    \end{equation}
That makes $S^{\not\to}[A|B]_{\n}$ to be additive under tensor-product of bipartite quantum channels which in turn makes $S[A|B]_{\n}$ to be additive under tensor-product of bipartite quantum channels that are no-signaling from $A'\to B$. In particular, for quantum channels $\n^i_{A_i'B_i'\to A_iB_i}$, such that $\n^i\in{\rm S}_{A_i'\to B_i}$ for $i\in\{1,\ldots,n\}$, then
\begin{equation}
    S[A^n|B^n]_{\otimes_i\n^i}=S^{\not\to}[A^n|B^n]_{\otimes_i\n^i}=\sum_{i=1}^nS^{\not\to}[A_i|B_i]_{\n^i}=\sum_{i=1}^n\inf_{\psi\in\St(R_iA_i'B_i')}S(A_i|R_iB_i)_{\n^i}.
\end{equation}
Therefore, for $\n\in\NS$, we have
\begin{align}
      {\rm Dist}^{(\infty,0)}(\n,\mathcal{R}^{\pi})&\leq \frac{1}{2}\left(\log|A|-\lim_{n\to\infty}\frac{1}{n}S[A^n|B^n]_{\n^{\otimes n}}\right)\\
      &= \frac{1}{2}\left(\log|A|-\lim_{n\to\infty}\frac{1}{n}nS[A|B]_{\n}\right)\\
      &=\frac{1}{2}(\log|A|-S[A|B]_{\n})=\frac{1}{2}(\log|A|-S^{\not\to}[A|B]_{\n}),
\end{align}
which concludes the proof for the conditional distillation capacity because $  {\rm Dist}^{(\infty,0)}(\n,\mathcal{R}^{\pi})\geq \frac{1}{2}(\log|A|-S^{\not\to}[A|B]_{\n})$ holds for arbitrary $\n$.

For an arbitrary quantum channel $\n_{A'B'\to AB}$ and its reduced channel $\Q_{B'\to B}$, where for all states $\rho_{A'}$,
\begin{equation}
    \Q_{B'\to B}(\cdot)=\tr_A\circ\n_{A'B'\to AB}\circ\mathcal{R}^\pi_{A'\to A}(\rho_{A'}\otimes\cdot),
\end{equation}
we have from Eq.~\eqref{eq:asym-dist-lb},
for $\alpha>1$,
\begin{align}
     \frac{1}{n} S_{\infty}^\varepsilon[A^n|B^n]_{\n^{\otimes n}}\ge  -D_\alpha[\n\Vert\R_{A'\to A}\otimes\mathcal{Q}] +\frac{1}{n}f(\varepsilon',\delta).\label{eq:asym-dist-lb2}
\end{align}
We note that $\tr_{RA}({\mathcal{N}(\psi_{RA'B'})}=\Q(\psi_{RB'})$, for all input states $\psi_{RA'B'}$ and $\psi_{RB'}=\tr_{A'}(\psi_{RA'B'})$, if $\n\in\NS$. For $\n\in\NS$, taking the limit $n\to \infty$ and then $\lim \alpha\to 1^+$ in Eq.~\eqref{eq:asym-dist-lb2}, we get
\begin{equation}
     \frac{1}{n} S_{\infty}^\varepsilon[A^n|B^n]_{\n^{\otimes n}}\ge S[A|B]_{\n}=\inf_{\psi\in\St(RA'B')}S(A|RB)_{\n(\psi)}.
\end{equation}
That is, for $\n\in\NS$,
\begin{equation}
    {\rm Cost}^{(\infty,0)}(\n,\mathcal{R}^\pi)\leq \frac{1}{2}(\log|A|-S^{\not\to}[A|B]_{\n}),
\end{equation}
concluding the proof as $  {\rm Dist}^{(\infty,0)}(\n,\mathcal{R}^\pi)\leq{\rm Cost}^{(\infty,0)}(\n,\mathcal{R}^\pi)$.
\end{proof}

\subsection{Proof of \autoref{thm:SwapBound}}\label{app:proof-SwapBound}
The swap operation on a bipartite system simply exchanges the subsystems, such that the action of the corresponding unitary is given by $U^{\rm SWAP}_{A'B'\to AB}(\ket{i}_{A'}\ket{j}_{B'})=\ket{j}_A\ket{i}_B$. The Choi state of the swap operation is maximally entangled across the $R_AA:R_BB$ partition. We refer to any channel as a swap-like or maximally entangling channel whose Choi state is equivalent to $\Phi_{R_AA:R_BB}$. Below we provide the dimensional constraints between the input and output systems under any isometry that maps two maximally entangled states.
\begin{lemma}\label{lemma:max_ent}
Given two maximally entangled states $\Phi_{A:B}$ and $\Phi_{A:C}$, there exists  $\mathcal{E}_{B\to C}\in \mathrm{Ch}(B,C)$ such that $\mathcal{E}_{B\to C}(\Phi_{A:B})=\Phi_{A:C}$, only if either $|A|\le |B|\le |C|$, or $|A|\ge |B|=|C|$.
\end{lemma}
\begin{proof}
Since the Choi state of the map $\mathcal{E}_{\mathcal{B\to C}}$ is a pure state, it must be an isometry, and hence  $|B|\le |C|$ . Let $\{a_i\}_i $, $\{b_j\}_j$ and $\{c_k\}_k$ be orthonormal basis of the Hilbert spaces $A$, $B$ and $C$ respectively. We can write the maximally entangled in $AB$ state $\Phi_{A:B}=\dfrac{1}{d}\sum_{i,j}\ket{a_i,b_i}\bra{a_j,b_j}$, where $d=\min\{|A|, |B|\}$. Let $\mathcal{E}_{B\to C}$ be a channel that maps $\Phi_{A:B}$ to $\Phi_{A:C}$, that is 
    \begin{equation}
        \mathcal{E}_{B\to C}\big(\dfrac{1}{d}\sum_{i.j}^d \ket{a_i,b_i}\bra{a_j,b_j}\big)=\dfrac{1}{d'}\sum_{k,l}^{d'} \ket{a_k,c_k}\bra{a_l,c_l},
    \end{equation}
where, $d'=\min\{|A|, |C|\}$. Since a local map cannot increase the Schmidt rank of a state, we must have $d\ge d'$. Combining this with the observation $|B|\le|C|$, if $|A|\ge |B|$, we must have $|C|=|B|$, else we must have $|A|\le |B|$.\end{proof}
Now we proceed to prove the main theorem.
\begin{theorem*} The conditional min-entropy of a bipartite quantum channel $\n_{A'B'\to AB}$ achieves the minimum if and only if the channel is a maximally entangling unitary operation. For a quantum channel $\n_{A'B'\to AB}\in{\rm SWAP}(A;B)$, we have 
\begin{equation}
    S_{\infty}[A|B]_{\n\in\mathrm{SWAP}(A;B)}=-\log \min\{|A'|^2|A|,|A'||B'||B|\}.
\end{equation}
\end{theorem*}
\begin{proof}
We will first show that the upper and lower bound of $S_{\infty}[A|B]_\n$ as given in \autoref{thm:min-ent-bounds} coincide at $\Phi^\n$ if it is a maximally entangled state, that is 
\begin{equation}
   \min_\n S^{\downarrow}_{\infty}(R_AA|R_BB)_{\Phi^{\mathcal{N}}} =\min_\n S_{\infty}(R_AA|R_BB)_{\Phi^{\mathcal{N}}} =S^{\downarrow}_{\infty}(R_AA|R_BB)_{\Phi^{\mathrm{SWAP}_{A:B}}} 
\end{equation} Consider the SDP for $ S_{\infty}(A|B)_{\rho_{AB}}$ which can be written as~\cite{KRS09}
\begin{equation}\label{eq:sminup_bound}
        S_{\infty}(A|B)_{\rho_{AB}}=-\log \sup_{X_{AB}}\{\tr(\rho_{AB} X_{AB}): \tr_A(X_{AB})= \mathbbm{1}_{B}, X_{AB}\ge 0\}.  
    \end{equation}
We can write $X_{AB}=\mathcal{E}_{B'\to A}(\Gamma_{BB'})$ as a Choi matrix of a channel $\mathcal{E}_{B'\to A}$, which gives us 
\begin{equation}
\begin{aligned}
   S_{\infty}(A|B)_{\rho_{AB}}&=-\log |d|-\log \sup_{\mathcal{E}_{B'\to A}}\{\tr(\rho_{AB}\mathcal{E}_{B'\to A} (\Phi_{BB'})): \mathcal{E}_{B'\to A} \in \mathrm{Ch}(B',A)\}
\end{aligned}    
\end{equation}
Here $d=\min\{|B|,|B'|\}$. Note that for density matrices $\rho$ and $\sigma$, $\tr(\rho\sigma)\le 1$ and it is maximum iff $\rho=\sigma=\ket{\psi}\bra{\psi}$. Therefore, we have 
\begin{equation}
     S_{\infty}(A|B)_{\rho_{AB}}\ge -\log d,
\end{equation}
and the bound is achieved iff $\rho_{AB}=\mathcal{E}_{B'\to A} (\Phi_{BB'})=\Phi_{AB}$. From \autoref{lemma:max_ent} we see that this is possible in two cases: case 1: $|B|\le |B'|\le |A|$, which implies $d=|B|$ and case 2: $|B|\ge |B'|=|A|$, which implies $d=|A|$. Hence,
\begin{equation}
    \min_{\rho_{AB}} S_{\infty}(A|B)_{\rho_{AB}}= -\log d= \begin{cases}
      -\log|B| & \text{if $|B|\le |A|$}.\\
      -\log|A| & \text{if $|A|\le |B|$}.
    \end{cases} 
\end{equation}
Equivalently,
\begin{equation}
    \min_{\rho_{AB}} S_{\infty}(A|B)_{\rho_{AB}}= S_{\infty}(A|B)_{\Phi_{AB}}=-\log\min\{|A|,|B|\}.
\end{equation}
For $S_{\infty}^\downarrow(A|B)_{\rho}$, the SDP can be written as 
\begin{equation}\label{eq:smindown_bound}
       S_{\infty}^\downarrow(A|B)_{\rho_{AB}}=-\log\sup_{X_{AB}}\{\tr(\rho_{AB} X_{AB}): \tr((\mathbbm{1}_A \otimes\rho_{B})X_{AB})= 1, X_{AB}\ge 0\}. 
    \end{equation}
Here, $X_{AB}$ being a positive operator, can be written as a Choi matrix of a CP map: $X_{AB}=\mathcal{E}_{A'\to B}(\Gamma_{AA'})$. This makes the dual map $\mathcal{E}^*_{B\to A'}$ a CP map such that $\mathcal{E}^*_{B\to A'}(\rho_B)$ is a density matrix. This also implies that $\mathcal{E}^*_{B\to A'}(\rho_{AB})$ is also a density matrix. We have
\begin{equation}
\begin{aligned}
        \min_{\rho_{AB}} S_{\infty}^\downarrow(A|B)_{\rho_{AB}}&=-\log\sup_{\rho_{AB}}\{|d|\tr(\mathcal{E}^*_{A\to B'}(\rho_{AB})\Phi_{AA'})\}\\
        &=  S_{\infty}^\downarrow(A|B)_{\Phi_{AB}}\\
        &= -\log\min \{|A|,|B| \}.
    \end{aligned}
    \end{equation}
From Theorem\ \ref{thm:min-ent-bounds} we have the identity
\begin{align}
    S^{\downarrow}_{\infty}(R_AA|R_BB)_{\Phi^{\mathcal{N}}}-\log|A'|  \leq S_{\infty}[A|B]_{\mathcal{N}} & \leq  S_{\infty}(R_AA|R_BB)_{\Phi^{\mathcal{N}}}-\log|A'|.
    \end{align}
using the above results, we have
\begin{equation}
    \begin{aligned}
        \min_\n S^{\downarrow}_{\infty}(R_AA|R_BB)_{\Phi^{\mathcal{N}}} &\le S^{\downarrow}_{\infty}(R_AA|R_BB)_{\Phi^{\mathrm{SWAP}_{A:B}}}\\
        &= S^{\downarrow}_{\infty}(R_AA|R_BB)_{\Phi_{R_AA:R_BB}}\\
        &= -\log\min\{|A'||A|,|B'||B|\}\\
        &\le   \min_\n S^{\downarrow}_{\infty}(R_AA|R_BB)_{\Phi^{\mathcal{N}}}
    \end{aligned}
\end{equation}
Therefore, for swap-like operations $\n=\mathrm{SWAP}_{A:B}$, where $\Phi^\n=\Phi_{R_AA:R_BB}$, we have
\begin{align}
  \inf_\n S^{\downarrow}_{\infty}(R_AA|R_BB)_{\Phi^{\mathcal{N}}} = -\log \min\{|A'||A|,|B'||B|\}  
\end{align}
and similarly 
\begin{align}
    \inf_\n S_{\infty}(R_AA|R_BB)_{\Phi^{\mathcal{N}}} = -\log \min\{|A'||A|,|B'||B|\}
\end{align}
Since the upper and lower bounds meet, we have the following relation
\begin{equation}
   \inf_\n S_{\infty}[A|B]_{\mathcal{N}}=S_{\infty}[A|B]_{\mathrm{SWAP}_{A:B}}=-\log|A'|-\log \min\{|A'||A|,|B'||B|\}.
\end{equation}
The uniqueness follows since $S^{\downarrow}_{\infty}(R_AA|R_BB)_{\Phi^\n}=-\log \min\{|A'||A|,|B'||B|\}$ if and only if $\Phi^\n=\Phi_{R_AA:R_BB}$.
\end{proof}\
\section{Some results on two-qubit channels}\label{app:c4maps}
\begin{lemma}\label{lemma:unitaryroot}
 Given a unitary channel $\mathcal{U}_{A'B'\to AB}$, with the Choi state $\Phi^\U$ such that $|B'|=|B|=2$, the state $\rho_{R_BB}:= \frac{1}{\tr\sqrt{\Phi^\U_{R_BB}}}\sqrt{\Phi^\U_{R_BB}}$ is a Choi state of a unital qubit channel.  
\end{lemma}
\begin{proof}
Since $\Phi^{\U}$ is a Choi state of a unital CP map, we have $\Phi^{\U}_{R_AR_B}=\Phi^{\U}_{AB}=\frac{1}{2|A|}\mathbbm{1}_{2|A|}$. As a consequence, the state $\Phi^\U_{R_BB}\in \mathbb{C}^2\otimes \mathbb{C}^2$, is a Choi matrix of a unital qubit channel, since $\tr_{R_B}\Phi^\U_{R_BB}=\frac{1}{2}\mathbbm{1}_2$ and $\tr_{B}\Phi^\U_{R_BB}=\frac{1}{2}\mathbbm{1}_2$. In fact $\Phi^\U_{R_BB}\in \mathbb{C}^2\otimes \mathbb{C}^2$ is a Choi state of a unistochastic channel~\cite{MKZ13}. These conditions also imply that the state $\Phi^\U_{R_BB}$ can be written as a locally maximally mixed (LMM) state as follows
    \begin{align}
        \Phi^\U_{R_BB}=\dfrac{1}{4}\left(\mathbbm{1}_4+\sum_{i,j=1}^3R_{ij}\sigma_i\otimes \sigma_j\right).
    \end{align}
We know that a locally maximally mixed state on $\mathbbm{C}^4$ is equivalent to a Bell-diagonal state up to local unitary~\cite{Gam16,VDD01}. Consequently, we have 
\begin{align}
   \sigma_{R_BB}:=\U_{R_B}\otimes \U_B(\Phi^\U_{R_BB})=\dfrac{1}{4}(\mathbbm{1}_4+\sum_{i=1}^3S_{ii}\sigma_i\otimes \sigma_i).
\end{align}
Clearly, the square root $\sqrt{\sigma_{R_BB}}$ is also diagonal in the Bell basis. Since marginals of a Bell-diagonal operator are proportional to the Identity, we have 
\begin{align}
    \frac{1}{\tr\sqrt{\sigma_{R_BB}}}\tr_{R_B}\sqrt{\sigma_{R_BB}}&=\frac{1}{2}\mathbbm{1}_2,\\
\text{and},~~~   \frac{1}{\tr\sqrt{\sigma_{R_BB}}}\tr_{B}\sqrt{\sigma_{R_BB}}&=\frac{1}{2}\mathbbm{1}_2.
\end{align}
The Choi state condition for the state $\rho_{R_BB}=\frac{1}{\tr\sqrt{\Phi^\U_{R_BB}}}\sqrt{\Phi^\U_{R_BB}}$ follows since
\begin{align}
    \sqrt{\Phi^\U_{R_BB}}&=\U^\dagger_{R_B}\otimes \U^\dagger_B(\sqrt{\sigma_{R_BB}}),\\
    \text{which implies}~~~  \frac{1}{\tr\sqrt{\Phi^\U_{R_BB}}}\tr_{B}\sqrt{\Phi^\U_{R_BB}}&=\U^\dagger_{R_B}(\frac{1}{\tr\sqrt{\sigma_{R_BB}}}\tr_B\sqrt{\sigma_{R_BB}})\\
    &=\frac{1}{2}\mathbbm{1}_2.
\end{align}
Similarly, we can show that $\tr_{R_B}\rho_{R_BB}=\frac{1}{2}\mathbbm{1}_2$, proving the unitality of the qubit channel. This proves that the state $\rho_{R_BB}=\frac{1}{\tr\sqrt{\Phi^\U_{R_BB}}}\sqrt{\Phi^\U_{R_BB}}$ is a Choi state of a unital qubit map.
\end{proof}
Since a convex combination of LMM states is an LMM state, we have the following corollary of the above Lemma.
\begin{corollary}\label{corr:unitarycomb}
Given a channel $\n_{A'B'\to AB}$ that can be expressed as a convex combination of unitary channels $\mathcal{U}_{A'B'\to AB}$ such that $|B'|=|B|=2$, the state $\rho_{R_BB}:=\frac{1}{\tr\sqrt{\Phi^{\n}_{R_BB}}}\sqrt{\Phi^{\n}_{R_BB}}$ is the Choi state of a unital qubit channel.
\end{corollary}
\begin{lemma}\label{lemma:convex_purestate}~\cite[Lemma 10]{BGD25} Given a state with spectral decomposition $\rho=\sum_ip_i\op{\psi_i}{\psi_i}$, the max-relative entropy $ D_{\infty}(\rho\Vert\sigma)$ where $\sigma\ge 0$, is bounded as follows
    \begin{align}
        \log(\max_ip_i\bra{\psi_i}\sigma^{-1}\ket{\psi_i})\le D_{\infty}(\rho\Vert\sigma)\le\log(\sum_ip_i\bra{\psi_i}\sigma^{-1}\ket{\psi_i}).
    \end{align}
\end{lemma}
\begin{corollary}\label{corr:pure-optimal}
The max-relative entropy $D_{\infty}(\rho\Vert\sigma)$ for a fixed operator $\sigma\ge 0$ maximizes when the state $\rho$ is a pure state. That is,
    \begin{align}
        \max_\rho D_{\infty}(\rho\Vert\sigma)=\max_{\op{\psi}} D_{\infty}(\op{\psi}\Vert\sigma),
    \end{align}
    where $\op{\psi}$ are pure states.
\end{corollary}
\begin{proof}
Since pure states are a subset of the set of states, we naturally have $\max_\rho D_{\infty}(\rho\Vert\sigma)\ge\max_\psi D_{\infty}(\psi\Vert\sigma).$ To show the converse relation, let us have a Hermitian operator $\Lambda$ with spectral decomposition $
\Lambda = \sum_i \lambda_i \, |\lambda_i\rangle\langle \lambda_i| $, and $\lambda_{\max} := \max_i \lambda_i $.  Then, for every density operator $\rho$, $\operatorname{tr}(\Lambda \rho) \le \tr(\Lambda\op{\lambda_{\max}})$. As a consequence, since $\sigma^{-1}$ is a Hermitian operator, we have the following relation
\begin{align}
    \max_{\{p_i,\psi_i\}}\sum_ip_i\bra{\psi_i}\sigma^{-1}\ket{\psi_i}=\max_{\varphi}\bra{\varphi}\sigma^{-1}\ket{\varphi} 
\end{align}
where $\{p_i\}_i$ forms a probability distribution and $\{\psi_i\}_i$ forms an orthonormal basis. From \autoref{lemma:convex_purestate} we see that the upper and lower bounds for $D_{\infty}(\rho\Vert\sigma)$ coincide when $\rho$ is a pure state.
    \begin{align}
        \max_{\rho}D_{\infty}(\rho\Vert\sigma)
        &\le \max_{\{p_i,\psi_i\}}\log\sum_ip_i\bra{\psi_i}\sigma^{-1}\ket{\psi_i}\\
        &=\max_{\varphi}\log\bra{\varphi}\sigma^{-1}\ket{\varphi} \\
        &=\max_{\op{\psi}}D_{\infty}(\op{\psi}\Vert\sigma).
    \end{align}
\end{proof}
\section{The conditional Gibbs-preserving superchannels}\label{app:freesup}
Since the action of any superchannel $\Theta$ on a channel $\n$ can be written as a pre-processing and post-processing of $\n$, the conditional Gibbs-preserving superchannel $\Theta^\beta$ can be characterized by any superchannel where the post-processing is conditional Gibbs covariant as we see below,
\begin{align}
    \Theta^\beta(\mathcal{T}^\beta_{A'\to A}\otimes\mathcal{Q}_{B'\to B})&= \mathcal{T}^\beta_{C'\to C}\otimes\mathcal{Q}_{D'\to D}\\
    \text{or},~~~ \mathcal{A}_{AB\to CD}\circ \big(\mathcal{T}^\beta_{A'\to A}\otimes\mathcal{Q}_{B'\to B}\big)\circ \mathcal{B}_{C'D'\to A'B'}&= \mathcal{T}^\beta_{C'\to C}\otimes\mathcal{Q}_{D'\to D}\\
    \text{equivalently,}~~~ \mathcal{A}_{AB\to CD}\circ \big(\mathcal{T}^\beta_{C''\to A}\otimes\mathcal{Q}_{D''\to B}\big)&= \mathcal{T}^\beta_{C'\to C}\otimes\mathcal{Q}_{D'\to D}
\end{align}
such that $C''D''\cong C'D'$. Action of channel $\mathcal{A}_{AB\to CD}$ is therefore governed by the constraint
\begin{align}
    \mathcal{A}_{AB\to CD}(\gamma_A^\beta\otimes \rho_B)=\gamma^\beta_C\otimes\sigma_D.
\end{align}
for arbitrary states $\rho_B$ and $\sigma_D$. These are exactly the free operations in the resource theory of conditional athermality as defined in~\cite[Proposition 3]{JGW25}. Such channel are no-signaling from $A$ to $D$ and are conditionally covariant with respect to the absolutely thermal channel: $\mathcal{T}^\beta_{C\to C}\circ\mathcal{A}_{AB\to CD}=\mathcal{A}_{AB\to CD}\circ\mathcal{T}^\beta_{A\to A}$.
\end{appendix}
\bibliography{output}{}
\end{document}